\documentclass[prd,tightenlines,nofootinbib,superscriptaddress]{revtex4}

\usepackage{amsfonts,amssymb,amsthm,bbm,amsmath,wasysym}
\usepackage{hyperref}

\usepackage{subcaption}
\usepackage{color,psfrag}
\usepackage[dvips]{graphicx}

\usepackage{tikz}
\usetikzlibrary{calc}
\usetikzlibrary{decorations.pathmorphing}

\newcommand{\C}{{\mathbb C}}
\newcommand{\N}{{\mathbb N}}
\newcommand{\R}{{\mathbb R}}

\newcommand{\cI}{{\mathcal I}}

\newcommand{\cL}{{\mathcal L}}
\newcommand{\cH}{{\mathcal H}}

\newcommand{\cV}{{\mathcal V}}

\newcommand{\cC}{{\mathcal C}}
\newcommand{\cS}{{\mathcal S}}

\newcommand{\SU}{\mathrm{SU}}

\newcommand{\SL}{\mathrm{SL}}
\newcommand{\SO}{\mathrm{SO}}
\newcommand{\U}{\mathrm{U}}

\newcommand{\be}{\begin{equation}}
\newcommand{\ee}{\end{equation}}
\newcommand{\beq}{\begin{eqnarray}}
\newcommand{\eeq}{\end{eqnarray}}
\newcommand{\bes}{\begin{eqnarray}}
\newcommand{\ees}{\end{eqnarray}}

\newcommand{\mat} [2] {\left ( \begin{array}{#1}#2\end{array} \right ) }

\newcommand{\su}{{\mathfrak{su}}}

\newcommand{\la}{\langle}
\newcommand{\ra}{\rangle}

\newcommand{\tr}{{\mathrm{Tr}}}

\newcommand{\f}{\frac}

\def\nn{\nonumber}
\def\pp{\partial}

\def\vphi{\varphi}
\def\eps{\epsilon}

\newcommand{\id}{{\mathbb{I}}}

\def\vx{\vec{x}}
\def\vJ{\vec{J}}

\def\hN{\hat{N}}

\def\arr{\rightarrow}

\def\tDelta{\widetilde{\Delta}}
\def\tF{\widetilde{F}}
\def\tA{\widetilde{A}}
\def\cHl{{\cal H}^{\mathrm{loopy}}}
\def\cHs{{\cal H}^{\mathrm{sym}}}

\newtheorem{theorem}{Theorem}[section]
\newtheorem{lemma}[theorem]{Lemma}
\newtheorem{prop}[theorem]{Proposition}

\newtheorem{definition}[theorem]{Definition}

\def\restriction#1#2{\mathchoice
              {\setbox1\hbox{${\displaystyle #1}_{\scriptstyle #2}$}
              \restrictionaux{#1}{#2}}
              {\setbox1\hbox{${\textstyle #1}_{\scriptstyle #2}$}
              \restrictionaux{#1}{#2}}
              {\setbox1\hbox{${\scriptstyle #1}_{\scriptscriptstyle #2}$}
              \restrictionaux{#1}{#2}}
              {\setbox1\hbox{${\scriptscriptstyle #1}_{\scriptscriptstyle #2}$}
              \restrictionaux{#1}{#2}}}
\def\restrictionaux#1#2{{#1\,\smash{\vrule height .8\ht1 depth .85\dp1}}_{\,#2}} 

\newcommand*\circled[1]{\tikz[baseline=(char.base)]{
  \node[shape=circle,draw,inner sep=1pt] (char) {#1};}}

\def\tGamma{\widetilde{\Gamma}}
\def\tpsi{\widetilde{\psi}}
\def\tf{\widetilde{f}}
\def\tk{\tilde{k}}
\newcommand{\binomial} [2] {\left ( \begin{array}{c}#1 \\ #2\end{array} \right ) }
\def\hchi{\widehat{\chi}}
\def\vv{\vec{v}}
\def\dd{\mathrm{d}}
\def\tg{\tilde{g}}

\begin{document}

\title{The Fock Space of Loopy Spin Networks for Quantum Gravity}

\author{{\bf Christoph Charles}}\email{christoph.charles@ens-lyon.fr}
\affiliation{Laboratoire de Physique, ENS de Lyon, Universit\'e de Lyon, CNRS (UMR 5672), Lyon, France}

\author{{\bf Etera R. Livine}}\email{etera.livine@ens-lyon.fr}
\affiliation{Laboratoire de Physique, ENS de Lyon, Universit\'e de Lyon, CNRS (UMR 5672), Lyon, France}

\date{\today}

\begin{abstract}
In the context of the coarse-graining of loop quantum gravity, we
introduce loopy and tagged spin networks, which generalize the
standard spin network states to account explicitly for non-trivial
curvature and torsion.
Both structures relax the closure constraints imposed at the spin network vertices.
While tagged spin networks merely carry an extra spin at every vertex encoding the overall closure defect, loopy spin networks allow for an arbitrary number of loops attached to each vertex. 
These little loops can be interpreted as local excitations of the quantum gravitational field and we discuss the  statistics  to endow them with.
The resulting Fock space of loopy spin networks realizes  new truncation of loop quantum gravity, allowing to formulate its graph-changing dynamics on a fixed background graph plus local degrees of freedom attached to the graph nodes. This provides a framework for re-introducing a non-trivial background quantum geometry around which we would study the effective dynamics of perturbations.
We study how to implement the dynamics of topological BF theory in this framework. We realize the projection on flat connections through holonomy constraints and we pay special attention to their often overlooked non-trivial flat solutions defined by higher derivatives of the $\delta$-distribution. 

\end{abstract}

\maketitle

\tableofcontents

\section{Introduction}

Loop quantum gravity (for lecture books, see \cite{rovelli2007quantum,thiemann2007modern,Gambini:2011zz}) proposes a non-perturbative and background independent framework for quantum gravity. Based on a 3+1 splitting of space-time distinguishing time from the 3d space, it realizes a canonical quantization of the general relativity reformulated as a gauge field theory. The canonically conjugate fields are the triad, defining the local 3d frame, and the Ashtekar-Barbero $\SU(2)$ connection \cite{PhysRevLett.57.2244,Barbero:1994ap,Immirzi:1996di}. Quantum states of geometry, called spin network states, define the excitations of those fields over the kinematical Ashtekar-Lewandowski vacuum of vanishing triad and connection (corresponding to the ``nothing''-state of a degenerate vanishing metric). Their dynamics is implemented through the Hamiltonian constraints, ensuring the invariance of the theory under space-time diffeomorphisms.
Several explicit proposals for the quantum dynamics exist, either in the pure canonical formalism from the original Thiemann definition of the Hamiltonian constraint operator \cite{Thiemann:1996ay,Thiemann:1996aw,Thiemann:1996av} to more recent constructions \cite{Thiemann:2003zv,Alesci:2011aj,Alesci:2015wla,Assanioussi:2015gka} (for a review see \cite{Bonzom:2011jv}), or in a covariant path integral approach with transition amplitudes defined from spinfoam models as in the EPRL model  \cite{Engle:2007wy,Geloun:2010vj,rovelli2014quantum} and variations \cite{Dupuis:2011fz,Speziale:2012nu,Wieland:2013cr} (for a review of the spinfoam framework see \cite{Livine:2010zx,Perez:2012wv,Bianchi:2012nk}).

The main challenges in this context are the correct definition of the quantum dynamics and the coarse-graining of the theory. These two issues are intimately intertwined in that the proper quantum gravity dynamics should address and solve the perturbative non-renormalisability of general relativity, and more generally because a consistent theory of quantum gravity should give us the effective dynamics of the geometry at all scale of length and energy and provide us with a flow from the probably discrete and quantum dynamics at the Planck scale to the classical dynamics of a classical space-time manifold prescribed by general relativity.
The main difficulty in realizing this program in the loop quantum gravity  framework is to proceed to a coarse-graining of the gravitational degrees of freedom in a background independent theory with no a priori length or energy scale and no a priori regular geometrical background on which to define a coarse-graining procedure.
This translates into the problem of defining new vacuum states representing non-degenerate metrics and geometries (such as the flat Minkowski space-time) and working out how to expand the quantum gravity theory, initially defined above the ``nothing''-state,   around them with an explicit dictionary between the fundamental geometry excited states and the new effective gravity excitations.

Some progress in this direction has been achieved by Koslowski and Sahlmann in \cite{Koslowski:2011vn}, where they define new vacuum states for loop quantum gravity peaked on some classical field configuration for the triad and connection and describe the spin network excitations above them. Another approach by Dittrich and Geiller was to define another vacuum state, given as the flat connection vacuum of a topological BF theory, and to work out a Hilbert space better suited to account for curvature defects  \cite{Dittrich:2014wpa,Bahr:2015bra}. Another line of research in the spinfoam framework is the promising  renormalisation program for tensor models and group field theories  \cite{Oriti:2013aqa,Rivasseau:2011hm,Rivasseau:2013uca,Carrozza:2013mna,Carrozza:2014rya}, but this deviates from the canonical point of view that we will pursue in the present paper. Indeed, we would like to propose another path to define the effective loop quantum gravity dynamics on non-trivial background quantum geometries following the previous work on the coarse-graining of spin network states \cite{Livine:2006xk,Livine:2013gna}.

\medskip

In loop quantum gravity, quantum states of geometry are wave-functions of the Ashtekar-Barbero connection. More precisely, they are defined as cylindrical functionals of that connection, in that they realize a finite sampling of the connection. Indeed each wave-function is constructed with respect to a graph $\Gamma$ (embedded in the canonical hypersurface) and depends on the holonomies of the connection along the graph edges (defined as $\SU(2)$ group elements). The set of wave-functions is obtained by taking the union over all embedded graphs with the requirement of cylindrical consistency\footnotemark, that is a wave-function initially defined over a given graph $\Gamma$ is considered as equivalently defined over any refinement of that graph (i.e. any graph containing $\Gamma$ as a subgraph). This union quotiented by the cylindrical consistency is rigorously defined as a projective limit \cite{Ashtekar:1994mh,Ashtekar:1994wa}. It leads to the Hilbert space of quantum states of geometry of loop quantum gravity and has been shown to be the $L^{2}$ space of functionals of the connection with respect to the Ashtekar-Lewandowski measure \cite{Ashtekar:1993wf}.
\footnotetext{The definition of cylindrical consistency can be generalized and extended to a mapping or identification between two wave-functions living on a given graph and a refinement of it. Formulated as such, it becomes equivalent to the choice of a coarse-graining procedure for the quantum states of geometry. This logic has been used to construct new Hilbert spaces for loop quantum states describing the excited states of geometry above non-trivial vacua \cite{Koslowski:2011vn}.
}

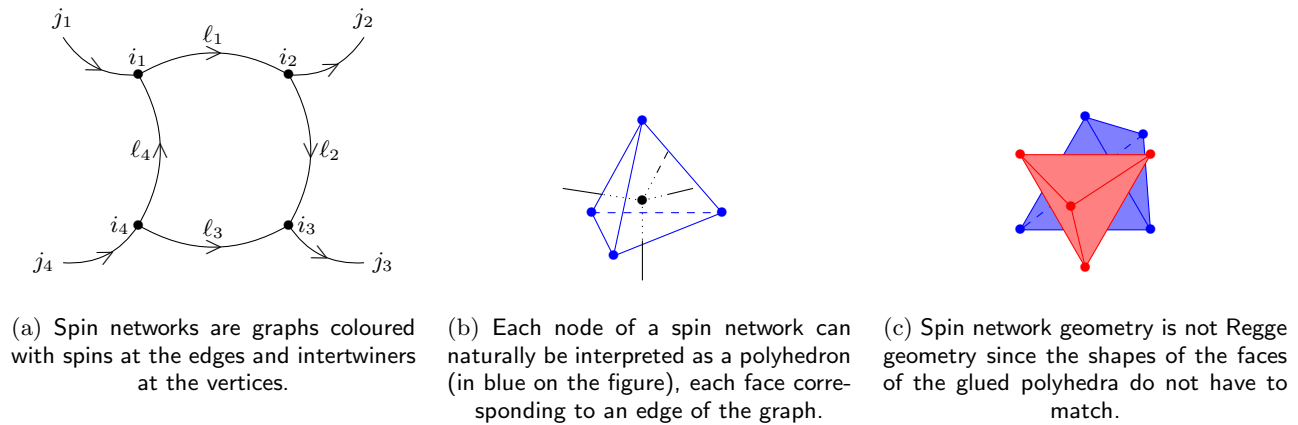
\begin{figure}


\begin{subfigure}[t]{.3\linewidth}
\begin{tikzpicture}[scale=1]
\coordinate(A) at (0,0);
\coordinate(B) at (2,0);
\coordinate(C) at (2,-2);
\coordinate(D) at (0,-2);

\draw (A) to[bend left] node[midway,sloped]{$>$} node[midway,above]{$\ell_1$} (B);
\draw (B) to[bend left] node[midway,sloped]{$>$} node[midway,right]{$\ell_2$} (C);
\draw (C) to[bend left] node[midway,sloped]{$>$} node[midway,above]{$\ell_3$} (D);
\draw (D) to[bend right] node[midway,sloped]{$>$} node[midway,left]{$\ell_4$} (A);

\draw (A) to[bend left] node[midway,sloped]{$>$} ++(-1,0.5) node[above]{$j_1$};
\draw (B) to[bend right] node[midway,sloped]{$>$} ++(1,0.5) node[above]{$j_2$};
\draw (C) to[bend right] node[midway,sloped]{$>$} ++(1,-0.5) node[right]{$j_3$};
\draw (D) to[bend left] node[midway,sloped]{$>$} ++(-1,-0.5) node[left]{$j_4$};

\draw (A) node {$\bullet$} node[above]{$i_1$};
\draw (B) node {$\bullet$} node[above]{$i_2$};
\draw (C) node {$\bullet$} node[right]{$i_3$};
\draw (D) node {$\bullet$} node[left]{$i_4$};
\end{tikzpicture}
\caption{Spin networks are graphs coloured with spins at the edges and intertwiners at the vertices.}
\label{fig:spinnetwork_a}
\end{subfigure}
\hspace{2mm}
\begin{subfigure}[t]{.3\linewidth}
\begin{tikzpicture}[scale=1]
\coordinate (O) at (0,0,0);

\coordinate (A) at (0,1.061,0);
\coordinate (B) at (0,-0.354,1);
\coordinate (C) at (-0.866,-0.354,-0.5);
\coordinate (D) at (0.866,-0.354,-0.5);

\draw[blue] (A) -- (B);
\draw[blue] (A) -- (C);
\draw[blue] (A) -- (D);
\draw[blue] (B) -- (C);
\draw[dashed,blue] (C) -- (D);
\draw[blue] (D) -- (B);

\draw[dotted] (O) -- ++(0,-0.531,0);
\draw (0,-0.531,0) -- ++(0,-0.531,0);

\draw[dotted] (O) -- ++(0,0.177,-0.5);
\draw[dashed] (0,0.177,-0.5) -- ++(0,0.177,-0.5);

\draw[dotted] (O) -- ++(0.433,0.177,0.25);
\draw (0.433,0.177,0.25) -- ++(0.433,0.177,0.25);

\draw[dotted] (O) -- ++(-0.433,0.177,0.25);
\draw (-0.433,0.177,0.25) -- ++(-0.433,0.177,0.25);

\draw (O) node{$\bullet$};
\draw[blue] (A) node{$\bullet$};
\draw[blue] (B) node{$\bullet$};
\draw[blue] (C) node{$\bullet$};
\draw[blue] (D) node{$\bullet$};
\end{tikzpicture}

\caption{Each node of a spin network can naturally be interpreted as a polyhedron (in blue on the figure), each face corresponding to an edge of the graph.}
\label{fig:spinnetwork_b}
\end{subfigure}
\hspace{2mm}
\begin{subfigure}[t]{.3\linewidth}
\begin{tikzpicture}[scale=1]
\coordinate (A1) at (0,0,-2);
\coordinate (B1) at (0,1,0);
\coordinate (C1) at (-0.866,-0.5,0);
\coordinate (D1) at (0.866,-0.5,0);

\fill[blue!50] (B1) -- (C1) -- (D1) -- cycle;
\fill[blue!50] (B1) -- (D1) -- (A1) -- cycle;

\draw[blue] (A1) -- (B1);
\draw[blue,dashed] (A1) -- (C1);
\draw[blue] (A1) -- (D1);
\draw[blue] (B1) -- (C1);
\draw[blue] (C1) -- (D1);
\draw[blue] (D1) -- (B1);

\draw[blue] (A1) node{$\bullet$};
\draw[blue] (B1) node{$\bullet$};
\draw[blue] (C1) node{$\bullet$};
\draw[blue] (D1) node{$\bullet$};

\coordinate (A2) at (0,0,0.5);
\coordinate (B2) at (0,-1,0);
\coordinate (C2) at (-0.866,0.5,0);
\coordinate (D2) at (0.866,0.5,0);

\fill[red!50] (A2) -- (B2) -- (D2) -- cycle;
\fill[red!50] (A2) -- (C2) -- (D2) -- cycle;
\fill[red!50] (A2) -- (B2) -- (C2) -- cycle;

\draw[red] (A2) -- (B2);
\draw[red] (A2) -- (C2);
\draw[red] (A2) -- (D2);
\draw[red] (B2) -- (C2);
\draw[red] (C2) -- (D2);
\draw[red] (D2) -- (B2);

\draw[red] (A2) node{$\bullet$};
\draw[red] (B2) node{$\bullet$};
\draw[red] (C2) node{$\bullet$};
\draw[red] (D2) node{$\bullet$};
\end{tikzpicture}

\caption{Spin network geometry is not Regge geometry since the shapes of the faces of the glued polyhedra do not have to match.}
\label{fig:spinnetwork_c}
\end{subfigure}

\caption{The geometrical interpretation of spin networks}
\label{fig:spinnetwork}
\end{figure}

Spin networks provide basis states of this Hilbert space. They are introduced as diagonalizing the area and volume operators, both shown to have discrete spectra \cite{Rovelli:1994ge,Rovelli:1995ac}. A spin network state, as drawn on fig.\ref{fig:spinnetwork_a}, lives on a given graph $\Gamma$ (and therefore lives on any refinement of that graph by cylindrical consistency) and is defined by spins (as half-integers determining an irreducible representation of the Lie group $\SU(2)$) on each edge and intertwiners (as singlet states) on each vertex or node of the graph. Spins define quanta of area while intertwiners give the quanta of volume. This hints towards a geometrical interpretation of spin networks as discrete geometries. This can be realized explicitly and  spin networks have been shown to be the quantization of twisted geometries \cite{Freidel:2010aq,Dupuis:2012yw}, which generalize 3d Regge geometries to account for some torsion.
Nodes of the graph represent polyhedra, which are glued along faces dual to the graph edges, as illustrated on fig.\ref{fig:spinnetwork_b}.
The precise face matching of Regge geometries is relaxed to a simpler area face matching, and the resulting potential shape mismatch is interpreted as torsion along the graph edge, as depicted on fig.\ref{fig:spinnetwork_c}. 

The Hamiltonian operators of loop quantum gravity then act on both algebraic and combinatorial data, meaning that they modify both  spins and intertwiners living on a given graph and the graph itself. Thus the loop quantum gravity (LQG) dynamics seems to be a careful balance between fixed-graph dynamics and graph-changing dynamics, which reflects the classical dynamics of general relativity as both dynamics on a given space-time manifold and of the space-time geometry itself. This mixture between these two types of dynamics render analytical and numerical calculations extremely difficult.
%
%
The usual strategy for discrete systems on fixed graphs, as in condensed matter theory, is to coarse-grain the theory, that is, to integrate the microscopic degrees of freedom inside bounded regions, thus assimilated to points, and to write effective theories for the relevant macroscopic degrees of freedom. This process of coarse-graining ultimately leads to the continuum limit of the theory.
One can also study the statistical physics of a varying graph, for instance using matrix models for 2d quantum gravity. Putting these two ingredients together, for example to study matter coupled to 2d quantum gravity through condensed matter models on random lattices, is much more involved. Results in this direction, like the KPZ conjecture  \cite{Knizhnik:1988ak}, mostly relie on conformal field theory techniques in the continuum limit.

\medskip

In the loop quantum gravity context,  through the logic of coarse-graining, we would like to map the varying graphs dynamics onto a fixed graph dynamics.
The rational behind this is the following: starting from a base graph, each node will correspond to a varying coarse-grained region. This means that we will add some extra internal degrees of freedom to the graph vertices in the effective theory, which should reflect that each vertex represent in fact a possibly varying subgraph itself.  This method should mimic a development around this base graph considered as a skeleton graph for the gravity excitations. 
Digging deeper in the structure of the spin networks, the wave-functions carry non-trivial curvature around the loops of the graph (when the composition of the holonomies along the edges of the loops does not yield the identity). 
When coarse-graining, curvature should build up and the effective coarse-grained vertices should naturally describe locally curved geometries. 
As illustrated in fig.\ref{fig:newstructure}, we will need additional data at each vertex to encode the coarse-grained curvature now located at the vertices. 
We will realize this program using the coarse-graining through gauge fixing procedure already developed in  \cite{Freidel:2002xb,Livine:2013gna}. The idea is that the gauge fixing procedure collapses a subgraph into a single vertex plus some additional (self-)loops connecting that vertex to itself. These self-loops carry the curvature initially carried by the loops of the subgraph. In this picture, coarse-graining a spin network state thus leads to a state living on a coarser graph with many little loops living at each vertex.

Reversing the logic of this procedure, we propose to fix a background graph and define the Fock space of ``loopy spin networks'' above that graph. The base states will be spin network states living on the background graph itself, while excitations will be spin networks living on the graph plus an arbitrary number of little loops attached to its vertices. These little loops allow to represent excitations of the curvature and take into account that each vertex of the background graph is in fact a coarse structure which could be unfolded into a non-trivial, possibly complex, subgraph. Taking into account these little loops as localized excitations of the quantum geometry allow to project the graph changing LQG dynamics onto a fixed background graph but with dynamical curvature excitations living at each vertex.
In some sense, the underlying full graph is still dynamical and changes, but we always coarse-grain it to the same skeleton graph plus some little loops. The dynamics then affect the little loops without touching the skeleton graph and of course change the algebraic  data -spin and intertwiners- carried by the edges and vertices of the graph.
This explicitly realizes  the idea proposed in the conclusion of \cite{Livine:2013gna}.

So we are here proposing an expansion of the theory around an arbitrary non-trivial background graph, in some sense truncating the dynamics by a coarse-graining procedure to keep the quantum states living on that chosen background graph with some localized excitations of the geometry. This is rather different from the expansion around a continuous background metric, proposed up to now in loop quantum gravity as in \cite{Koslowski:2011vn}, and it will be necessary to later compare these two approaches  in a continuum limit of spin networks.

\begin{figure}
\centering

\begin{tikzpicture}[scale=0.85]
\coordinate(A) at (0,0);
\coordinate(B) at (2,0);
\coordinate(C) at (2,-2);
\coordinate(D) at (0,-2);

\draw (A) -- (B) node[midway,above] {$\ell_1$};
\draw (B) -- (C) node[midway,right] {$\ell_2$};
\draw (C) -- (D) node[midway,below] {$\ell_3$};
\draw (D) -- (A) node[midway,left] {$\ell_4$};
\draw (B) -- (D) node[midway,sloped,above] {$\ell_5$};

\draw (A) -- ++(-1,0.5) node[above] {$j_1$};
\draw (B) -- ++(1,0.5) node[above] {$j_2$};
\draw (C) -- ++(1,-0.5) node[below] {$j_3$};
\draw (D) -- ++(-1,-0.5) node[below] {$j_4$};

\draw (A) node {$\bullet$} node[above]{$i_1$};
\draw (B) node {$\bullet$} node[above]{$i_2$};
\draw (C) node {$\bullet$} node[below]{$i_3$};
\draw (D) node {$\bullet$} node[below]{$i_4$};

\draw[gray,dashed] (1,-1) circle(2);

\draw[->,>=stealth,very thick] (3.5,-1) -- (5.5,-1);

\coordinate(O) at (8,-1);
\draw (O) -- ++(-2,1) node[above] {$j_1$};
\draw (O) -- ++(2,1) node[above] {$j_2$};
\draw (O) -- ++(2,-1) node[below] {$j_3$};
\draw (O) -- ++(-2,-1) node[below] {$j_4$};

\draw[fill=lightgray] (O) circle(0.5);
\draw[scale=2] (O) node{\textbf{?}};

\end{tikzpicture}

\caption{When coarse-graining, the curvature carried by the loops in the collapsed bounded region leads to curvature. Coarse-grained vertices thus need a new structure to be described.}
\label{fig:newstructure}
\end{figure}
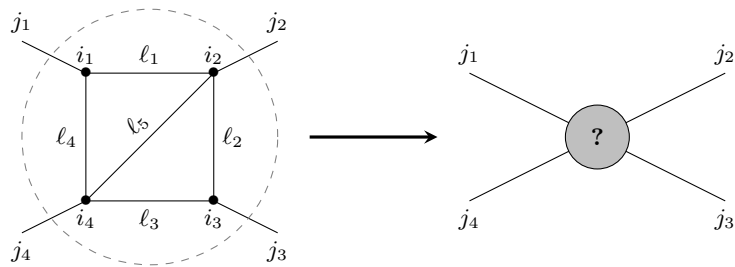

\medskip

This paper is organized as follows.
The first section will review the Hilbert space of spin networks for loop quantum gravity, underlining the cylindrical consistency requirement, and the ``coarse-graining through gauge-fixing'' of these quantum states of geometry. From this perspective, we will introduce a hierarchy of possible extensions of spin networks encoding extra information at the graph nodes: folded, loopy and tagged spin networks, from the finer to coarser objects.
Folded spin networks allow for an arbitrary number of little (self-)loops at every node of the graph and moreover contain the data of a circuit at each node, that is a tree linking  the little loop ends to the edges attached to the node. Loopy spin networks forget about the local circuit data, while tagged spin networks simply trace out all the little loop data and only retains the resulting closure defect at each node.

The second section is dedicated to the definition and investigation of loopy spin networks. We discuss the holonomy operator, which is the elementary brick of the loop quantum gravity formalism, and check the compatibility of our definition with the cylindrical consistency conditions. We apply this construction to the topological BF theory and solve for the physical state of flat connections in our new Hilbert of loopy spin networks.
We define the BF Hamiltonian constraints as holonomy constraints peaking the group elements along the little loops on the identity. We show that these are (unexpectedly) not enough to enforce the uniqueness of the physical state and lead to an infinite-dimensional space of (almost)-flat states defined from higher derivatives of the $\delta$-distribution. This is due to the highly non-trivial structure of the intertwiner space recoupling the loops. We supplement these constraints with new Laplacian constraints, which decouple the loop and trivialize the intertwiner, finally leading to a unique flat state defined by the $\delta$-distribution.

In a third section, we explore the possibility of endowing the little loops at the graph's nodes  with bosonic statistics and define the symmetrized Fock space of loopy spin networks. We discuss the interplay between loop creation, annihilation and spin shift in the definition of the holonomy operator. This leads us to define the Hamiltonian constraints for BF theory in terms of creation and annihilation operators.

The fourth section develops tagged spin networks and shows how they provide a basis for reduced density matrix when  coarse-graining spin networks.
We conclude this paper with a discussion on the potential applications of this new framework, for instance using fixed skeleton graph as background lattices or to the coarse-graining of loop quantum gravity dynamics.

\section{Spin networks and their coarse-graining}
\label{sec:overview}

\subsection{Spin networks as projective limits}
\label{ProjectiveLimits}

Loop quantum gravity is based on the first order reformulation of general relativity in terms of the Ashtekar-Barbero variables\cite{PhysRevLett.57.2244}. The fundamental variables of the theory on the canonical hypersurface  are the densitized triad and the Ashtekar-Barbero connection\footnotemark, which are endowed with the following symplectic structure:
\begin{equation}
\{E^a_i(x),A^j_b(y)\} = \frac{\beta \kappa}{2} \delta^a_b \delta^j_i \delta(x-y)\,,
\end{equation}
with all other brackets being zero.
The indices $a,b,c,..$  denote space coordinates while the indices $i,j,k,...$ refer to tangent space coordinates.
The canonical fields are the densitized triad $E^a_i(x)$ and  the Ashtekar-Barbero $\SU(2)$ connection $A^j_b(x)$. The Poisson backet coupling is given in terms of the gravitational constant $\kappa = 16\pi G$ and  the Immirzi parameter $\beta$ \cite{Barbero:1994ap,Immirzi:1996di,Holst:1995pc}. 
\footnotetext{The Ashtekar-Barbero connection is only a space connection defined on the canonical hypersurface and is not generically the pull-back of a space-time connection \cite{Samuel:2000ue,Alexandrov:2001wt,Geiller:2011cv,Geiller:2011bh,Charles:2015rda}, except in the case of the self-dual and anti-self dual connections given by the purely imaginary choice of Immirzi paramater $\beta=\pm i$.} 
The classical theory is defined by imposing on this phase space a set of seven first class constraints: the three Gauss constraints generating the local $\SU(2)$ gauge invariance and the four constraints generating space-time diffeomorphisms. 
%
%

\begin{figure}[t!]
\centering
\begin{tikzpicture}
\coordinate(A) at (0,0);
\coordinate(B) at (2,0);
\coordinate(C) at (2,-2);
\coordinate(D) at (0,-2);

\draw (A) to[bend left] node[midway,sloped]{$>$} node[midway,above]{$g_e$} node[midway,below]{e} (B);
\draw[lightgray] (B) to[bend left] node[midway,sloped]{$>$} (C);
\draw[lightgray] (C) to[bend left] node[midway,sloped]{$>$} (D);
\draw[lightgray] (D) to[bend right] node[midway,sloped]{$>$} (A);

\draw[lightgray] (A) to[bend left] node[midway,sloped]{$>$} ++(-1,0.5);
\draw[lightgray] (B) to[bend right] node[midway,sloped]{$>$} ++(1,0.5);
\draw[lightgray] (C) to[bend right] node[midway,sloped]{$>$} ++(1,-0.5);
\draw[lightgray] (D) to[bend left] node[midway,sloped]{$>$} ++(-1,-0.5);

\draw (A) node {$\bullet$} node[above]{$s(e)$};
\draw (B) node {$\bullet$} node[above]{$t(e)$};
\draw[lightgray] (C) node {$\bullet$};
\draw[lightgray] (D) node {$\bullet$};
\end{tikzpicture}
\caption{We consider cylindrical wavefunctions: the function only depends on a support graph with oriented edges. Each edge $e$ carries a colouring $g_e$ which is a group element and corresponds to the (open) holonomy along the corresponding path.}
\label{fig:graph}
\end{figure}
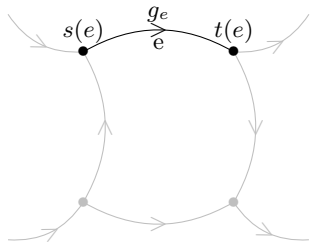
At the quantum level, we consider cylindrical wave-functions of the Ashtekar-Barbero connection $A$. We choose an arbitrary oriented graph $\Gamma$ embedded in the canonical hypersurface and consider functions of the holonomies of the connection along the links or edges $e$ of the graph, as illustrated in fig.\ref{fig:graph}:
\be
\psi_{\Gamma}[A]
\,\equiv\,
\psi\left(
\{U_{e}[A]\}_{e\in\Gamma}
\right)\,,
\qquad
U_{e}[A]\in\SU(2)\,.
\ee 
Such functionals realize a finite sampling of the connection along the considered graph. We require these functionals to solve the Gauss law, that is to be  invariant under local $\SU(2)$ gauge transformations. These acts at the end points of the holonomies, that is at the nodes or vertices $v$ of the graph $\Gamma$:
\be
\label{gaugeinv}
\forall h_{v}\in\SU(2)^{\times V}\,,
\quad
\psi\left(
\{U_{e}[A]\}_{e\in\Gamma}
\right)
\,=\,
\psi\left(
\{h_{s(e)}^{-1}\,U_{e}[A]\,h_{t(e)}\}_{e\in\Gamma}
\right)\,,
\ee
where $V$ is the number of vertices of the graph $\Gamma$, while $s(e)$ and $t(e)$ respectively denote the source and target vertices of the oriented edge $e$. The Hilbert space of states of the fixed graph $\Gamma$ is defined by endowing this set of wave-functions with the natural scalar product induced by the Haar measure on $\SU(2)$:
\be
\cH_{\Gamma}\,\equiv\,L^{2}\left(\SU(2)^{E}/\SU(2)^{V}\right)\,,\nn
\ee
\be
\forall \Psi\,,\widetilde{\Psi}\,\in\,\cH_{\Gamma}\,,\quad
\langle \Psi | \widetilde{\Psi} \rangle
\,=\,
\int_{\SU(2)^{E}} \prod_{e=1}^{E} dg_e\,\,
\overline{\Psi(g_1,...,g_{E})} \widetilde{\Psi}(g_1,...,g_{E})\,,
\ee
where $E$ counts the number of edges in the graph. A basis of this space is provided by the spin networks with support on the graph $\Gamma$. Technically, these are obtained through the Peter-Weyl decomposition of $L^{2}$ functions on the Lie group $\SU(2)$ in terms of the orthogonal Wigner matrices in the irreducible representations of $\SU(2)$. As a result, a spin network state is labeled by a spin on each edge, which is a half-integer $j_{e}\in\N/2$ determining the corresponding irreducible $\SU(2)$-representation $\cV^{j_{e}}$ of dimension $(2j_{e}+1)$,  and an intertwiner $i_{v}$ at each vertex $v$, which is an invariant tensor in the tensor product of the representations living on the incoming and outgoing edges attached to the vertex $v$:
\be
i_{v}:\bigotimes_{e|s(e)=v}\cV^{j_{e}}\longrightarrow\bigotimes_{e|t(e)=v}\cV^{j_{e}}\,.
\ee
The spin network function is defined by contracting the chosen intertwiners with the Wigner matrices of the holonomies living along the graph edges:
\beq
\psi_{\Gamma}^{\{j_{e,i_{v}}\}}
\Big{(}\{g_{e}\}_{e\in\Gamma}\Big{)}
&=&
\tr\,\left[\bigotimes_{e} D^{j_{e}}(g_{e}) \otimes \bigotimes_{v} i_{v}\right]
\nn\\
&=&
\prod_{e}\la j_{e}m_{e}^{s}|\,g_{e}\,|j_{e}m_{e}^{t}\ra\,
\prod_{v}\la \otimes_{e|t(e)=v}j_{e}m_{e}^{t}|\,i_{v\,}|\otimes_{e|s(e)=v}j_{e}m_{e}^{s}\ra\,,
\eeq
with an implicit sum over all the $m_{e}^{s,t}$ indices,  where we have introduced the usual spin basis $|j,m\ra$ of the Hilbert space $\cV^{j}$ with the index $m$ running from $-j$ to $+j$ by integer steps.
That spin network functional is automatically gauge-invariant due to the $\SU(2)$-invariance of the intertwiners at each vertex. Intertwiners $i_{v}$ at the vertex give the volume excitations, thus representing chunks of volume dual to each vertex, while the spin $j_{e}$ living on the edge $e$ linking two vertices give the area quanta of the (quantum) surface boundary between the corresponding two chunks of space.
This endows spin networks with a natural interpretation as discrete quantum geometries.

From here, the loop quantum gravity programme proceeds in two steps. First, one sums over all possible graphs $\Gamma$ imposing cylindrical consistency. This yields the kinematical Hilbert space of spin network states. Second, one imposes the Hamiltonian constraints generating the space-time diffeomorphisms at the quantum level to define the physical Hilbert space of loop quantum gravity.
%

\medskip

Indeed, in order to consider the full space of connections and not just its finite sampling on a fixed graph $\Gamma$, we will consider all possible graphs and sums of cylindrical functions over different graphs. To this purpose, one needs to compare wave-functions with support on different graphs, take their sum and scalar product. This is achieved through requiring cylindrical consistency. A function $\psi$ on a graph $\Gamma$ is considered as equivalent to another function $\tpsi$ defined on a larger graph $\tGamma$, containing $\Gamma$ as a subgraph, if the finer function does not depend on the group elements living on the extra edges and coincides with the original coarser function $\psi$ on the subgraph:
\be
\Gamma\subset\tGamma,\qquad
\psi_{\Gamma}\,\sim\,\tpsi_{\tGamma}
\quad\Leftrightarrow\quad
\forall g_{e}\in\SU(2)\,\quad \tpsi(\{g_{e}\}_{e\in\tGamma})=\psi(\{g_{e}\}_{e\in\Gamma})\,.
\ee
This means that any wave-function defined on a graph $\Gamma$ is automatically extended to live on any refinement of $\Gamma$. Now, when summing two wave-functions living on a priori different graphs or taking their scalar product, one will refine the two graphs, say $\Gamma$ and $\Gamma'$, to a larger and finer graph $\tGamma$ containing both of them as subgraphs: the two wave-functions will then be compared on that larger graph $\tGamma$ and their sum will be be defined as living on it.
A precise and rigorous treatment of these projective limit techniques can be found in \cite{Ashtekar:1993wf,Ashtekar:1994mh,Ashtekar:1994wa}. The set of wave-functions is defined by the union of the sets of wave-functions with support on every graph quotiented by the cylindrical consistency equivalence relation, and the resulting kinematical Hilbert space for loop quantum gravity is obtained by the sum of all graphs also quotiented by the cylindrical consistency:
\be
\mathcal{F}= \left.\left(\bigcup_\Gamma \mathcal{F}_\Gamma\right)\right/\sim\,,
\qquad
\mathcal{H}_\textrm{kin} = \left.\left(\bigoplus_\Gamma \mathcal{H}_\Gamma\right)\right/\sim\,.
\ee
This Hilbert space, defined as a projective limit, was shown to be the space of $L^{2}$-functionals of the connection with respect to the Ashtekar-Lewandowski measure \cite{Ashtekar:1994wa}.
Going to the spin network basis, the cylindrical consistency corresponds to identifying graphs with edges $e$ carrying a vanishing spin $j_{e}=0$ to the same graphs without those edges. Thus, if we went to pick a specific representative for every equivalence class, we could simply choose all spin networks carrying no vanishing spins:
\be
\mathcal{H}_\textrm{kin} = \bigoplus_\Gamma \widetilde{\mathcal{H}}_\Gamma\,,
\qquad
\widetilde{\mathcal{H}}_\Gamma = \bigoplus_{\{j_e \neq 0,i_v\}} \mathbb{C}|j_e,i_v\rangle
\ee

\begin{figure}[t!]
\centering

\begin{tikzpicture}[scale=0.5]
\coordinate(A) at (0,0);
\coordinate(B) at (2,0);
\coordinate(C) at (4,0);
\coordinate(D) at (0,-2);
\coordinate(E) at (2,-2);
\coordinate(F) at (4,-2);

\draw[decorate, decoration=snake] (A) to[bend left] (B);
\draw[decorate, decoration=snake] (B) to[bend right] (C);
\draw[decorate, decoration=snake] (D) to[bend left] (E);
\draw[decorate, decoration=snake] (E) to[bend right] (F);

\draw[decorate, decoration=snake] (A) -- (D);
\draw[decorate, decoration=snake] (B) -- (E);
\draw[decorate, decoration=snake] (C) -- (F);

\draw (A) node {$\bullet$};
\draw (B) node {$\bullet$};
\draw (C) node {$\bullet$};
\draw (D) node {$\bullet$};
\draw (E) node {$\bullet$};
\draw (F) node {$\bullet$};

\coordinate(O1) at (-4,6);
\coordinate(A1) at ($(O1)+(0,0)$);
\coordinate(B1) at ($(O1)+(2,0)$);
\coordinate(C1) at ($(O1)+(4,0)$);
\coordinate(D1) at ($(O1)+(0,-2)$);
\coordinate(E1) at ($(O1)+(2,-2)$);
\coordinate(F1) at ($(O1)+(4,-2)$);

\draw[<-,>=stealth,very thick] (1,1) -- (-1,3);

\draw[decorate, decoration=snake] (A1) to[bend left] (B1);
\draw[gray,dashed,decorate, decoration=snake] (B1) to[bend right] (C1);
\draw[decorate, decoration=snake] (D1) to[bend left] (E1);
\draw[gray,dashed,decorate, decoration=snake] (E1) to[bend right] (F1);

\draw[decorate, decoration=snake] (A1) -- (D1);
\draw[decorate, decoration=snake] (B1) -- (E1);
\draw[gray,dashed,decorate, decoration=snake] (C1) -- (F1);

\draw (A1) node {$\bullet$};
\draw (B1) node {$\bullet$};
\draw[gray] (C1) node {$\bullet$};
\draw (D1) node {$\bullet$};
\draw (E1) node {$\bullet$};
\draw[gray] (F1) node {$\bullet$};

\coordinate(O2) at (4,6);
\coordinate(A2) at ($(O2)+(0,0)$);
\coordinate(B2) at ($(O2)+(2,0)$);
\coordinate(C2) at ($(O2)+(4,0)$);
\coordinate(D2) at ($(O2)+(0,-2)$);
\coordinate(E2) at ($(O2)+(2,-2)$);
\coordinate(F2) at ($(O2)+(4,-2)$);

\draw[<-,>=stealth,very thick] (3,1) -- (5,3);

\draw[gray,dashed,decorate, decoration=snake] (A2) to[bend left] (B2);
\draw[decorate, decoration=snake] (B2) to[bend right] (C2);
\draw[gray,dashed,decorate, decoration=snake] (D2) to[bend left] (E2);
\draw[decorate, decoration=snake] (E2) to[bend right] (F2);

\draw[gray,dashed,decorate, decoration=snake] (A2) -- (D2);
\draw[decorate, decoration=snake] (B2) -- (E2);
\draw[decorate, decoration=snake] (C2) -- (F2);

\draw[gray] (A2) node {$\bullet$};
\draw (B2) node {$\bullet$};
\draw (C2) node {$\bullet$};
\draw[gray] (D2) node {$\bullet$};
\draw (E2) node {$\bullet$};
\draw (F2) node {$\bullet$};
\end{tikzpicture}

\caption{To define the scalar product in the continuum, we use cylindrical consistency: wavefunctions with trivial dependancy on some edges are identified with functions on coarser graphs (with the gray dashed edges removed from the graph). As a consequence, two coarse graphs can always be considered as being embedded in another finer graph on which the scalar product is well-defined.}
\label{fig:cylindrical}
\end{figure}
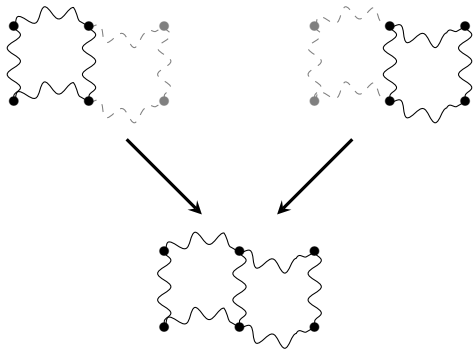

\medskip

This projective limit technique was introduced for graphs embedded in the canonical hypersurface, but was also shown to work for equivalence classes of graphs under (spatial) diffeomorphisms, and can be directly extended to abstract graphs defined purely combinatorially without reference to an embedding in a specific manifold. In the following, we will not make direct use of the embedding of spin network states, so our definition and procedures can be applied to any of those cases. Nevertheless, since we do not discuss the coarse-graining from an embedding point of view, it is simpler to consider all our definitions as for abstract graphs.

\subsection{Coarse-graining by gauge-fixing}

Let us now discuss the main context of this paper: the coarse-graining of spin networks for  loop quantum gravity. The idea of coarse-graining is to integrate out the microscopic degrees of freedom, by an iterative procedure, up to some given energy or length scale to get the effective dynamics of the macroscopic degrees of freedom.
In condensed matter models, one typically works on a regular lattice with degrees of freedom living on its edges and/or nodes and one can decimate consistently the variables, integrating out one node out of two for example, and thus derive an effective Hamiltonian on the coarser lattice. The length scale is set by the lattice spacing. In quantum field theory, the renormalisation group scheme integrates out quantum fluctuations of the field of high momentum and energy to derive an effective dynamics on the low momentum degrees of freedom.
In general relativity, the main difficulty is that the space-time geometry itself has become dynamical thus leading to some serious obstacles: in a background independent context, we face the problems of defining consistently a length or energy scale and of properly localizing perturbations and degrees of freedom both in position and momentum. These issues persist in the quantum theory.

In loop quantum gravity, one could think that the natural graph structure of the theory makes it simpler to tackle the coarse-graining of the theory. However, even putting aside the huge complication of fluctuating graphs and graph superpositions, working out   the coarse-graining of loop quantum gravity on a fixed graph still faces the problem of localizing and determining the energy scale of the geometry fluctuations. Indeed, a natural coarse-graining procedure on a fixed graph is to subdivide it into a partition of bounded (usually connected) regions and to collapse those subgraphs to single points. The internal geometrical information carried by the spin network state on those subgraphs would be coarse-grained to some effective data living at the new node of the coarser graph, as illustrated on fig.\ref{fig:coarsegraining}.  Integrating over these local degrees of freedom would lead to new effective dynamics on the coarser graph. Such a  procedure would then be iterated  to obtain a tower of effective theories \textit{\`a la} Wilson for loop quantum gravity towards a large scale limit.
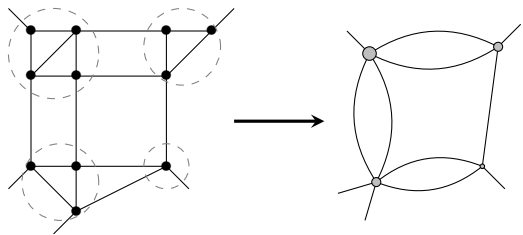
\begin{figure}[h!]

\centering

\begin{tikzpicture}[scale=0.3]
\coordinate(A) at (0,0);
\coordinate(B) at (2,0);
\coordinate(C) at (2,-2);
\coordinate(D) at (0,-2);

\coordinate(E) at (6,0);
\coordinate(F) at (6,-2);
\coordinate(G) at (8,0);
\coordinate(H) at (0,-6);
\coordinate(I) at (2,-6);
\coordinate(J) at (2,-8);
\coordinate(K) at (6,-6);

\draw (A) -- (B);
\draw (B) -- (C);
\draw (C) -- (D);
\draw (D) -- (A);
\draw (B) -- (D);

\draw (A) -- ++(-1,1);
\draw (B) -- (E);
\draw (C) -- (F);
\draw (C) -- (I);
\draw (D) -- (H);

\draw (G) -- ++(1,1);
\draw (E) -- (F) -- (G) -- (E);

\draw (H) -- ++(-1,-1);
\draw (J) -- ++(-1,-1);
\draw (H) -- (I) -- (J) -- (H);

\draw (I) -- (K);
\draw (J) -- (K);
\draw (K) -- (F);

\draw (K) -- ++(1,-1);

\draw (A) node {$\bullet$};
\draw (B) node {$\bullet$};
\draw (C) node {$\bullet$};
\draw (D) node {$\bullet$};
\draw (E) node {$\bullet$};
\draw (F) node {$\bullet$};
\draw (G) node {$\bullet$};
\draw (H) node {$\bullet$};
\draw (I) node {$\bullet$};
\draw (J) node {$\bullet$};
\draw (K) node {$\bullet$};

\draw[gray,dashed] (1,-1) circle(2);
\draw[gray,dashed] (1.3,-6.7) circle(1.7);
\draw[gray,dashed] (6.7,-0.7) circle(1.7);
\draw[gray,dashed] (K) circle(1);

\draw[->,>=stealth,very thick] (9,-4) -- (13,-4);

\coordinate(O1) at (15,-1);
\coordinate(O2) at (15.3,-6.7);
\coordinate(O3) at (20.7,-0.7);
\coordinate(O4) at (20,-6);
\draw (O1) -- ++(-1,1);
\draw (O2) -- ++(-1.7,-0.5);
\draw (O2) -- ++(-0.5,-1.7);
\draw (O3) -- ++(1,1);
\draw (O4) -- ++(1,-1);

\draw (O1) to[bend left] (O2);
\draw (O1) to[bend right] (O2);
\draw (O1) to[bend left] (O3);
\draw (O1) to[bend right] (O3);
\draw (O2) to[bend left] (O4);
\draw (O2) to[bend right] (O4);
\draw (O4) -- (O3);

\draw[fill=lightgray] (O1) circle(0.3);
\draw[fill=lightgray] (O2) circle(0.2);
\draw[fill=lightgray] (O3) circle(0.2);
\draw[fill=lightgray] (O4) circle(0.1);
\end{tikzpicture}

\caption{\label{fig:coarsegraining}
We coarse-grain a graph by partionning it into disjoint connected subgraphs. We will reduce each of these bounded region of space by a single vertex of the coarser graph. Since each of these regions of space had some internal geometrical structure and were likely carrying curvature, the natural question is whether spin network vertices carry each data to account for these internal structure and curvature. We will see that standard spin network vertices can be interpreted as flat and that we need to introduce some new notion of ``curved vertices'' carrying extra algebraic information and define new extensions of spin network states more suitable to the process of coarse-graining loop quantum gravity.}

\end{figure}

It remains to decide which partition to choose in practice, which one is the most ``coarse-grainable''. We need to identify the regions of the spin network state whose geometry has the smallest (quantum) fluctuations. Since the algebraic data -spins and intertwiners- living on the graph determine the discrete geometry defined by the spin network state, the combinatorial data of the subgraph is not enough to decide if it is to be coarse-grained. One actually needs to find a suitable scale function -length or energy or another geometrical observable such as curvature- and to use an optimization algorithm running through all possible bounded regions and partitions of the graph in order to find the correct partition to coarse-grain at each step. In simpler words, the obstacle is that the geometry corresponding to the considered graph is not entirely determined by the combinatorial definition of the graph but crucially depends on the algebraic data living on it and carried by the spin network state. And on top of this difficulty remains to find a consistent way to deal with graph fluctuations and superpositions.

\medskip

We propose a truncation of the theory re-introducing a background lattice through coarse-graining. From the point of view of a given observer, one chooses a lattice, which defines the network of points whose geometry the observer will probe. The lattice is not considered as the fundamental graph underlying the physical spin network state. Instead, since the observer is assumed to have a finite resolution, its nodes represent bounded regions of space whose internal geometry can fluctuate. Then, if we consider a spin network states based on a graph with a very fine structure, we will coarse-grain it onto our chosen lattice. Such a scheme allows to take into account 
graph fluctuations and superpositions while actually working on a fixed lattice. Indeed, considering a superposition of graphs, it will live by cylindrical consistency on a finer graph containing both graphs. Then we will coarse-grain the quantum geometry state on the finer graph until it lives on our reference lattice.

A key step of this procedure is the coarse-graining of subgraphs to nodes. We use the ``coarse-graining through gauge-fixing'' procedure introduced in \cite{Livine:2006xk, Livine:2013gna} and also exploited in \cite{Dittrich:2014wpa,Bahr:2015bra} to reformulate the algebra of geometrical observables in loop quantum gravity. This is based on the  gauge-fixing  for spin networks defined earlier in \cite{Freidel:2002xb}, which allows to collapse an arbitrary subgraph to a {\it flower} , that is a single vertex with self-loops -or petals- attached to it. These loops account  for the building-up of the curvature and thus of the gravitational energy density within these microscopic bounded regions which we will coarse-grain to single points on the measurement lattice chosen by the observer.

\medskip

Let us give a closer look to this gauge-fixing procedure and the resulting coarse-graining of spin networks.
At the classical level, a spin network state is given by the graph dressed with discrete holonomy-flux data: each oriented edge carries a $\SU(2)$ group element $g_{e}\,\in\SU(2)$ while each edge's extremity around  a vertex is colored with a vector $X^{v}_{e}\in\R^{3}$. So one edge carries two vectors, one living at its source vertex and the other living at its target vertex, respectively $X^{s,t}_{e}\equiv X^{s,t(e)}_{e}$. The group element gives the parallel transport of the vectors along the edges, that is $X^{t}_{e}=-\,g_{e}\triangleright X^{s}_{e}$ with the action of $g_{e}$ as a $\SO(3)$-rotation on the flat 3d space.
%
%
This obviously forces the two vectors to have equal norm, $|X^{t}_{e}|=|X^{s}_{e}|$, which is called the (area-)matching constraint.
One requires another set of constraints: we impose the closure constraint at each vertex $v$, so that the sum of the fluxes around the vertex vanishes, $\sum_{e\ni v} X^{v}_{e}=0$. This holonomy-flux data can be interpreted as some discrete geometry in the framework of twisted geometries
\cite{Freidel:2010aq,Freidel:2013bfa}.
This is achieved through Minkowski's theorem stating that the closure constraint determines a unique convex polyhedron in flat 3d space dual to each vertex $v$, such that the fluxes $X^{v}_{e}$ are the normal vectors to the polyhedron faces.

Curvature appears as non-trivial holonomies around loops $\cL$ of the graph, when $\overrightarrow{\prod_{e\in\cL}}g_{e}\,\ne\id$. As pointed out in \cite{Livine:2013gna}, coarse-graining a subgraph carrying non-trivial curvature leads to an effective vertex breaking the closure constraint. This underlines the fact that a generalization of spin network states is required in order to properly carry out a coarse-graining procedure: we need an extended structure allowing for {\it curved} vertices.

We illustrate this in fig.\ref{fig:nonclosure}, in 2d instead of 3d. Let us consider a bounded region in space and the normals to its boundary.
%
%
Due to gauge-invariance, if the region contains a single vertex, the sum of the normals will sum up to zero. But if there are loops inside the region, the parallel transport around these loops might introduce non-trivial rotations. And indeed, as soon as the parallel transport around the loops is non-trivial, the sum of the normals is no longer zero, leading to a closure defect \cite{Livine:2013gna}. This is natural and translates the fact that curvature is carried by the loops of the spin network. And this must be taken  into account when coarse-graining.

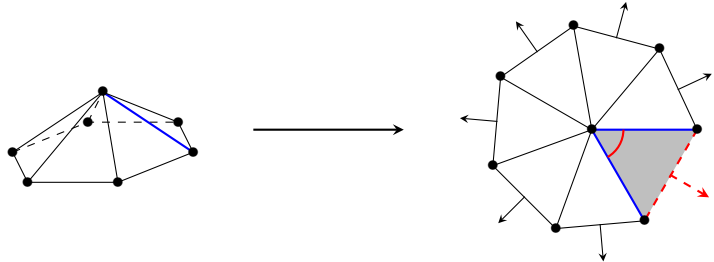
\begin{figure}[h!]

\centering

\begin{tikzpicture}

\def \scale {1.2}
\def \d {0.2}

\coordinate (A1) at (1*\scale,\d,0);
\coordinate (A2) at (0.5*\scale,\d,0.866*\scale);
\coordinate (A3) at (-0.5*\scale,\d,0.866*\scale);
\coordinate (A4) at (-1*\scale,\d,0);
\coordinate (A5) at (-0.5*\scale,\d,-0.866*\scale);
\coordinate (A6) at (0.5*\scale,\d,-0.866*\scale);
\coordinate (B) at (0,0.8+\d,0);

\draw (A6) -- (A1) -- (A2) -- (A3) -- (A4);
\draw[dashed] (A4) -- (A5) -- (A6);

\draw[blue,thick] (A1) -- (B);
\draw (A2) -- (B);
\draw (A3) -- (B);
\draw (A4) -- (B);
\draw[dashed] (A5) -- (B);
\draw (A6) -- (B);

\draw (A1) node {$\bullet$};
\draw (A2) node {$\bullet$};
\draw (A3) node {$\bullet$};
\draw (A4) node {$\bullet$};
\draw (A5) node {$\bullet$};
\draw (A6) node {$\bullet$};
\draw (B) node {$\bullet$};

\draw[->,>=stealth,thick] (2,0.5) -- (4,0.5);

\coordinate (P) at (6.5,0.5);

\def \step {50}
\def \r {1.4}

\coordinate (C1) at ($(P) + (0*\step:\r)$);
\coordinate (C2) at ($(P) + (1*\step:\r)$);
\coordinate (C3) at ($(P) + (2*\step:\r)$);
\coordinate (C4) at ($(P) + (3*\step:\r)$);
\coordinate (C5) at ($(P) + (4*\step:\r)$);
\coordinate (C6) at ($(P) + (5*\step:\r)$);
\coordinate (C7) at ($(P) + (6*\step:\r)$);

\fill[lightgray] (C7) -- (P) -- (C1) -- cycle;
\draw[red,thick,dashed] (C1) -- (C7);

\draw[blue,thick] (P) -- (C1);
\draw (P) -- (C2);
\draw (P) -- (C3);
\draw (P) -- (C4);
\draw (P) -- (C5);
\draw (P) -- (C6);
\draw[blue,thick] (P) -- (C7);

\draw (C1) -- (C2) -- (C3) -- (C4) -- (C5) -- (C6) -- (C7);

\draw[red,thick] ($(P) + (6*\step:0.3*\r)$) arc (6*\step:360:0.3*\r);

\def \s {0.01}
\def \si {0.9}

\draw[->,>=stealth] ($(P) + (0.5*\step:\r*\si)$) -- ++(0.5*\step:\step*\s);
\draw[->,>=stealth] ($(P) + (1.5*\step:\r*\si)$) -- ++(1.5*\step:\step*\s);
\draw[->,>=stealth] ($(P) + (2.5*\step:\r*\si)$) -- ++(2.5*\step:\step*\s);
\draw[->,>=stealth] ($(P) + (3.5*\step:\r*\si)$) -- ++(3.5*\step:\step*\s);
\draw[->,>=stealth] ($(P) + (4.5*\step:\r*\si)$) -- ++(4.5*\step:\step*\s);
\draw[->,>=stealth] ($(P) + (5.5*\step:\r*\si)$) -- ++(5.5*\step:\step*\s);

\draw[->,>=stealth,red,thick,dashed] ($(P) + ({180+3*\step}:1.2)$) -- ++({180+3*\step}:{360*\s-6*\step*\s});

\draw (P) node {$\bullet$};

\draw (C1) node {$\bullet$};
\draw (C2) node {$\bullet$};
\draw (C3) node {$\bullet$};
\draw (C4) node {$\bullet$};
\draw (C5) node {$\bullet$};
\draw (C6) node {$\bullet$};
\draw (C7) node {$\bullet$};

\end{tikzpicture}

\caption{On this figure, we represented the dual graph of a 2d trivalent graph. The curvature at the vertex manifests itself as a defect in the closure condition. This can be seen by flattening the triangulation, which amounts to gauge-fix the variables. The curvature manifests itself as a gap (in gray on the figure) at some edge (in blue on the figure) in the flattened manifold. The closure defect can be seen as the missing normal coming from the closure of the flattened polygon (in red on the figure).}
\label{fig:nonclosure}

\end{figure}


A rigorous way to make this explicit is to gauge-fix the spin network state, following the procedure devised in \cite{Freidel:2002xb}.
Let us consider a bounded region of a larger spin network, defined as a finite connected subgraph $\gamma$ of the larger graph $\Gamma$,  as in fig.\ref{fig:GaugeFix}.
The procedure goes as follow:
\begin{enumerate}

\item Choose arbitrarily a root vertex $v_{0}$ of the subgraph and select a maximal tree $T$ of the region: 

The subgraph being connected, the maximal tree  goes through every vertex of the region and defines a unique path of edges from the root vertex $v_{0}$ to any vertex of the subgraph.

\item Gauge-fix iteratively all the group elements along the edges of the tree $g_{e\in T}=\id$:

Using the gauge-invariance of the wave-functions as given by \eqref{gaugeinv} with gauge transformations acting at every vertex by $\SU(2)$ group elements $h_{v}$ as $g_{e}\arr h_{s(e)}^{-1}g_{e}h_{t(e)}$, we can start from the root of the tree $v_{0}$ and progress through the tree until we reach the boundary of our subgraph. We define the appropriate gauge transformations $h_{v}$ at every vertex in order to fix all the group elements along the edges of the tree to the identity $\id$. The absence of loops in the tree, by definition, guarantees the consistency of this gauge-fixing. We can somewhat interpret this maximal tree as a synchronization network: we set all the parallel transports along the tree edges to the identity, thus synchronizing the reference frames at all the vertices and identifying them to a single reference frame living at the root of the subgraph. This realizes the coarse-graining of the subgraph $\gamma$ to its chosen root vertex $v_{0}$.
The action of $\SU(2)$ gauge transformations inside the region is not entirely gauge-fixed and we are still left with the $\SU(2)$ gauge transformations at the root vertex.
%

\item Having collapsed the subgraph $\gamma$ to its root vertex $v_{0}$, the edges of the subgraph $\gamma$ which are not in the tree, $e\in\gamma\setminus T$ label all the (independent) loops of the subgraph and lead to self-loops attached to the $v_{0}$:

As illustrated on fig.\ref{fig:GaugeFix}, these self-loops or little loops carry the holonomies around the loops of the original subgraph $\gamma$, that is the curvature living in the bounded region. The flux-vectors living on the boundary edges, linking the region to the outside bulk, generically do not satisfy the closure constraint anymore since the effective vertex does satisfy a closure constraint which takes into account the flux-vectors of those boundary edges but also of the internal loops. The closure defect, induced by the little loops, thus reflects the non-trivial internal structure  of the coarse-grained subgraph and curvature developed in the corresponding region of the spin network state. The interested reader can find details and proof in the previous work \cite{Livine:2013gna}.


\end{enumerate}

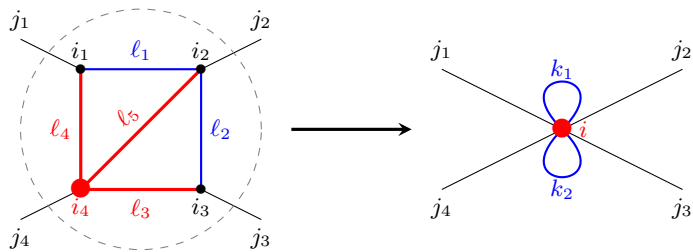
\begin{figure}[h!]

\begin{tikzpicture}[scale=0.8]
\coordinate(A) at (0,0);
\coordinate(B) at (2,0);
\coordinate(C) at (2,-2);
\coordinate(D) at (0,-2);

\draw[blue,thick] (A) -- (B) node[midway,above] {$\ell_1$};
\draw[blue,thick] (B) -- (C) node[midway,right] {$\ell_2$};
\draw[red,very thick] (C) -- (D) node[midway,below] {$\ell_3$};
\draw[red,very thick] (D) -- (A) node[midway,left] {$\ell_4$};
\draw[red,very thick] (B) -- (D) node[midway,sloped,above] {$\ell_5$};

\draw (A) -- ++(-1,0.5) node[above] {$j_1$};
\draw (B) -- ++(1,0.5) node[above] {$j_2$};
\draw (C) -- ++(1,-0.5) node[below] {$j_3$};
\draw (D) -- ++(-1,-0.5) node[below] {$j_4$};

\draw (A) node {$\bullet$} node[above]{$i_1$};
\draw (B) node {$\bullet$} node[above]{$i_2$};
\draw (C) node {$\bullet$} node[below]{$i_3$};
\draw[red] (D) node[scale=2] {$\bullet$} node[below]{$i_4$};

\draw[gray,dashed] (1,-1) circle(2);

\draw[->,>=stealth,very thick] (3.5,-1) -- (5.5,-1);

\coordinate(O) at (8,-1);

\draw (O) -- ++(-2,1) node[above] {$j_1$};
\draw (O) -- ++(2,1) node[above] {$j_2$};
\draw (O) -- ++(2,-1) node[below] {$j_3$};
\draw (O) -- ++(-2,-1) node[below] {$j_4$};
\draw[blue,thick,scale=3] (O) to[loop] (O);
\draw[blue,thick] (O) ++(0,1) node {$k_1$};
\draw[blue,thick,scale=3,rotate=180] (O) to[loop] (O);
\draw[blue,thick] (O) ++(0,-1) node {$k_2$};

\draw[red] (O) node[scale=2] {$\bullet$} ++(0.35,0) node{$i$};
\end{tikzpicture}

\caption{Coarse-graining via gauge-fixing: we can gauge-fix the subgraph using a maximal subtree (in red). The remaining edges (in blue) correspond to loops on the coarse-grained vertex. There is a residual gauge-freedom at the coarse-grained vertex that corresponds to the action of the gauge group at the root of the tree (red vertex on the figure).}
\label{fig:GaugeFix}
\end{figure}

This gauge-fixing procedure allows to clearly identify and distinguish between the degrees of freedom of the internal geometry of the considered bounded region of space to coarse-grain. The tree encodes the internal combinatorial structure of the region and describes the network of points and links within: they provide the bulk structure on which we can create curvature.  The little loops and the $\SU(2)$ group elements coloring them are the  excitations of the parallel transport and curvature. Together, tree and little loops attached to a vertex describe all its internal structure and are the extra data needed to define \textit{curved vertices} for the effective coarse-grained theory. These curvature excitations create a closure defect for the flux-vectors living on the boundary edges linking the coarse-grained vertex -the root vertex- to the rest of the spin network (obtained by the actually satisfied closure constraint between boundary edges and little loops)

When coarse-graining in practice, we do not want to retain all the information about the internal geometry, but only want to retain the degrees of freedom most relevant to the dynamics and interaction with the exterior geometry. In the next section, we will therefore introduce a hierarchy of extensions of spin network states with {\it curved vertices}, from the finest notion of spin networks decorated with both trees and little loops to the coarser notion of spin networks with a simple tag at each vertex recording the induced closure defect.

%

\subsection{A hierarchy of coarse-grained spin network structures}

In loop quantum gravity, we start with spin network states, which are graphs decorated with spins on the edges and intertwiners at the vertices:
\be
{\mathcal{H}}_\Gamma = \bigoplus_{\{j_e,i_v\}} \mathbb{C}|j_e,i_v\rangle\,.
\ee
Curvature is carried loops of the graph.
We have argued that coarse-graining these networks should naturally lead to extended spin networks that can carry localized curvature excitations at the vertices. Following the coarse-graining through gauge-fixing procedure\footnotemark, we propose a hierarchy of three possible extensions of the spin network states, which depend on how much extra information and structure are added to each vertex:
\footnotetext{
Another approach is to define spin networks made of intertwiners directly interpretable as dual to polyhedron in a curved space. These has been developed in the framework of spin networks for loop quantum gravity with a non-vanishing cosmological constant and is based on a quantum deformation of the $\SU(2)$ gauge group \cite{Dupuis:2013haa,Bonzom:2014wva,Dupuis:2014fya,Charles:2015lva,Haggard:2014xoa,Haggard:2015ima,Haggard:2015yda}. However, it is not yet clear how, if possible, to depart from a homogeneous curvature and glue pieces carrying a different curvature, thus obtaining actual spin networks with variable curvature.
One possible link with our present framework would be to show that these curved and quantum-deformed intertwiners can be obtained in a continuum limit as a vertex with an infinite number of little loops creating a constant homogeneous curvature excitation (for instance, triangulating a hyperbolic tetrahedron with finer and finer tetrahedra which can be considered as flat in an infinite refinement limit).
}
\begin{figure}
\begin{subfigure}[t]{.33\linewidth}
\centering
\begin{tikzpicture}[scale=0.7]
\coordinate(O1) at (0,0);
\coordinate(O2) at (2.5,3.5);
\coordinate(A) at ($(O2)+(0,0)$);
\coordinate(B) at ($(O2)+(1,0)$);
\coordinate(C) at ($(O2)+(1,-1)$);
\coordinate(D) at ($(O2)+(0,-1)$);

\draw (O1) -- ++(-2,1) node[above] {$j_1$};
\draw (O1) -- ++(2,1) node[above] {$j_2$};
\draw (O1) -- ++(2,-1) node[below] {$j_3$};
\draw (O1) -- ++(-2,-1) node[below] {$j_4$};
\draw[blue,thick] (O1) to[loop,scale=3] (O1) ++(0,1) node {$k_1$};
\draw[blue,thick] (O1) to[loop,scale=3,rotate=180] (O1) ++(0,-1) node {$k_2$};

\draw[red] (O1) node[scale=2] {$\bullet$} ++(0.35,0) node{$i$};

\draw[gray,dashed] (O1) circle (0.5) ++(45:0.5) -- (45:3) ++(45:1.2) coordinate (O3) circle (1.2);

\clip (O3) circle (1.2);

\draw[blue,thick] (A) -- (B);
\draw[blue,thick] (B) -- (C);
\draw[red,very thick] (C) -- (D);
\draw[red,very thick] (D) -- (A);
\draw[red,very thick] (B) -- (D);

\draw (A) -- ++(-2,1);
\draw (B) -- ++(2,1);
\draw (C) -- ++(2,-1);
\draw (D) -- ++(-2,-1);

\draw (A) node {$\bullet$};
\draw (B) node {$\bullet$};
\draw (C) node {$\bullet$};
\draw[red] (D) node[scale=2] {$\bullet$};
\end{tikzpicture}

\caption{All the information can be preserved by carrying the $SU(2)$ labels and an unfolding tree describing the inner details of the coarse-grained vertex.}\label{fig:loopy_a}
\end{subfigure}%
\hspace{2mm}
\begin{subfigure}[t]{.28\linewidth}
\centering
\begin{tikzpicture}[scale=0.7]
\coordinate(O1) at (0,0);

\draw (O1) -- ++(-2,1) node[above] {$j_1$};
\draw (O1) -- ++(2,1) node[above] {$j_2$};
\draw (O1) -- ++(2,-1) node[below] {$j_3$};
\draw (O1) -- ++(-2,-1) node[below] {$j_4$};
\draw[blue,thick] (O1) to[loop,scale=3] (O1) ++(0,1) node {$k_1$};
\draw[blue,thick] (O1) to[loop,scale=3,rotate=180] (O1) ++(0,-1) node {$k_2$};

\draw[red] (O1) node[scale=2] {$\bullet$} ++(0.35,0) node{$i$};
\end{tikzpicture}
\caption{The particular subgraph can be forgotten and only the $SU(2)$ information is preserved.}\label{fig:loopy_b}
\end{subfigure}
\hspace{2mm}
\begin{subfigure}[t]{.28\linewidth}
\centering
\begin{tikzpicture}[scale=0.7]
\coordinate(O1) at (0,0);

\draw (O1) -- ++(-2,1) node[above] {$j_1$};
\draw (O1) -- ++(2,1) node[above] {$j_2$};
\draw (O1) -- ++(2,-1) node[below] {$j_3$};
\draw (O1) -- ++(-2,-1) node[below] {$j_4$};
\draw[blue,thick] (O1) to ++(0,0.5) node[above]{$j,m$};

\draw[red] (O1) node[scale=2] {$\bullet$} ++(0,-0.3) node{$i'$};

\end{tikzpicture}
\caption{Everything except the closure defect is forgotten. Only a ``tag'' remains.}\label{fig:loopy_c}
\end{subfigure}
\caption{The hierarchy of possible coarse-graining frameworks}\label{fig:loopy}
\end{figure}
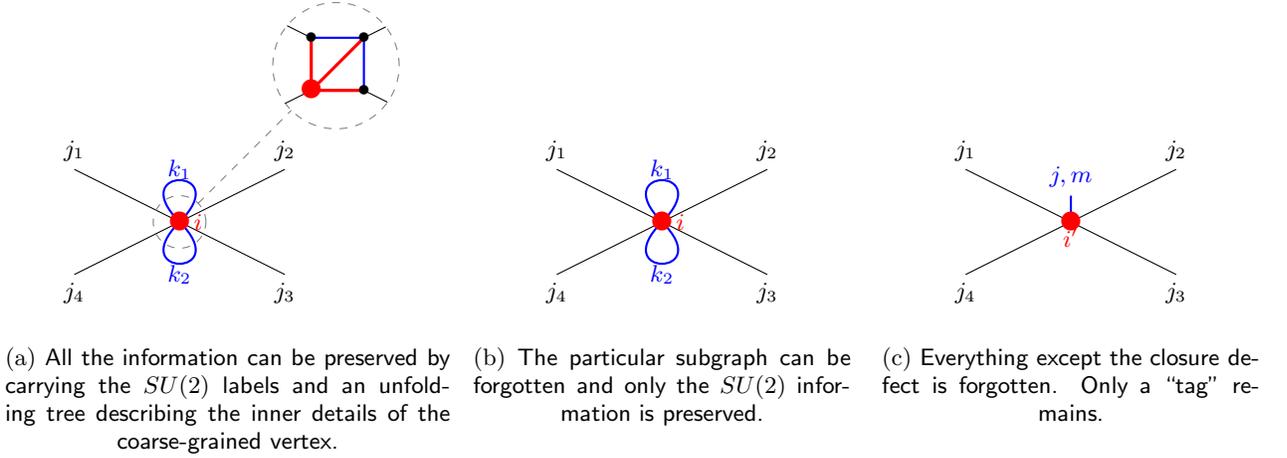
\begin{enumerate}

\item {\bf Folded spin networks~:}

In the first scenario, we follow the gauge-fixing procedure but we do a minimal coarse-graining, retaining as much information as possible on the original state.
Each vertex is allowed with an arbitrary number of little loops attached to it and is endowed with a tree connecting the ends of the external edges and of the internal loops, as represented in fig.\ref{fig:loopy_a}. This tree can be seen as a circuit telling us how to unfold the vertex, reversing the gauge-fixing procedure and recovering the original (finer) graph.
This Hilbert space $\cH_{\Gamma}^\mathrm{folded}$ can be written formally as:
\begin{equation}
\cH_{\Gamma}^\mathrm{folded}
= \bigoplus_{\{j_e,j^{(v)}_{\ell},i_v,\mathcal{T}_v\}} \mathbb{C}\,|j_e,j^{(v)}_{\ell},i_v,\mathcal{T}_v\rangle\,.
\end{equation}
 $\mathcal{T}_v$ is the unfolding tree for each vertex, $j^{(v)}_{\ell}$ are the spins carried by the additional loops labeled by the index $\ell$ and the intertwiners $i_{v}$ now lives in the tensor product of the spins $j_{e}$ of the edges linking to the other neighboring vertices and (twice) the spins $j^{(v)}_{\ell}$ living on the internal loops (because each loop has its two ends at the vertex).
 
 With such an internal space at each vertex, we actually lose no information at all on the internal degrees of freedom. Starting with a spin network state living on a finer graph $\tGamma$, we simply gauge-fix it to a spin network on our coarser graph $\Gamma$. And we can follow the reverse path. Using the tree at each vertex, we can fully reconstruct the original finer graph $\tGamma$ thus simply perform generic gauge transformations to recover the fully gauge-invariant spin network state.

Thus the chosen graph $\Gamma$ can be considered as a skeleton graph, to which we can add extra information to represent spin network states living on any (finer) graph. In a sense, we have not done any coarse-graining yet. The truncation of the theory will happen when defining the dynamics on the folded spin network Hilbert space, distinguishing actual edges and spins of our  skeleton lattice -the background- from spins and edges on the unfolding trees and little loops, when the fundamental dynamics would have considered them on equal footing.

\item {\bf Loopy spin networks~:}

In a second scenario, we coarse-grain the internal structure of the effective vertices by discarding the unfolding trees. We keep the curvature excitations living on the little loops, but we discard the combinatorial information of the internal subgraph: we forget that the  vertex effectively represents an actual extended region of space and we localize all the internal curvature degrees of freedom on that coarse-grained vertex. This leads to loopy spin networks, with an arbitrary number of loops at each vertex but no unfolding tree data:
\begin{equation}
\mathcal{H}^{\mathrm{loopy}}_{\Gamma} = \bigoplus_{\{j_e,j^{(v)}_{\ell},i_v\}} \mathbb{C}|j_e,j^{(v)}_{\ell},i_v\rangle\,,
\end{equation}
where the $j^{(v)}_{\ell}$ are the spins living on the little loops attached to the vertex $v$ and the intertwiners $i_{v}$ live again in the tensor product of the spins carried by the graph edges attached to the vertex $v$ and the spins carried by its little loops.

Now our chosen graph $\Gamma$ for loopy spin network states is to be considered as a background graph. The little loops are explicit local excitations of the gravitational fields located at each vertex of the graph. A given loopy spin network comes from the coarse-graining of several possible finer spin network states living on finer graph, but we lack the unfolding tree information to recover the original more fundamental state.

The truncation of full theory is clear. Spin network states on the ``loopy graphs'' living on top on $\Gamma$, that is the base graph $\Gamma$ plus an arbitrary number of self-loops at every vertices, are already in the Hilbert space of the loop quantum gravity, although we do not usually focus on such graphs. Restricting ourselves to this subset of states is a clear truncation of the full Hilbert space. The difference with the standard interpretation is that we think here of the base graph $\Gamma$ as embedded in the space manifold, while the little loops are abstract objects decorating the base graph vertices.

Since we have local degrees of freedom, carried by the little loops, we need to discuss their statistics, which leads to a few variations of this theme:

\begin{enumerate}
\item {\it Distinguishable loops~:}
First, it is natural to consider that the loops are distinguishable as they come from a substructure. The loops do come from different edges of a finer graph and create curvature excitations  at different places within the coarse-grained bounded region. As a result, we should distinguish them and allow to number and order them.
%

\item {\it Undistinguishable bosonic loops~:}
A second possibility is to push further along the logic  of coarse-graining  and to consider that the loops undistinguishable since we do not have access anymore to the specific substructure. This should lead to bosonic statistics, as expected for gravitational field exicitations. Formally, this can be written as the identification:
\begin{equation}
|j_e,i_v,j^{(v)}_{\ell}\rangle = |j_e,i_v,j^{(v)}_{\sigma_{v}(\ell)}\rangle
\end{equation}
for any permutation $\sigma_{v}\in\,S_{\#\ell}$ in the symmetric group of order $\#\ell$ when the vertex $v$ has $\#\ell$ loops. This point of view is compatible with considering the action of space diffeomorphisms on the little loops around the vertex as gauge transformations.

\item  {\it Anyonic statistics~:} 
We can easily imagine other statistics, for instance by allowing for a phase in the equality above (i.e a non-trivial representation of the permutation group). In fact, instead of thinking of the vertex as a mere point, we can represent the boundary of the bounded region as a sphere and consider the little loops as living on a sphere around it. Then the diffeomorphism invariance on the sphere will lead to an action of the braiding group leading to interesting anyonics statistics, similarly to the punctures of a Chern-Simons theory as  already explored  in the case of black holes in loop quantum gravity \cite{Pithis:2014uva}.

\end{enumerate}

\item {\bf Tagged spin networks~:}

In this third and last scenario, we fully coarse-grain the internal geometry of the bounded region now reduced to a graph vertex. We discard the unfolding tree, used in the gauge-fixing and unfixing procedure, and we integrate out the little loops attached to the vertex. All we retain is the closure defect induced by the non-trivial holonomies and spins carried by those little loops. The fact that coarse-graining spin networks, or their classical counterpart of twisted geometries, leads to closure defect, accounting for the presence of a non-trivial curvature within the coarse-grained region was already pointed out in \cite{Livine:2013gna}.
Here, the simplest method to see how this comes about is to use the intermediate spin decomposition of the intertwiner at the vertices, as illustrated on fig.\ref{fig:intermediatespin}, introducing a fiducial link separating the external edges from the internal loops:
\be
\mathrm{Inv}_{\SU(2)}\,\Big{[}
\bigotimes_{e}\cV^{j_{e}}
\otimes
\bigotimes_{\ell}\big{(}\cV^{j_{\ell}}\otimes\cV^{j_{\ell}}\big{)}
\Big{]}
\,=\,
\bigoplus_{J}
\,
\mathrm{Inv}_{\SU(2)}\,\Big{[}
\cV^{J}
\otimes
\bigotimes_{e}\cV^{j_{e}}
\Big{]}
\otimes
\mathrm{Inv}_{\SU(2)}\,\Big{[}
\cV^{J}
\otimes
\bigotimes_{\ell}\big{(}\cV^{j_{\ell}}\otimes\cV^{j_{\ell}}\big{)}
\Big{]}\,.
\ee
\begin{figure}[h!]

\centering

\begin{tikzpicture}[scale=0.8]
\coordinate(O1) at (0,0);

\draw (O1) -- ++(-2,1);
\draw (O1) -- ++(-2,0.5);
\draw (O1) -- ++(-2,-0.5);
\draw (O1) -- ++(-2,-1);
\draw (O1) to[in=-25,out=25,loop,scale=3] (O1);
\draw (O1) to[in=30,out=80,loop,scale=3] (O1);
\draw (O1) to[in=-80,out=-30,loop,scale=3] (O1);

\draw (O1) node {$\bullet$} ++(-0.15,0.5) node{$i_v$};

\draw[->,>=stealth,very thick] (3,0) -- (5,0);

\coordinate(O2) at (8,0);
\coordinate(O3) at (10,0);

\draw (O2) -- ++(-2,1);
\draw (O2) -- ++(-2,0.5);
\draw (O2) -- ++(-2,-0.5);
\draw (O2) -- ++(-2,-1);

\draw[in=-25,out=25,scale=3] (O3) to[loop] (O3);
\draw[in=30,out=80,scale=3] (O3) to[loop] (O3);
\draw[in=-80,out=-30,scale=3] (O3) to[loop] (O3);

\draw[red] (O2) -- (O3) node[midway,below]{$J$};
\draw (O2) node {$\bullet$} ++(0.12,0.4) node{$i_v^{J}$};
\draw (O3) node {$\bullet$} ++(-0.2,0.4) node{$\tilde{i}_v^{J}$};

\end{tikzpicture}

\caption{We represent a loopy vertex $v$, here with three little loops attached to it. The intertwiner $i_{v}$ can be decomposed onto the intermediate spin basis, where we introduce a fiducial edge between the external legs and the internal loops. This orthogonal basis is labeled by the intermediate spin $J$, and two intertwiners $i^{J}_{v}$ and $\tilde{i}^{J}_{v}$ intertwining between that intermediate spin  and respectively the external legs or the internal loops.} 
\label{fig:intermediatespin}

\end{figure}
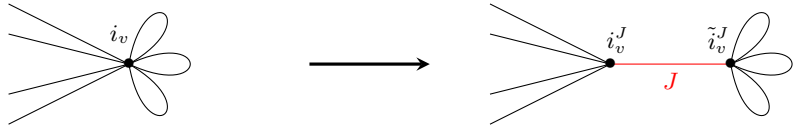
This spin $J_{v}$ living at the vertex $v$ encodes the closure defect and is the only extra information with which we decorate the graph.
We call it the {\it tag} and amounts to adding an open leg to every vertex of the graph. This open edge is colored with the spin $J_{v}$ and a vector in that $\SU(2)$ representation. Using the standard spin basis labeled by magnetic moment number $M$, the Hilbert space of {\it tagged spin networks} on the base graph $\Gamma$ is then formally defined as:
\be
\mathcal{H}_\Gamma^{\mathrm{tag}} = \bigoplus_{\{j_e,J_v,M_v,i_v\}} \mathbb{C}|j_e,J_v,M_v,i_v\rangle\,,
\ee
where the intertwiner $i_{v}$ at the vertex $v$ now lives in the tensor product of the spins $j_{e\ni v}$ on the external edges $e$ attached to the vertex and of the vertex tag $J_{v}$.

The state $|J_v,M_v\ra$ is the quantized version of the closure defect vector. Indeed, at the classical level, as shown in \cite{Livine:2013gna}, the sum of the flux-vectors living on the external edges $e\ni v$ does not vanish anymore and should be balanced by the sum of the flux-vectors living on the internal loops. This defect vector means that there is  no convex polyhedron dual to the vertex, as usual in twisted geometries. One way to go is to try to open the polyhedron somehow, which wouldn't have a clear geometrical interpretation. Instead we propose to interpret it as the dual convex polyhedron should not be embedded in flat space but in a (homogeneous) curved space, the curvature radius depending on the actual value of the closure defect. Progress in this direction has been achieved in the study of hyperbolic and spherical tetrahedra \cite{Bonzom:2014wva,Charles:2015lva,Haggard:2015ima} but we do not yet have an explicit  embedding and formula relating the curvature to the norm of the defect. It would ultimately be enlightening to relate this tag $J_{v}$ to the spectrum of some quasi-local energy operator in loop quantum gravity (e.g. \cite{Yang:2008th}), which would allow to view it as a measure of the gravitational energy density within the bounded coarse-grained region.

\end{enumerate}

\medskip

These three extended spin network structures are the heart of our present proposal for studying effective truncations for the coarse-graining of loop quantum gravity. The goal would be to reformulate the dynamics of loop quantum gravity on these new structures and study their renormalisation flow under coarse-graining. An important point is that these folded, loopy and tagged spin networks sidestep the problem of fluctuating graph dynamics and allow to project the whole dynamics on a fixed background graph, or skeleton, interpreted as the lattice postulated by the observer\footnotemark. We then have local excitations of the geometry, representing the internal fluctuations of the gravitational field in the coarse-grained regions, living at the graph vertices and represented by the new information attached to them, respectively unfolding trees, little loops or tags.
\footnotetext{
The background lattice can then be adapted to the studied models. We could choose a regular lattice or a much simpler graph, such as a flower with a single vertex and an arbitrary number of little loops. Such simple graphs could reveal useful in the study of highly symmetric problems as is the case in cosmology or in the study of Einstein-Rosen waves \cite{Korotkin:1997ps,Ashtekar:1996cm}.
}
The use of a background lattice, which might be regular, would simplify greatly the setting of a systematic coarse-graining of loop quantum gravity.

The folded spin networks are mathematically a simple gauge-fixing of spin networks onto the skeleton graph. In the following sections, we will focus on providing a clean mathematical definition of loopy and tagged spin networks and exploring the definition of a Fock space of loopy spin networks with bosonic statistics for the little loops living at every graph vertex.

\section{Loopy spin networks}

Here we  would like to define properly loopy spin networks and investigate their propreties.
Choosing a fixed graph $\Gamma$ with $E$ edges, and given numbers of little loops $N_{v}$ at each vertex $v$, we consider the following space of wave-functions on $\SU(2)^{\times\,(E+\sum_{v}N_{v})}$ invariant under $\SU(2)$  gauge transformations acting at every vertices:
\be
\psi\Big{(}\{g_{e}\,,\,h^{v}_{\ell}\}_{e,v\in\Gamma}\Big{)}
\,=\,
\psi\Big{(}\{a_{s(e)}g_{e}a_{t(e)}^{-1}\,,\,a_{v}h^{v}_{\ell}a_{v}^{-1}\}\Big{)}
\,,\qquad
\forall a_{v}\in\SU(2)^{\times V}\,.
\ee
The $\SU(2)$ gauge transformations act as usual on the edges $e$ of the graph, while they act by conjugation as expected on the little loops.
A basis is provided by the spin decomposition on functions in $L^{2}(\SU(2))$ as with standard spin networks. The loopy spin network basis states are labeled with a spin $j_{e}$ on each edge $e$, a spin $k^{v}_{\ell}$ on each little loop $\ell$ attached to a vertex $v$, and an intertwiners $i_{v}$ at each vertex leaving in the tensor product of the attached edges and of the loop spins:
\be
i_{v}\in\,\mathrm{Inv}_{\SU(2)}\,
\Big{[}
\bigotimes_{e\ni v}\cV^{j_{e}}
\,\otimes\,
\bigotimes_{\ell \ni v}(\cV^{k^{v}_{\ell}}\otimes\bar{\cV}^{k^{v}_{\ell}})
\Big{]}\,
\ee
so that the Hilbert space of loopy spin networks on the graph $\Gamma$ with given number $N_{v}$ of little loops at every vertex is, as announced in the previous section presenting the hierarchy of extended spin network structures:
\be
\cH_{\Gamma,\{N_{v}\}}^{\mathrm{loopy}}
\,=\,
L^{2}\big{(}
\SU(2)^{\times\,(E+\sum_{v}N_{v})}
\,/\,
\SU(2)^{\times V}
\big{)}
\,=\,
\bigoplus_{\{j_{e},k^{v}_{\ell},i_{v}\}}\,
\C\,|j_{e},k^{v}_{\ell},i_{v}\ra\,.
\ee
What needs to be properly defined and analyzed is the Hilbert space of states with arbitrary number of little loops, allowing $N_{v}$ to run all over $\N$ and summing over all these possibilities. To this purpose, the full graph structure $\Gamma$ does not intervene and we can ignore it and focus on the space of little loops around a single vertex. Thus, for the sake of simplifying the discussion, we will focus on a single vertex with no external, but with an arbitrary umber of little loops attached to it. This is the {\it flower} graph.

In this section, we will assume the little loops to be distinguishable. We define the spin network states with a given number of loops -the flower graph with fixed number of petals- and we then discuss the whole Hilbert space of states with arbitrary number of excitations by a projective limit. We define and analyze the holonomy operators acting on that space and we finally  implement the BF theory dynamics on that space as a first application of our framework and a consistency check.
We will tackle the case of indistinguishable little loops in the next section, imposing bosonic statistics and defining the holonomy operator on symmetrized  states.

\subsection{Loopy intertwiners}
\label{sec:LoopyIntertwiners}

Let us start with the flower graph with a fixed number $N$ of petals, that is a single vertex with $N$ little loops attached to it as drawn on fig.\ref{fig:singlevertex}. We are going to define the wave-functions on that graph, the corresponding decomposition on the spin and intertwiner basis and the action of the holonomy operators.
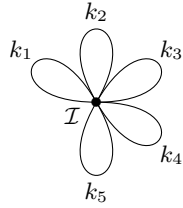
\begin{figure}[h!]
\centering
\begin{tikzpicture}
\coordinate(O1) at (0,0);

\draw (O1) to[in=-30,out=+30,loop,scale=3,rotate=90] (O1) ++(0,1.2) node {$k_2$};
\draw (O1) to[in=-30,out=+30,loop,scale=3,rotate=30] (O1) ++(1,0.7) node {$k_3$};
\draw (O1) to[in=-30,out=+30,loop,scale=3,rotate=-30] (O1) ++(1,-0.7) node {$k_4$};
\draw (O1) to[in=-30,out=+30,loop,scale=3,rotate=150] (O1) ++(0,-1.2) node {$k_5$};
\draw (O1) to[in=-30,out=+30,loop,scale=3,rotate=-90] (O1) ++(-1,0.7) node {$k_{1}$};

\draw (O1) node[scale=1] {$\bullet$} ++(-0.32,-0.18) node{$\cI$};

\end{tikzpicture}
\caption{We consider the class of special graph, flowers, with a single vertex and an arbitrary number $N$ of little loops attached to it. Here we have drawn a flower with $N=5$ petals. The spin network states on such graphs are labeled by a spin on each loop, $k_{\ell=1..N}$, and an intertwiner $\cI$ living in the tensor product $\bigotimes_{\ell=1}^{N} (\cV^{k_{\ell}}\otimes\bar{\cV}^{k_{\ell}})$.}
\label{fig:singlevertex}
\end{figure}

Wave-functions are gauge-invariant functions of $N$ group elements, that is functions on $\SU(2)^{\times N}$ invariant under the global action by conjugation:
\be
\Psi(h_1,...,h_N) = \Psi(gh_1g^{-1},...,gh_Ng^{-1})\,.
\ee
The scalar product is defined by integration with respect to the Haar measure on $\SU(2)$ and the resulting Hilbert space is: 
\be
\cH_{N}=
L^2\,\Big{(}\SU(2)^{\times N}/\mathrm{Ad}\,\SU(2)\Big{)}\,.
\ee
%
%
A basis of this space is provided as usual by the spin network states, labeled by a spin on each loop, $k_{\ell=1..N}\,\in\f\N2$, and an intertwiner $\cI$ living in the tensor product $\bigotimes_{\ell=1}^{N} (\cV^{k_{\ell}}\otimes\bar{\cV}^{k_{\ell}})$ and invariant under the action of $\SU(2)$:
\be
\Psi^{\{k_{\ell},\cI\}}\big{(}\{h_{\ell}\}_{\ell=1..N}\big{)}
\,=\,
\la h_{\ell}\,|\,k_{\ell},\cI\ra
\,=\,
\tr\,\Big{[}
\cI\otimes\bigotimes_{\ell=1}^{N}D^{k_{\ell}}(h_{\ell})
\Big{]}\,,
\ee
where the trace is taken over the tensor product $\bigotimes_{\ell=1}^{N} (\cV^{k_{\ell}}\otimes\bar{\cV}^{k_{\ell}})$. To underline that each spin repreentation is doubled and that $\cI$ is an intertwiner between the loops around the vertex, we can dub it a {\it loopy intertwiner}

%
The holonomy operator is the basic gauge-invariant operator of loop quantum gravity. It can shift and increase the spins along the edges on which it acts and so is used in practice as a creation operator. We define the holonomy operators $\hat{\chi}_\ell$ along the loops around the vertex as acting by multiplication on the wave-functions in the group representation:
\be
(\hat{\chi}_\ell \triangleright\Psi)\,(h_1,...,h_N)
\,=\,
\chi_\frac{1}{2}(h_\ell) \Psi(h_1,...,h_N)
\ee
where $\chi_\frac{1}{2}$ is the trace operator in the fundamental two-dimensional representation of $\SU(2)$. We can of course also consider holonomy operators that wrap around several loops around the flower:
\be 
(\hat{\chi}_{i,j,k,l,...}  \triangleright\Psi)\,(h_1,...,h_N)
\,=\,
\chi_\frac{1}{2}(h_i h_j h_k h_l ...) \Psi(h_1,...,h_N)\,,
\ee
where the $i,j,k,l,..$ indices label loops. These operators are obviously still gauge-invariant, and we can further take the inverse or arbitrary powers of each group element.
%
%
There are two remarks we should do about these multi-loop operators. First, they can be decomposed as a composition of single loop operators combining both holonomy operators and grasping operators (action of the $\su(2)$ generators as a quantization of the flux-vectors) by iterating the following 2-loop identity:
\be
\chi_\frac{1}{2}(h_i h_j)
\,=\,
\f12\,\left[
\chi_\frac{1}{2}(h_i)\chi_\frac{1}{2}(h_j)
+\sum_{a=1}^{3}\chi_\frac{1}{2}(h_i\sigma_{a})\chi_\frac{1}{2}(h_j\sigma_{a})
\right]\,,
\ee
where the $\sigma_{a}$'s are the three Pauli matrices, normalized such that their square is equal to the identity matrix.
Second, if the loopy spin network state comes from the gauge fixing of a more complicated graph down to a single vertex, we had chosen a particular maximal tree on that graph to define the gauge-fixing procedure. The loops around the coarse-grained vertex correspond to the edges that didn't belong to the folding tree. Changing the tree actually maps the single loop holonomies onto multi-loop holonomies \cite{Livine:2013gna}. So, from the coarse-graining perspective, there is no special reason to prefer single loops over multi-loop operators.

\subsection{Superposition of number of loops}
\label{sec:loopyprojective}

We would like to allow for an arbitrary number of loops $N$, with possibly an infinite number of loops, and superpositions of number of loops. We will apply the usual projective limit techniques used in loop quantum gravity, as briefly reviewed in  section \ref{ProjectiveLimits}.
We assume here that the little loops are all distinguishable, so we avoid all symmetrization issue. The case of indistinguishable loops will be dealt with in the next section \ref{sec:bosonisation}. We discuss the countable infinity of loops around the vertex, so we can number them using the integers $\N$. The point, as with standard spin networks, is that a state with a spin-0 on an edge does not actually depend on the group element carried by that edge and is thus equivalent to a state on the flower without that edge. Reversing this logic, a state built on a finite number of loops is equivalent to a state with an arbitrary larger number of loops carrying a spin-0 on all the extra edges , which will allow to define it in the projective limit as a state on the flower with an infinite number of loops.

\begin{figure}[h!]
\centering
\begin{tikzpicture}

\coordinate(O1) at (0,0);

\draw (O1) to[in=-30,out=+30,loop,scale=3,rotate=90] (O1) ++(0,1.2) node {$k_2$};
\draw (O1) to[in=-30,out=+30,loop,scale=3,rotate=-90] (O1) ++(0,-1.2) node {$k_5$};
\draw (O1) to[in=-30,out=+30,loop,scale=3,rotate=30] (O1) ++(1,0.7) node {$k_3$};
\draw[dashed] (O1) to[in=-30,out=+30,loop,scale=3,rotate=-30] (O1) ++(1,-0.7) node {$k_4$};
\draw[dashed] (O1) to[in=-30,out=+30,loop,scale=3,rotate=150] (O1) ++(-1,0.7) node {$k_{1}$};
\draw[black!50,dashed] (O1) to[in=-30,out=+30,loop,scale=3,rotate=-150] (O1) ++(-1,-0.7) node {$k_6$};

\draw (O1) node[scale=1] {$\bullet$} ;

\end{tikzpicture}

\caption{We consider a loopy spin network state with a varying number of loops as a superposition of states with support over different loops.}
\label{fig:variousloops}

\end{figure}
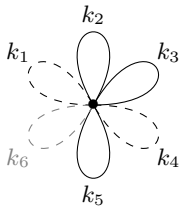

Let us consider  the set $\mathcal{P}_{<\infty}(\mathbb{N})$  of all finite subsets of $\mathbb{N}$. A flower with a finite number of loops corresponds to a finite subset $E\in \,\mathcal{P}_{<\infty}(\mathbb{N})$ of indices labeling its loops. Since we keep the loops distinguishable, we do not identify all the subsets with same cardinality and keep on distinguishing them. We define the Hilbert space of gauge-invariant wave-functions on the flower corresponding to $E$:
\be
\cH_{E}
\,=\,
L^2\,\Big{(}\SU(2)^{E}/\mathrm{Ad}\,\SU(2)\Big{)}\,,
\qquad
\Psi(\{h_{\ell}\}_{\ell\in E})=\Psi(\{gh_{\ell}g^{-1}\}_{\ell\in E})
\quad
\forall g\in\SU(2)\,.
\ee
We would like to consider arbitrary superpositions of states with support on arbitrary subsets $E$ of loops, but we do not wish to brutally consider the direct sum over all $E$'s. We still require cylindrical consistency. Indeed, a function on $\SU(2)^{E}$ which actually  does not depend at all on the loop $\ell_{0}\in E$ can legitimately be considered as a function on $\SU(2)^{E\setminus\ell_{0}}$. We introduce the equivalence relation making this explicit. For two subsets $E\subset F$, and two functions $\Psi$ and $\widetilde{\Psi}$ respectively on  $\SU(2)^{E}$ and  $\SU(2)^{F}$, the two wave-functions are defined as equivalent if:
\be
E\subset F\,,
\quad
\Psi:\SU(2)^{E}\rightarrow\C\,,
\quad
\widetilde{\Psi}:\SU(2)^{F}\rightarrow\C\,,
\qquad
\Psi\sim\widetilde{\Psi}
\quad\Leftrightarrow\quad
\widetilde{\Psi}(\{h_{\ell}\}_{\ell\in F})
\,=\,
{\Psi}(\{h_{\ell}\}_{\ell\in E})\,,
\ee
that is the function $\widetilde{\Psi}$ on the larger set $F$ does not depend on the group elements $h_{\ell}$ for $\ell\in F\setminus E$ and coincides with the function $\Psi$ on the smaller set $E$. More generally, when the two subsets $E$ and $F$ do not contain one or the other, we transite trough their intersection $E\cap F$.

The space of wave-functions in the projective limit  is defined as the union over all subsets $E$ of functions on  $\SU(2)^{E}$, quotiented by this equivalence. We similarly define the projective limit of the integration measure over $\SU(2)$. We use this measure to define the Hilbert space $\cH^{\mathrm{loopy}}$ of states on the flower with an arbitrary number of loops. All the rigorous mathematical definitions and proofs are given in the appendix  \ref{app:ProjectiveLimit}.

The practical way to see this Hilbert space is to use the spin network basis and understand that a loop carrying a spin-0 means that the wave-function actually does not depend on the group element living on that loop. For every state, we can thus reduce its underlying graph to the minimal possible one removing all the loops with trivial dependency. Following this logic, for every subset $E$, we define the space of proper states living on $E$, that is without any spin-0 on its loops. This amounts to removing all possible 0-modes:
\be
\mathcal{H}_{E}^0 
\,=\,
\Bigg{\{}
\Psi\in\cH_{E}
\,:\,
\forall \ell_{0}\in E\,,\,\,\int_{\SU(2)}\mathrm{d}h_{\ell_{0}}\,\Psi =0
\Bigg{\}}\,.
\ee
We can decompose the Hilbert space of states on the subset $E\subset\N$ of loops onto proper states:
\begin{prop}
\label{proper}
The Hilbert space $\cH_{E}$ on loopy intertwiners on the set of loops $E$ decomposes as a direct sum of the Hilbert spaces of proper states with support on every subset of $E$:
\be
\mathcal{H}_{E} \simeq \bigoplus_{F \subset E} \mathcal{H}_{F}^0
\,.
\ee
This isomorphism is realized through the projections $f_{F}=P_{E,F}f\in\cH^{0}_{F}$, acting on wave-functions $f\in\cH_{E}$, defined for an arbitrary subset $F\subset E$:
\be
f_{F}\big{(}
\{h_{\ell}\}_{\ell\in F}
\big{)}
\,=\,
\sum_{\tF\subset  F}
(-1)^{\#\tF}
\int \prod_{\ell\in E\setminus F}\mathrm{d}g_{\ell}
\prod_{\ell\in\tF}\mathrm{d}k_{\ell}\,
f\big{(}
\{h_{\ell}\}_{\ell\in F\setminus \tF},
\{k_{\ell}\}_{\ell\in\tF},
\{g_{\ell}\}_{\ell\in E\setminus F}
\big{)}\,.
\ee
These projections realize a combinatorial transform of the state $f\in\cH_{E}$:
\be
f=\sum_{F\subset E} f_{F}\,,
\qquad
f_{F}\in\cH_{F}^{0}\,,
\qquad
\forall\ell\in F\,,\quad\int \mathrm{d}h_{\ell}\,f_{F}=0\,.
\ee
\end{prop}
This decomposition is straightforward to prove. It will also be crucial in the case of undistinguishable loops and  symmetrized states, as we will see in the next section \ref{sec:bosonisation}.
Then, as we show in the appendix  \ref{app:ProjectiveLimit},  the Hilbert space of loopy spin networks on the flower, with an arbitrary number of distinguishable loops, defined as the projective limit of the Hilbert spaces $\cH_{E}$ is realized as  the direct sum of those spaces of proper states:
\be
\mathcal{H}^{\mathrm{loopy}} \simeq \bigoplus_{F \in \mathcal{P}_{<\infty}(\mathbb{N})} \mathcal{H}_{F}^0\,,
\qquad
\mathcal{H}_{F}^0 = \bigoplus_{j_{\ell \in F} \neq 0, \cI} \mathbb{C} |j_{\ell \in F}, \cI\rangle\,.
\ee

\subsection{Holonomy operators as creation and annihilation operators}

We can revisit the definition of the holonomy operators\footnotemark on our Hilbert space $\cH$ of states with arbitrary number of loops. Let us consider the loop $\ell_{0}\in\N$ and define the corresponding holonomy operator $\hat{\chi}_{\ell_{0}}$. Looking at its action on a state $\Psi$ with finite number of loops living in the Hilbert space $\cH_{E}$, we have two possibilities: either the loop $\ell_{0}$ belongs to the subset $E$ or it doesn't. If the acting loop $\ell_{0}$ is already a loop of our state $\Psi$, then the holonomy operator acts on as before by multiplication:
\footnotetext{
In order to identify a complete set of operators acting on the Hilbert space $\cH^{\mathrm{loopy}}$, we should further consider multi-loops holonomy operators or grasping operators or deformation operators such as $\U(N)$ operators \cite{Borja:2010rc}, but in the first exploration we propose, in this paper, we decide to focus on the single-loop holonomy operator.}
\be
\ell_{0}\in E\,,
\quad
\Psi\in\cH_{E}\,,
\quad
\hat{\chi}_{\ell_{0}}\,\Psi\in\cH_{E}\,,
\qquad
(\hat{\chi}_{\ell_{0}}\Psi)\,(\{h_{\ell}\}_{\ell\in E})
\,=\,
\chi_{\frac{1}{2}}(h_{\ell_{0}})\,\Psi\,(\{h_{\ell}\}_{\ell\in E})\,.
\ee
If the acting loop doesn't belong to the initial subset $E$, we use the cylindrical consistency equivalence relation and we embed both the new loop and the initial loops in a larger graph, say $E\cup\{\ell_{0}\}$,
\be
\ell_{0}\notin E\,,
\quad
\Psi\in\cH_{E}\,,
\quad
\hat{\chi}_{\ell_{0}}\,\Psi\in\cH_{E\cup\{\ell_{0}\}}\,,
\qquad
(\hat{\chi}_{\ell_{0}}\Psi)\,(\{h_{\ell}\}_{\ell\in E})
\,=\,
\chi_{\frac{1}{2}}(h_{\ell_{0}})\,\Psi\,(\{h_{\ell}\}_{\ell\in E})\,,
\ee
with the holonomy operator $\hat{\chi}_{\ell_{0}}$ acting as a creation operator, creating a new loop and curvature excitation.
Since the $\SU(2)$ character $\chi_{\frac{1}{2}}$ is real and bounded by two, $|\chi_{\frac{1}{2}}|\le 2$, we can check that the holonomy operators $\hat{\chi}_{\ell}$ are Hermitian, bounded and thus essentially self-adjoint.

The holonomy operator $\hat{\chi}_{\ell}$  is Hermitian and has a component acting as a creation operator. It must have an annihilation counterpart. The best way to see this explicitly is to write its action on proper states, consistently removing the zero-modes. Indeed, if a loop carries a spin $\f12$, then it gets partly annihilated by the holonomy operator:
\be
\ell_{0}\in E\,,
\quad
\Psi\in\cH_{E}^{0}\,,
\quad
\hat{\chi}_{\ell_{0}}\,\Psi\in\cH_{E}^{0}\oplus\cH_{E\setminus\{\ell_{0}\}}^{0}\,, \nn
\ee
\be
(\hat{\chi}_{\ell_{0}}\Psi)\,(\{h_{\ell}\}_{\ell\in E})
\,=\,
\underset{\in\,\cH_{E}^{0}}{\underbrace{\Bigg{[}\chi_{\frac{1}{2}}(h_{\ell_{0}})\Psi\,(\{h_{\ell}\}_{\ell\in E})
-\int\mathrm{d}h_{\ell_{0}}\chi_{\frac{1}{2}}(h_{\ell_{0}})\Psi\,(\{h_{\ell}\}_{\ell\in E})\Bigg{]}}}
+\underset{\in\,\cH_{E\setminus\{\ell_{0}\}}^{0}}
{\underbrace{\Bigg{[}\int\mathrm{d}h_{\ell_{0}}\chi_{\frac{1}{2}}(h_{\ell_{0}})\Psi\,(\{h_{\ell}\}_{\ell\in E})\Bigg{]}}}
\,.
\ee
This way, it is clear that the holonomy operator $\hat{\chi}_{\ell_{0}}$ creates transition adding and removing one loop. This proper state decomposition of the holonomy operator will become essential when defining it on the Fock space of symmetrized loopy spin networks in the next section \ref{sec:bosonisation}.

\subsection{Imposing BF dynamics on loopy spin networks}

Now that we have describe the whole kinematics of loopy spin networks, with distinguishable loops, we would like to tackle the issue of the dynamics and imposing the Hamiltonian constraints on the Hilbert space of loopy states $\cH^{\mathrm{loopy}}$. The final goal of our proposal is to write the Hamiltonian constraints of loop quantum gravity on  $\cH^{\mathrm{loopy}}$, such that it allows explicitly for local degrees of freedom,  study its renormalization group flow under the coarse-graining and extract its large scale or continuum limit.
Here we will instead describe the much simpler BF dynamics. BF theory can be considered as a consistency check for all attempts and methods to define of dynamics in (loop) quantum gravity\footnotemark. Its physical states are well-known and the Hamiltonian constraints project onto flat connection states. It is furthermore a topological theory with no local degrees of freedom -they are pure gauge.
Finally it has a trivial renormalization flow. Indeed the flatness constraint behaves very nicely under coarse-graining, as illustrated on fig.\ref{fig:Flat}~: considering a spin network graph, imposing the flatness of the connection on all small loops  garanties that  larger loops will be flat too.
All these features must reflect in any proposal for the quantum dynamics of BF theory. 
\footnotetext{
Once the dynamics of BF theory is properly implemented and well under control in a certain framework, one usually use it as a starting point for imposing the true gravity dynamics, with local degrees of freedom, relying on the reformulation of general relativity as a BF theory with constraints. This is for instance the logic behind the construction of spinfoam models for a quantum gravity path integral \cite{Livine:2010zx,Perez:2012wv,Bianchi:2012nk}.
}

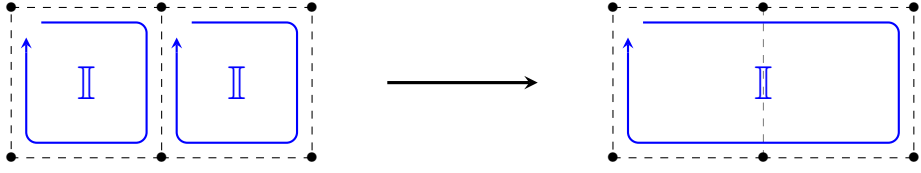
\begin{figure}[h!]

\centering

\begin{tikzpicture}
\coordinate(A) at (0,0);
\coordinate(A1) at (0.2,-0.2);
\coordinate(A2) at (0.4,-0.2);
\coordinate(B) at (2,0);
\coordinate(B1) at (1.8,-0.2);
\coordinate(B2) at (2.2,-0.2);
\coordinate(B3) at (2.4,-0.2);
\coordinate(C) at (4,0);
\coordinate(C1) at (3.8,-0.2);
\coordinate(D) at (0,-2);
\coordinate(D1) at (0.2,-1.8);
\coordinate(D2) at (0.2,-0.6);
\coordinate(D3) at (0.2,-0.4);
\coordinate(E) at (2,-2);
\coordinate(E1) at (1.8,-1.8);
\coordinate(E2) at (2.2,-1.8);
\coordinate(E3) at (2.2,-0.6);
\coordinate(E4) at (2.2,-0.4);
\coordinate(F) at (4,-2);
\coordinate(F1) at (3.8,-1.8);
\coordinate(O1) at (1,-1);
\coordinate(O2) at (3,-1);

\draw[dashed] (A) -- (B) -- (E) -- (D) -- (A);
\draw[dashed] (B) -- (C) -- (F) -- (E);

\draw (A) node {$\bullet$};
\draw (B) node {$\bullet$};
\draw (C) node {$\bullet$};
\draw (D) node {$\bullet$};
\draw (E) node {$\bullet$};
\draw (F) node {$\bullet$};

\draw[blue,rounded corners,thick] (A2) -- (B1) -- (E1) -- (D1) -- (D2);
\draw[blue,thick,->,>=stealth] (D2) -- (D3);
\draw[blue,rounded corners,thick] (B3) -- (C1) -- (F1) -- (E2) -- (E3);
\draw[blue,thick,->,>=stealth] (E3) -- (E4);

\draw[blue] (O1) node[scale=2] {\textbf{$\id$}};
\draw[blue] (O2) node[scale=2] {\textbf{$\id$}};

\draw[->,>=stealth,very thick] (5,-1) -- (7,-1);

\coordinate(P) at (8,0);
\coordinate(A0) at ($(P)+(0,0)$);
\coordinate(A01) at ($(P)+(0.2,-0.2)$);
\coordinate(A02) at ($(P)+(0.4,-0.2)$);
\coordinate(B0) at ($(P)+(2,0)$);
\coordinate(B01) at ($(P)+(1.8,-0.2)$);
\coordinate(B02) at ($(P)+(2.2,-0.2)$);
\coordinate(B03) at ($(P)+(2.4,-0.2)$);
\coordinate(C0) at ($(P)+(4,0)$);
\coordinate(C01) at ($(P)+(3.8,-0.2)$);
\coordinate(D0) at ($(P)+(0,-2)$);
\coordinate(D01) at ($(P)+(0.2,-1.8)$);
\coordinate(D02) at ($(P)+(0.2,-0.6)$);
\coordinate(D03) at ($(P)+(0.2,-0.4)$);
\coordinate(E0) at ($(P)+(2,-2)$);
\coordinate(E01) at ($(P)+(1.8,-1.8)$);
\coordinate(E02) at ($(P)+(2.2,-1.8)$);
\coordinate(E03) at ($(P)+(2.2,-0.6)$);
\coordinate(E04) at ($(P)+(2.2,-0.4)$);
\coordinate(F0) at ($(P)+(4,-2)$);
\coordinate(F01) at ($(P)+(3.8,-1.8)$);
\coordinate(O0) at ($(P)+(2,-1)$);

\draw[dashed] (A0) -- (B0);
\draw[dashed,gray] (B0) -- (E0);
\draw[dashed] (E0) -- (D0) -- (A0);
\draw[dashed] (B0) -- (C0) -- (F0) -- (E0);

\draw (A0) node {$\bullet$};
\draw (B0) node {$\bullet$};
\draw (C0) node {$\bullet$};
\draw (D0) node {$\bullet$};
\draw (E0) node {$\bullet$};
\draw (F0) node {$\bullet$};

\draw[blue,rounded corners,thick] (A02) -- (C01) -- (F01) -- (D01) -- (D02);
\draw[blue,thick,->,>=stealth] (D02) -- (D03);

\draw[blue] (O0) node[scale=2] {\textbf{$\id$}};

\end{tikzpicture}

\caption{In BF theory, holonomies behave very nicely under coarse-graining. If each small loops is flat, large loops are flat too. In other words, the physical state of BF theory is a flat space, which is flat at all scales.}
\label{fig:Flat}

\end{figure}

Considering the full space of loopy spin networks on some arbitrary graph $\Gamma$, we would like the BF Hamiltonian constraints to project onto the flat connection state(s), that is impose flatness around all the loops of the graph $\Gamma$ and also kill all the local excitations represented by the little loops at every vertex. Flatness around the loops of the background graph is the standard result for BF constraints. So here we will focus on the fate of the little loops, that we introduced. To this purpose, it suffices to focus on a single vertex, that is to work on the flower graph.

Considering the flower graph with arbitrary number of loops, as we have defined above, we introduce  the following set of constraints:
\begin{equation}
\forall \ell \in \mathbb{N},
\quad
\big{(}\hat{\chi}_\ell-2\big{)} |\Psi\rangle \,=\, 0\,.
\end{equation}
We impose one constraint for every (possible) loop by imposing that the corresponding holonomy operator saturates its bound and projects on its highest eigenvalue. These constraints all commute with each other.
Let us underline the dual role of Hamiltonian constraints. As first class constraints, we need to solve them and identify their solution space, but they also generate gauge transformations and we need to gauge out their action. Here, the holonomy constraint operators both impose the flatness of the connection, but they also imply that the little loops are pure gauge, so that their action can change the number of loops to arbitrary values.
We will see below that these one-loop holonomy constraints are almost enough to fully constrain the theory to the single flat state on the flower graph.

\bigskip

Let us solve these constraints and consider a loop $\ell_{0}$ and its action of its holonomy operator $\hchi_{\ell_{0}}$ on a wave-function $\Psi\in\cH_{E}$ with support on the finite subset $E\subset\N$ of loops. A first case is when $\ell_{0}\in E$ belongs to the subset, in which case we have a simple functional equation on $\SU(2)^{E}$:
$$
(\hat{\chi}_{\ell_{0}}\Psi)\,(\{h_{\ell}\}_{\ell\in E})
\,=\,
\chi_{\frac{1}{2}}(h_{\ell_{0}})\Psi\,(\{h_{\ell}\}_{\ell\in E})
\,=\,
2\,\Psi\,(\{h_{\ell}\}_{\ell\in E})\,.
$$
The second case is when the considered loop $\ell_{0}\notin E$  doesn't belong to the subset. The holonomy operator $\hat{\chi}_{\ell_{0}}$ then creates a loop, making a transition from $\cH_{E}^{0}$ to the orthogonal space $\cH_{E\cup \{\ell_{0}\}}^{0}$. This illustrates that the flow generated by those Hamiltonian constraints can arbitrarily shift the number of loops and therefore the little loops become pure gauge at the dynamical level in BF theory. This also means that there is no solution to all holonomy constraints with support on a finite subset $E$ and a physical state must have support on all possible loops.

\medskip

To be rigorous, we need to go to the dual space  $(\cH^{\mathrm{loopy}})^*$ and solve the holonomy constraints on the space of distribution defined in the projective limit.  We are looking for a  family of distributions $\vphi_{E}$ on $\SU(2)^{E}$, that is continuous linear forms over smooth functions on  $\SU(2)^{E}$ (see appendix \ref{app:distribution} for a discussion of the definition of distributions over $\SU(2)$). The cylindrical consistency means that their evaluations on two cylindrically equivalent smooth functions must be equal:
$$
\forall E\subset \widetilde{E}\,,\,\,
f_{E}\sim f_{\widetilde{E}}\quad
\Rightarrow
\vphi_{E}(f_{E})=\int_{\SU(2)^{E}}\vphi_{E}f_{E}
\,=\,
\int_{\SU(2)^{\widetilde{E}}}\vphi_{\widetilde{E}}f_{\widetilde{E}}=\vphi_{\widetilde{E}}(f_{\widetilde{E}})\,.
$$
Then the holonomy constraints read:
$$
\forall \ell\in\N\,,
\forall E \ni \ell\,,
\forall f_{E}\in{\cal C}^{\infty}_{\SU(2)^{E}}\,,
\,\,
\int_{\SU(2)^{E}}\ \vphi_{E} (\chi_{\ell}-2)f_{E}
\,=\,0\,,
$$
where we have considered by default that the loop $\ell$ belongs to the wave-function support $E$. Indeed, if $\ell$ didn't belong to $E$, then we could enlarge the subset $E$ to $E\cup \{\ell\}$ by cylindrical consistency and consider both the test function $f$ and the distribution $\vphi$ as living on that larger subset.
Our goal is to show that the unique solution to these equations is the flat state, i.e. that there exists $\lambda\in\C$ such that $\vphi_{E}=\lambda\,\delta^{\otimes E}$:
\be
\forall f_{E}\in{\cal C}^{\infty}_{\SU(2)^{E}}\,,
\vphi_{E}(f_{E})=\lambda \delta_{E}(f_{E})
=
\lambda\int_{\SU(2)^{E}}\prod_{\ell\in E}\delta(h_{\ell}) f_{E}(\{h_{\ell}\}_{\ell\in E})
=
\lambda f_{E}(\id,..,\id)\,.
\ee
Cylindrical consistency simply requires that the factor $\lambda$ does not depend on the subset $E$.
So we are led to solve the holonomy constrain on every finite subset $E$. Thus, let us consider the functional equation on $\SU(2)^{N}$:
\be
\forall 1\le\ell\le N\,,\,\,
\left(\hat{\chi}_\ell - 2\right)\varphi = 0\,,
\ee
where we drop the subset label $E$.

\subsubsection{Holonomy constraint on $\SU(2)$}

Let us start with the one-loop case and solve for distributions $\vphi$ on $\SU(2)$ the equation:
\be
\forall h\in\SU(2)\,,\,\,
\chi_{\f12}(h)\vphi(h)=2\vphi(h)\,.
\ee
Since the character $\chi_{\f12}$ is smooth and reaches its maximum value $2$ at a single point, the identity $\id$, it seems natural that the $\vphi$ must be a distribution peaked at the identity. We therefore expect that the only solution be the $\delta$-distribution on $\SU(2)$, $\vphi=\delta$. However, since the identity is actually an extremum of $\chi_{\f12}$ and that the first derivatives of the character thus vanishes at this point, this equation admit more solutions: the first derivatives of the $\delta$-distribution. This clearly came as a surprise for us.

Let us first assume that $\vphi$ is gauge-invariant, i.e. invariant under conjugation. Its Fourier decomposition on $\SU(2)$ involves only the characters in all spins:
$$
\vphi=\sum_{j\in\f\N2}\vphi_{j}\chi_{j}\,.
$$
As well known, the holonomy constraint leads to a recursion relation on the coefficients $\vphi_{j}$:
\be
\chi_{\f12}\chi_{j}=\chi_{j-\f12}+\chi_{j+\f12}
\quad\Rightarrow\qquad
2 \vphi_{0}=\vphi_{\f12}\,,\quad
2\vphi_{j\ge\f12}=\vphi_{j-\f12}+\vphi_{j+\f12}\,.
\ee
Once the initial condition $\vphi_{0}$ is fixed, these lead to a unique solution:
\be
\vphi_{j}=(2j+1)\vphi_{0}\,,
\qquad
\vphi=\vphi_{0}\sum_{j}(2j+1)\chi_{j}=\vphi_{0}\delta\,.
\ee
%
%
When solving such functional equations in the Fourier basis, one should nevertheless be very careful to work with well-defined distributions. These are characterized by Fourier coefficients $\vphi^{j}$ growing at most polynomially with the spin $j$. This ensures that evaluations $\int f\vphi$ of the distribution $\vphi$ on smooth test functions $f$ are convergent series. The $\delta$-distribution is clearly a good solution. But, as an example, solving for eigenvectors of the holonomy operator associated to (real) eigenvalues (strictly) larger than 2 would lead to exponentially growing Fourier coefficients, which are too divergent to define a proper distribution. The interested reader will find more details in the appendix \ref{app:distribution}.

\medskip

On the space of  functions invariant under conjugation, everything works as expected. Let us now consider the general case dropping the requirement of gauge-invariance. The $\delta$-distribution is obviously still a solution:
\be
\forall f\in\cC^{\infty}_{\SU(2)}\,,
\,\,
\int f\, (\chi_{\f12}-2)\,\delta
=
(\chi_{\f12}(\id)-2)f(\id)
=0\,.
\ee
But, now the first derivatives of the $\delta$-distributions are also  solutions:
\be
\forall f\in\cC^{\infty}_{\SU(2)}\,,
\,\,
\int f\, (\chi_{\f12}-2)\,\pp_{x}\delta
=
\left.-(\pp_{x}\chi_{\f12})\,f-(\chi_{\f12}-2)\,f\right|_{\id}
=0\,,
\label{ppdelta}
\ee
where $x\in\R^{3}$ indicates the direction of the derivative and the derivatives of the character vanish at the identity since it is a extremum.
We remind the reader that the right-derivative $\pp_{x}^{R}$ on $\SU(2)$ is a anti-Hermitian operator ($i\pp$ is Hermitian) defined by the infinitesimal action of the $\su(2)$ generator $\vx\cdot\vJ$ (where the $\vJ$ in the fundamental spin-$\f12$ representation are simply half the Pauli matrices):
\be
\pp_{x}^{R}f(h)=
\lim_{\eps\arr 0}\f{f(h e^{i\eps \vx\cdot \vJ})-f(h)}{\eps}
=f(h\,  x)\,,\quad
\textrm{with}\quad x= \vx\cdot \vJ\,.
\ee
We usually differentiate along the three directions in $\R^{3}\sim\su(2)$ leading to the insertion of the generators $J_{a=1,2,3}$:
\be
\pp_{a}^{R}f(h)=if(h J_{a})\,,
\quad
\pp_{a}^{L}f(h)=if(J_{a}h)\,.
\ee
Acting on the $\delta$-distribution gives the following Fourier decomposition for its derivatives $\pp_{a}^{L}\delta=\pp_{a}^{R}\delta=\pp_{a}\delta$:
\be
\pp_{a}\delta(h)
=i\sum_{j} (2j+1)D^{j}_{nm}(J_{a})\,D^{j}_{mn}(h)\,,
\ee
where we use the Wigner matrices for the group element $h$ and the $\su(2)$ generators.

We can actually generate a whole tower of higher derivative solutions to the holonomy constraints. We simply need to identify the differential operators whose action on the spin-$\f12$ character vanishes at the identity. Thus, at second order, we get five new independent solutions given by the following operators:
\be
\pp_{1}\pp_{2}
\,,\,\,
\pp_{1}\pp_{3}
\,,\,\,
\pp_{2}\pp_{3}
\,,\,\,
(\pp_{1}\pp_{1}-\pp_{2}\pp_{2})
\,,\,\,
(\pp_{1}\pp_{1}-\pp_{3}\pp_{3})
\,,
\ee
%
that is the $\pp_{a}\pp_{b}$ and $(\pp_{a}\pp_{a}-\pp_{b}\pp_{b})$ for $a\ne b$. Following this logic, we will get 7 new independent solutions at third order, and so on with $(2n+1)$ independent differential operators at order $n$, for a total of $(n+1)^{2}$ independent solutions to the holonomy constraints given by differential operators of order at most $n$ acting on the $\delta$-distribution.

Such as in the conjugation-invariant case, it is enlightening to switch to the Fourier decomposition and translate the holonomy constraint into a recursion relation on the Fourier coefficients. The difference is that we had one Fourier coefficient $\vphi^{j}$ for each spin $j$ in the gauge-invariant case while in the general case $\vphi^{j}$ is a $(2j+1)\times(2j+1)$ matrix. Implementing the recursion, we start from spin 0 and work the way up to higher spins. The problem is that the recursion relations determine only $(2j)^{2}$ matrix elements of  $\vphi^{j}$ in terms of the lower spins coefficients, leaving $(2j+1)^{2}-(2j)^{2}=(4j+1)$ matrix elements free to be specified as initial conditions. This leads to an infinite number of solutions to the recursion relations, which reproduces the tower of higher order derivative solutions.
The interested reader will find all of the details on the recursion relations in appendix \ref{app:recursion}. 

\subsubsection{Introducing the Laplacian constraint on $\SU(2)$}
\label{derivativesolution1}

If we work with a single loop, a single petal on the flower, then the wave-function is obviously gauge-invariant and we do not have to deal with these extra solutions to the holonomy constraint\footnotemark.
\footnotetext{
Indeed, since we already proved that the $\delta$-distribution is the only gauge-invariant solution to the holonomy constraint,  all the derivative solutions can not be gauge-invariant. We can also prove this directly.
To get a solution invariant under conjugation, we need to contract the derivative indices together. Now a fundamental theorem on rotational invariants states that all $\SO(3)$-invariant polynomial of $n$ 3d-vectors $\vv_{i=1..n}$ are generated by scalar products $\vv_{i}\cdot\vv_{j}$ and triple products $\vv_{i}\cdot(\vv_{j}\wedge\vv_{k})$. We simply have to check that the Laplacian and triple-grasping of the $\delta$-distribution are not solutions of the holonomy constraints:
$$
\int f (\chi_{\f12}-2)\pp_{a}\pp_{a}\delta
=
\Delta \chi_{\f12}(\id)\,f(\id)
=
\f{-3}2\,f(\id)\ne 0
\,,
\quad
\int f (\chi_{\f12}-2)\eps^{abc}\pp_{a}\pp_{b}\pp_{c}\delta
=
\f{-i^{3}}{2^{3}}\eps^{abc}\chi_{\f12}(\sigma_{c}\sigma_{b}\sigma_{a})\,f(\id)
=
\f{3}2\,f(\id)\ne 0\,.
$$
}
However, as soon as we add external legs attached to the vertex (linking the flower to other vertices in the graph) or add more loops, then we have to find a way to suppress those derivative solutions, in $\pp_{a}\delta$ and so on, which would lead to extra degrees of freedom as some kind of polarized flat states.

Since we want to ensure the full flatness of the holonomy, the most natural proposal is to constrain all the components of the group element living on the loop and not only its trace:
$$
\forall m,n=\pm\f12\,,\,\,
D^{\f12}_{mn}(h) \,\vphi(h)=\delta_{mn}\,\vphi(h)\,.
$$
One can indeed check, both from the differential calculus point of view or the recursion relations in Fourier space, that these equations admit the $\delta$-distribution as unique solutions. We can also go beyond multiplicative operators and insert some differential operators. Then supplementing the trace holonomy constraint with the other constraints $\chi_{\f12}\pp_{a}\,\vphi =2\pp_{a}\vphi$ for $a=1,2,3$ also ensures a unique flat solution. However, these constraints are not gauge-invariant: the constraint operators map wave-functions invariant under conjugation to non-invariant functions.

In order to keep gauge-invariant constraints, we go to the second derivatives and consider the Laplacian operator. Actually, we introduce the right-Laplacian $\Delta\equiv \sum_{a}\pp_{a}^{R}\pp_{a}^{R}$ and a mixed Laplacian operator $\tDelta\equiv \sum_{a}\pp_{a}^{L}\pp_{a}^{R}$, and we propose a new constraint\footnotemark:
\be
\Delta \vphi =\tDelta \vphi\,,
\ee
\footnotetext{
As an example, we can see how $\Delta$ and $\tDelta$ differ through their action on the coupled character $\chi(h_{1}h_{2})$:
$$
\Delta_{1}\chi_{\f12}(h_{1}h_{2})=-\f14\chi_{\f12}(h_{1}\sigma_{a}\sigma_{a}h_{2})=-\f34\chi_{\f12}(h_{1}h_{2})
\,,
\quad
\tDelta_{1}\chi_{\f12}(h_{1}h_{2})=-\f14\chi_{\f12}(\sigma_{a}h_{1}\sigma_{a}h_{2})
=-\f14\,\big{(}2\chi_{\f12}(h_{1})\chi_{\f12}(h_{2})-\chi_{\f12}(h_{1}h_{2})\big{)}\,,
$$
which are of course equal at $h_{1}=\id$.
}
At the classical level, the differential operator $\pp_{a}$ represents the flux vector $X_{a}$: the right derivative represents the flux $\vec{X}^{s}$ at the source of the loop while the left derivative is the flux $\vec{X}^{t}$ at the target of the loop. The target flux is equal to the source flux parallely transported around the loop by the holonomy $h$. The Laplacian constraint is the equality of the scalar product $\vec{X}^{t}\cdot\vec{X}^{s}$ with the squared norm $\vec{X}^{s}\cdot\vec{X}^{s}$ and therefore means that the two flux are equal, $\vec{X}^{s}=\vec{X}^{t}$. This implies the flatness of the group element $h$ (up to the $\U(1)$ stabilizer of the flux vector). 

At the quantum level, the Laplacian constraint turns out to play a different role. It implies the invariance of the wave-function by conjugation:
\be
\Delta \vphi =\tDelta \vphi
\quad\Rightarrow\quad
\forall h,g\in\SU(2)\,,\,\,
\vphi(h)=\vphi(ghg^{-1})\,.
\ee
We rigorously prove this statement in the appendix \ref{app:Laplacian} solving explicitly the recursion relations implied by the Laplacian constraint on the Fourier coefficients of $\vphi$. Another way to understand the relation of the Laplacian constraint to the invariance under conjugation is to think in terms of spin recoupling. Let us call $\vJ^{L,R}$ respectively the $\su(2)$ generators living at the two ends of the loop and defining the left and right derivations. The two Casimirs, given by the two scalar products $\vJ^{L}\cdot\vJ^{L}$ and $\vJ^{R}\cdot\vJ^{R}$, are equal and their (eigen)value is $j(j+1)$ is the loop carries the spin $j$. Then the Laplacian constraint means that their recoupling is trivial:
\be
0=\vJ^{R}\cdot\vJ^{R}-\vJ^{R}\cdot\vJ^{L}=\f12(\vJ^{R}-\vJ^{L})^{2}\,,
\ee
so that the two ends of the loop recouple to the trivial representation, i.e. the spin-0. As illustrated on fig.\ref{fig:Laplacian}, this also allows to show that the Laplacian constraint operator $(\tDelta-\Delta)$ is positive and its spectrum is $k(k+1)/2$ where $k$ is an integer running from 0 to $(2j)$ if the loop carries the spin $j$.
\begin{figure}[h!]

\centering

\begin{tikzpicture}[scale=0.8]
\coordinate(O2) at (-2,0);
\coordinate(O3) at (0,0);

\draw (O3) to[in=-45,out=+45,loop,scale=5] (O3)++(1.6,0) node {$j$};

\draw[red] (O2) -- (O3) node[midway,below]{$k$};
\draw (O3) node {$\bullet$} ++(0.2,0.6) node{$\vJ^{L}$} ++(0,-1.2) node{$\vJ^{R}$};

\end{tikzpicture}

\caption{The left and right derivations respectively  act as graspings at the source and target of the loop, inserting $\su(2)$ generators in the wave-functions. The Laplacian operator $(\tDelta-\Delta)$ then measures the difference between the two scalar products $\vJ^{R}\cdot\vJ^{R}$ and $\vJ^{R}\cdot\vJ^{L}$, or equivalently the Casimir $(\vJ^{R}-\vJ^{L})^{2}/2$ of the recoupling of the spins at the two ends of the loop. Assuming that the loop carries the spin $j$ then recoupling $j$ with itself gives a spin $k$ running from 0 to $(2j)$.
} 
\label{fig:Laplacian}

\end{figure}
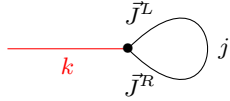
One can also see that the derivatives of the $\delta$-distribution are eigenstates of $(\tDelta-\Delta)$ with non-vanishing eigenvalues. For example, we compute:
\be
\int f(\tDelta-\Delta)\pp_{a}^{R}\delta
\,=\,
i^{3}\sum_{b}
\big{[}
f(J_{a}J_{b}J_{b}) -f(J_{b}J_{a}J_{b})
\big{]}
\,=\,
-if(J_{a})
\,=\,
+\int f \pp_{a}^{R}\delta
\,,
\ee
and so on with higher order differential operators. In particular, the derivative distribution $\pp_{a}\delta$ corresponds to the eigenvalue $k(k+1)/2$ for $k=1$. Higher order derivatives will explore higher eigenvalues.

To conclude, the original holonomy constraint, supplemented with the new Laplacian constraint, acting on functions on $\SU(2)$ admit the $\delta$-distribution as unique solution: the Laplacian constraint imposes invariance under conjugation while the holonomy constraint then imposes the flatness of the group element along the loop.

\begin{prop}
There is a unique solution (up to a numerical factor) as a distribution over $\SU(2)$ to the holonomy and Laplacian constraints:
\be
\left|
\begin{array}{l}
(\hchi-2)\,\vphi=0\\
(\Delta-\tDelta)\,\vphi=0
\end{array}
\right.
\quad\Longrightarrow\quad
\exists \lambda\in\C\,,\,\, \vphi(h)=\lambda\,\delta(h)
\,.
\ee
\end{prop}

Below, we look at the generic case of an arbitrary number of loops. We will show that we can supplement the holonomy constraints around each loop  either with Laplacian constraints for each loop or with  multi-loop holonomy constraints (that still act by multiplication) wrapping around several loops at once.

\subsubsection{Holonomy constraints on $\SU(2)^{N}$ for $N\ge 2$}
\label{derivativesolution2}

We now turn to the holonomy constraints on $\SU(2)^{N}$:
$$
\forall 1\le\ell\le N, \,\,(\hchi_{\ell}-2)\vphi=0\,,
$$
with the requirement of invariance under simultaneous conjugation of all the arguments $h_{\ell}$. Since we do not require the invariance under the individual action of conjugation on each little loop, the gauge invariance is not enough to kill the spurious solution identified above.
As proposed above, we can reach the uniqueness of the physical state by further imposing the Laplacian constraint on each loop:
\be
\forall \ell\in\N\,,\,\,
(\tDelta_{\ell}-\Delta_{\ell})\,\vphi=0\,.
\ee
This now  implies the invariance of the wave-function under the individual action of conjugation on each loop. In terms of spin recoupling, each little loop is linked to the vertex by a spin-0, as illustrated on fig.\ref{fig:0spincoupling},  this effectively trivializes the intertwiner space living at the vertex and the loops can be thought of as decoupled from one another.
The holonomy constraints then impose that the only solution state is the $\delta$-distribution. 
\begin{figure}[h!]
\centering
\begin{tikzpicture}[scale=0.55]

\coordinate(O) at (0,0);
\coordinate(O1) at (2,0);
\coordinate(O2) at (1,1.7);
\coordinate(O3) at (-1,1.7);
\coordinate(O4) at (-2,0);
\coordinate(O5) at (-1,-1.7);
\coordinate(O6) at (1,-1.7);

\draw (O) node[scale=1] {$\bullet$} ;
\draw (O1) node[scale=1] {$\bullet$} ;
\draw (O2) node[scale=1] {$\bullet$} ;
\draw (O3) node[scale=1] {$\bullet$} ;
\draw (O4) node[scale=1] {$\bullet$} ;
\draw (O5) node[scale=1] {$\bullet$} ;
\draw (O6) node[scale=1] {$\bullet$} ;

\draw[red] (O) -- (O1)  ++(-0.9,-0.3) node{$k_{1}$};
\draw[red] (O) -- (O2) node[midway,right]{$k_{2}$};
\draw[red] (O) -- (O3) node[midway,right]{$k_{3}$};
\draw[red] (O) -- (O4) node[midway,above]{$k_{4}$};
\draw[red] (O) -- (O5) ++(0,0.8) node{$k_{5}$};
\draw[red] (O) -- (O6) ++(-0.75,0.8) node{$k_{6}$};

\draw[in=-45,out=45,scale=5] (O1)  to[loop] (O1) ++(0.35,0) node {$j_{1}$};
\draw[in=-45,out=45,scale=5,rotate=60] (O2)  to[loop] (O2) ++(0.35,0) node {$j_{2}$};
\draw[in=-45,out=45,scale=5,rotate=120] (O3)  to[loop] (O3) ++(0.35,0) node {$j_{3}$};
\draw[in=-45,out=45,scale=5,rotate=180] (O4)  to[loop] (O4) ++(0.35,0) node {$j_{4}$};
\draw[in=-45,out=45,scale=5,rotate=-120] (O5)  to[loop] (O5) ++(0.35,0) node {$j_{5}$};
\draw[in=-45,out=45,scale=5,rotate=-60] (O6)  to[loop] (O6) ++(0.35,0) node {$j_{6}$};

\end{tikzpicture}

\caption{ The Laplacian constraint on a loop $\ell$ constraint the spin $j_{\ell}$ carried by the loop to recouple with itself into the trivial representation with vanishing spin $k_{\ell}=0$. Imposing this constraint on every loop, the vertex then recouples a collection of spin-0, the intertwiner is thus trivial and the loops are totally decoupled.}
\label{fig:0spincoupling}

\end{figure}
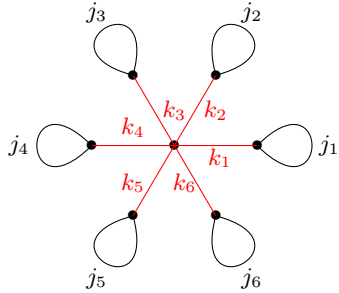

\medskip

Instead of imposing the Laplacian constraints, another way to proceed is to introduce multi-loop holonomy constraints. To prove this, let us start by describing the gauge-invariant derivative solutions to the holonomy constraints. The general structure is as follows. One acts with arbitrary derivatives on the $\delta$-distribution $\prod_{\ell=1}^{N}\delta(h_{\ell})$. Then to ensure invariance under simultaneous conjugation, one must contract all the indices with a $\SO(3)$-invariant tensor $\cI$:
\be
\vphi^{\cI}(\{h_{\ell}\})=
\sum_{\{a_{i}^{\ell}\}_{i=1..n_{\ell}}}
\cI^{a^{1}_{1}..a^{N}_{n_{N}}}\,
\prod_{\ell=1}^{N}\pp_{a^{\ell}_{1}}..\pp_{a^{\ell}_{n_{\ell}}}\delta(h_{\ell})\,,
\ee
where $n_{\ell}$ is the order of the differential operator acting on the loop $\ell$, for an overall order $n=\sum_{\ell}n_{\ell}$, and $\cI$ is a rotational invariant tensor defining the contraction of the differential indices $a$'s, i.e. it is an intertwiner between $n$ spin-1 representations.

To be explicit, for $n=2$ differential insertions, there is a single invariant tensor: $\cI^{ab}=\delta^{ab}$. Either we act with the two derivatives on the same group elements, but then we already know that $\Delta\delta$ is not a solution to the holonomy constraint, or we act on two different loops getting the non-trivial distribution $\sum_{a}\pp_{a}\delta(h_{1})\pp_{a}\delta(h_{2})$ (here we put aside all the other loops, where no differential operator act):
\be
\la \,\sum_{a}\pp_{a}\delta_{1}\pp_{a}\delta_{2}\,|f\ra=
-f(J_{a},J_{a})\,,
\ee
which yields the evaluation $f(J_{a},J_{a})$ of the spin network state obtained by acting with the double grasping $J_{a}\otimes J_{a}$ on the test wave-function $f$. 
We easily check that this provides a solution to the individual one-loop holonomy constraints:
\be
\forall f\in \cC^{\infty}_{\SU(2)^{2}}\,,\,\,
\int f (\chi_{\f12}(h_{1})-2)\,\sum_{a}^{3}\pp_{a}\delta(h_{1})\pp_{a}\delta(h_{2})
=
0\,,
\ee
The double grasping, as shown on fig.\ref{fig:doublegrasping}, couples the two loops. The goal is to suppress such coupling between the two loops in order to get as unique solution the factorized flat state $\delta^{\otimes N}$ where all the loops are entirely decoupled.
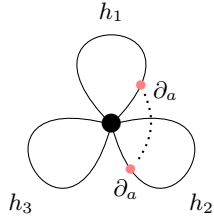
\begin{figure}[h!]
\centering
\begin{tikzpicture}[scale=1.5]

\coordinate(O1) at (0,0);
\draw (O1) node[scale=2] {$\bullet$} ;

\draw (O1)   to[in=-45,out=+45,loop,scale=3,rotate=90] node[very near end](A){}   (O1)++(0,1) node {$h_1$};
\draw (O1) to[in=-45,out=+45,loop,scale=3,rotate=-30] node[very near end](B){} (O1) ++(0.8,-0.7) node {$h_2$};
\draw (O1) to[in=-45,out=+45,loop,scale=3,rotate=-150] (O1) ++(-0.8,-0.7) node {$h_3$};

\draw[dotted,thick] (A) to[bend left] node[very near start,above,right](C){}    (B) ++(-0.1,-0.25) node{$\pp_{a}$};
\draw (A) node[scale=1,red!50] {$\bullet$} ++(0.22,-0.05)node{$\pp_{a}$};
\draw (B) node[scale=1,red!50] {$\bullet$} ;

\end{tikzpicture}

\caption{We act with derivatives $\pp_{a}$ on the group elements $h_{1}$ and $h_{2}$ and contract  the indices, which translates graphically as a double grasping linking the two loops.}
\label{fig:doublegrasping}

\end{figure}
To make the system more rigid, the natural constraint to introduce is a two-loop holonomy constraint, which would kill any correlation between the two loops:
\be
(\hchi_{12}-2)\vphi (h_{1},h_{2})\equiv
(\chi_{\f12}(h_{1}h_{2})-2)\vphi (h_{1},h_{2})
=0\,.
\ee
We check that this two-loop constraint eliminates the coupled solution proposed above:
\be
\int f (\chi_{\f12}(h_{1}h_{2})-2)\,\sum_{a}^{3}\pp_{a}\delta(h_{1})\pp_{a}\delta(h_{2})
=
\left.f\Delta\chi_{\f12}\right|_{\id}
=
\f32\,f(\id,\id)\,\ne0\,.
\ee

For $n=3$ differential insertions, we still have a unique intertwiner, given by the completely antisymmetric tensor $\eps^{abc}$. This corresponds a triple grasping. The three derivatives can all act on the same loop, in which case we do not get a solution of the one-loop holonomy constraint, or they can act on two different loops, in which case it is not a solution of the two-loop holonomy constraints we have just introduced, or they can act on three different loops in which case we need to introduce a three-loop holonomy constraint to discard it:
\be
\int f (\chi_{\f12}(h_{1}h_{2}h_{3})-2)\eps^{abc}\pp_{a}\delta(h_{1})\pp_{b}\delta(h_{2})\pp_{c}\delta(h_{3})
=
-\f i2^{3}f(\id)\eps^{abc}\chi_{\f12}(\sigma_{a}\sigma_{b}\sigma_{c})
=
-\f{3i}2\,f(\id)\,\,\ne0\,.
\ee

For an arbitrary number $n$ of differential insertions acting on the $N$ loops, the grasping will potentially couple the $N$ loops. In order to kill all those coupled solutions, we introduce all multi-loop holonomy constraints:
\be
\forall E\subset \{1,..,N\}\,,
\,\,
\Big{[}
\chi_{\f12}\big{(}
\prod_{\ell\in E}h_{\ell}
\big{)}
-2
\Big{]}\,\vphi=0
\,.
\ee
The ordering of the group elements is important of course for the precise definition of the multi-loop operator but is irrelevant to ensure that the action of the corresponding constraint operator on the coupled derivative distributions does not vanish.
In fact, looking deeper into the structure of $\SO(3)$-invariant tensors, 
a fundamental theorem on rotational invariants states that all $\SO(3)$-invariant polynomial of $n$ 3d-vectors $\vv_{i=1..n}$ are generated by scalar products $\vv_{i}\cdot\vv_{j}$ and triple products $\vv_{i}\cdot(\vv_{j}\wedge\vv_{k})$. This means that we only need the two-loop and three-loop holonomy constraints to ensure that the flat state, defined as the $\delta$-distribution, is the only solution to the Hamiltonian constraints.

\subsubsection{The full Hamiltonian constraints for BF theory on loopy spin networks}
\label{derivativesolution3}

To summarize the implementation of BF theory on loopy spin networks, we have introduced individual holonomy constraints on each little loop around each vertex of the background graph. This is the usual procedure, for instance when constructing spinfoam amplitudes for BF theory from a canonical point of view. Surprisingly, these constraints are not strong enough to fully constraint the theory to the single flat state and kill al the little loop excitations. This can be backtracked to the simple fact that the identity $\id$ is an extremum  of the $\SU(2)$-character $\chi_{\f12}$ and thus the derivative of the character vanishes at that point. As a result, the $\delta$-function on $\SU(2)$ is not the unique solution to the holonomy constraints, but its first derivative are also solutions. While all the solution distributions are peaked on the identity and vanish elsewhere, we are allowed  grasping operators coupling the loops together. To forbid such such coupling and force to have a unique physical state, we have showed that we can supplement the original one-loop holonomy constraints with either one-loop Laplacian constraints or with multi-loop holonomy constraints, which leads us to two proposals for the Hamiltonian constraints for BF theory on loopy intertwiners:
\begin{itemize}
\item We  impose on each loop two gauge-invariant constraints, the holonomy constraint that acts by multiplication and the Laplacian constraint which acts by differentiation:
\be
\forall \ell\,,\quad
\hchi_{\ell}\,\vphi=2\vphi\,,\quad
\Delta_{\ell}\,\vphi=\tDelta_{\ell}\,\vphi\,.
\ee

\item We impose all multi-loop holonomy constraints, requiring not only that the group elements $h_{\ell}$ on each loop $\ell$ is the identity $\id$ but also that all their products remain flat. This means one constraint for each finite subset $E$ of the set of all loops:
\be
\forall E\subset \N,\quad
\hchi_{E}\,\vphi=2\vphi\,,
\qquad
\hchi_{E}\,\vphi (\{h_{\ell}\}_{\ell\in\N})=
\chi_{\f12}\Big{(}\prod_{\ell \in E}h_{\ell}\Big{)}\,\vphi (\{h_{\ell}\}_{\ell\in\N})\,.
\ee
The ordering of the group elements does not matter in order to impose the flatness. These multi-loop constraints kill any correlation or entanglement between the loops. It is actually  sufficient to impose only the two-loop and three-loop holonomy constraints.

\end{itemize}

If we only impose the one-loop holonomy constraints, then the totally flat state defined by the $\delta$-distribution is not the only physical state. We get an infinite-dimensional  space of physical states, obtained by the action of first order grasping operators on the $\delta$-distribution, allowing for non-trivial coupling and correlations between the little loops. It would be interesting to understand the geometrical meaning of those states and if they play a special role\footnotemark in the spinfoam models for BF theory (the Ponzano-Regge and Turaev-Viro models for 3d BF theory and the Crane-Yetter model for 4d BF theory). Maybe those local excitations could provide a first extension of the topological BF theory to a field theory with local degrees of freedom.
\footnotetext{
As an example, we have in mind the recursion relation satisfied by the 6j symbol, which is understood to be the expression of the action of the holonomy operator on the flat state on the tetrahedron graph \cite{Bonzom:2009zd,Bonzom:2011hm,Bonzom:2014bua}. Our results suggest that the double and triple graspings on the 6j symbol might be other solutions to this recursion relation.That would be specially interesting since the triply grasped 6j symbols is understood to be the first order correction of the q-deformed 6j-symbol \cite{Freidel:1998ua}.
}

\medskip

On the other hand, imposing the full set of Hamiltonian constraints proposed above leads to a unique physical state for BF theory: the flat state $\vphi_{\mathrm{BF}}=\delta$.
This physical state  is clearly not normalizable. But since it is unique, it is not a big problem to define the scalar product on this final one-dimension Hilbert space\footnotemark.
\footnotetext{
If we had more solutions to the constraints, for example if we only impose the holonomy constraints without the Laplacian constraints, we would propose to render the distributions normalizable by modulating them by a damping operator $\exp(-\tau\Delta)$ (heat kernel) or similar.  Expanding all the functions and distributions in Fourier modes, all the scalar products would become finite and we could then compute expectation values of operators on the resulting physical Hilbert space, for instance:
$$
\int \delta e^{-\tau\Delta}\delta=\sum_{j}(2j+1)^{2}e^{-\tau j(j+1)} <+\infty
\,,\qquad
\int \delta e^{-\tau\Delta}\pp_{a}\delta=i\sum_{j}(2j+1)^{2}e^{-\tau j(j+1)}\chi_{j}(J_{a})=0\,.
$$
But we haven't investigated this line of research further.
}
The physical scalar product on the initial Hilbert space of loopy spin networks is defined by projecting on this physical state, which amounts at the end of the day to simply evaluate the wave-functions at the identity i.e. on flat connections:
\be
\forall f,\tilde{f}\in\cH^{\mathrm{loopy}}\,,\quad
\la f|\tilde{f}\ra_{\mathrm{phys}}
\,=\,
\la f|\vphi_{BF}\ra\,\la\vphi_{BF}|\tilde{f}\ra
\,=\,
\overline{\la\vphi_{BF}|f\ra}\,\la\vphi_{BF}|\tilde{f}\ra
\,=\,
\overline{f(\id)}\,\tilde{f}(\id)\,.
\ee
As expected, we are left with a single physical state on the flower, the little loops have been projected out and all local degrees of freedom have disappeared.

\medskip

Now that we have checked that loopy spin networks allow for a correct implementation of BF theory's topological dynamics, we would like to later introduce Hamiltonian constraints allowing for local degrees of freedom. We wouldn't want to kill the little loops as happens for BF theory. The goal would be to have dynamics coupling the little loops to the spins living on the links of the background graph, in such a way that it reproduces the propagation of the local geometry excitations of general relativity in a continuum limit. The strategy would be to slightly modify the BF dynamics -``constrain the BF theory''- most likely following the approaches for the dynamics of discrete/twisted geometries \cite{Bonzom:2011hm,Bonzom:2011nv,Bonzom:2013tna} or of EPRL spinfoam models \cite{Engle:2007wy,Geloun:2010vj,Bianchi:2012nk}.

\section{The Fock space of loopy spin networks}
\label{sec:bosonisation}

Up to now, we have introduced the loopy spin network states, on a fixed background graph with an arbitrary number of little loop excitations attached to each vertex. Here we would like to tackle the issue of endowing these local loops with bosonic statistics and define the Fock space of loop spin networks over a background skeleton graph with indistinguishable little loops living at its vertices.

From the perspective of coarse-graining, the little loops represent curvature excitations within the bounded region coarse-grained to a single vertex. Keeping these little loops distinguishable amounts to remembering that they are excitations of different parts of the internal geometry of that region, while fully coarse-graining the region should erase any memory of internal localization and the little loops should be considered as indistinguishable: incoming energy at the vertex would then equally excite any of those loops, irrespective to their a priori different localization on the internal subgraph that we coarse-grained.

In this section, we will thus symmetrize our spin network states over the little loops attached at each vertex. The difficulty reside in the compatibility of the symmetrization with the cylindrical consistency. Indeed, a little loop carrying a spin-0 is considering as a non-existing loop, and vice-versa. We cannot symmetrize these non-existing loops with the other loops carrying non-trivial spins: since we would like to allow for an infinite number of loops, the symmetrization operation would be ill-defined. In the framework of usual quantum field theory, this would be similar to considering particles carrying a vanishing momentum as non-existent. We would have to update the definition of the symmetrization to take this new fact into account. Here, we will show how to systematically subtract the 0-modes components of the loopy spin network states, symmetrize over non-trivial little loops and define an appropriate holonomy operator acting on symmetrized states. A resulting subtlety is that we will be led to distinguish three components of the holonomy operator, that respectively conserves the number of loops, acts as a creation operator adding one little loop or as an annihilation operator removing a loop.

%

\subsection{Bosonics statistics for loops}

We would like to define symmetrized loopy intertwiner states in $\cHl$. A direct way would be to work directly on states with an arbitrary number of loops\footnotemark, but imposing invariance under the symmetry group for a infinite number of loops introduces a lof of technicalities. So we follow a more constructive approach and work with finite number of loops, symmetrize and then allow for  varying number of loops.
\footnotetext{
We would use an extension of the finite symmetry groups $S_{n}$ to the group of permutations of integers which only act non-trivially on a finite subset:
$$
S_\infty = \{f : \mathbb{N} \rightarrow \mathbb{N}, f \textrm{ bijective and }\exists n \in \mathbb{N}, \forall m>n, f(m)=m\}\,.
$$
We would use the canonical action of $S_\infty$  on the Hilbert spaces of loopy spin networks $\cH_{E}$ with finite number of loops:
$$
\sigma : \cH_{E} \rightarrow \cH_{\sigma(E)}\,,
\qquad
(\sigma \triangleright f)(\{h_{e_i}\}) = f(\{h_{\sigma^{-1}(e_i)}\})\,.
$$
%
%
This action is compatible with the cylindrical consistency conditions and naturally extends to the projective limit.
However, requiring invariance of states $|\Psi\rangle \in\cHl$ under permutations, $\forall \sigma \in S_\infty,~\sigma \triangleright |\Psi\rangle = |\Psi\rangle$, only provides non-normalizable states. This forces us to work on the dual space to define symmetrized states and creates unnecessary technicalities for our present purpose.
}

We start from the definition of the loopy states in terms of proper states, $\cHl=\bigoplus_{E\in\mathcal{P}_{<\infty}(\mathbb{N})}\cH_{E}^{0}$. This decomposition has removed all spin-0 and avoids all of the redundancies due to the cylindrical consistency. We can now symmetrize the states. For each number of loops $N$, we consider gauge-invariant wave-functions, symmetric under the exchange of the $N$ loops and such that no loop carries a vanishing spin. The full symmetrized Hilbert space $\cHs$ will then be the direct sum over $N$ of all the finite symmetrized states.

Let us realize this programme explicitly. We start with the Hilbert space $\cHs_{N}$ of wave-functions, $f\in L^{2}(\SU(2)^{\times N})$, gauge-invariant and  symmetrized  on $N$ loops:
\beq
\forall k \in \SU(2),&&
f(h_1,...,h_N) = f(kh_1k^{-1},...,kh_N k^{-1})\,,\nn\\
\forall \sigma \in S_N, &&
f(h_1,...,h_N) = f(h_{\sigma(1)},...,h_{\sigma(N)})\,.
\eeq
We define the subspace of proper states, removing the 0 mode:
\be
\cH_{N}^{0}=
\Bigg{\{}
f\in\cHs_{N}
\,:\,
\int \mathrm{d}h_{1}\,f(h_1,...,h_N) =0
\Bigg{\}}\,.
\ee
We only need to impose one integration condition, since the function is invariant under permutation of its arguments.
We have a simplified version of the decomposition onto proper states given in lemma \ref{proper}:
\begin{lemma}
\label{propersym}
The Hilbert space of symmetrized states on $N$ loops decomposes as a direct sum of the Hilbert spaces of proper symmetrized states on at most $N$ loops:
\be
\cHs_{N}=\bigoplus_{n=0}^{N}\cH_{n}^{0}\,.
\ee
This isomorphism is realized through a combinatorial transform of the wave-functions:
\be
\label{symresum}
\forall f\in\cHs_{N}\,,\quad
f=\sum_{n=0}^{N}\,\,
\sum_{1\le i_{1}<..<i_{n}\le N}f_{n}
\big{(}
h_{i_{1}},..,h_{i_{n}}
\big{)}
\,,\quad
f_{n}\in\cH_{n}^{0}
\ee
\be
f_{n}(h_{1},..,h_{n})
\,=\,
\sum_{m=0}^{n}
(-1)^{n-m}
\int\prod_{i=m+1}^{n}\mathrm{d}k_{i}\,\prod_{i=n+1}^{N}\mathrm{d}g_{i}\,
\sum_{1\le i_{1}<..<i_{m}\le n}
f\big{(}
h_{i_{1}},..,h_{i_{m}},k_{m+1},..,k_{n-m},g_{n+1},..,g_{N}
\big{)}\,.
\ee
The scalar product is given by the integration with respect to the Haar measure. The integral condition (absence of 0-mode) for the proper states implies that two proper states with different support are immediately orthogonal:
\be
\label{scalarN}
\forall f,\tf\,\in\cHs_{N}\,,\quad
\la f|\tf\ra_{N}
\,=\,
\int \prod_{i=1}^{N}\mathrm{d}h_{i}\,\overline{f(h_{1},..,h_{N})}\,\tf(h_{1},..,h_{N})
\,=\,
\sum_{n=0}^{N} \binomial{N}{n}\,\la f_{n}|\tf_{n}\ra_{n}\,.
\ee
\end{lemma}
In the resummation formula \eqref{symresum} above, the sum over labels  $1\le i_{1}<..<i_{n}\le N$ corresponds to the sum over all subsets with $n$ elements -or $n$-uplets- among the first $N$ integers $\{1,..,N\}$. And the injection of the proper state Hilbert space $\cH_{n}^{0}$ in the larger symmetrized space $\cHs_{N}$ requires this sum over all possible choices of $n$-uplets. This leads to the binomial coefficient in the scalar product formula \eqref{scalarN}. This is a clear remnant of having distinguishable loops. Once the little loops are assumed to be bosonic and fully indistinguishable, there is no reason to distinguish a state $f_{n}(h_{a_{1}},..,h_{a_{n}})$ from $f_{n}(h_{b_{1}},..,h_{b_{n}})$ with different choice of $n$-uplets. 

Therefore, to define bosonic states in the projective limit $N\arr \infty$, we will keep the decomposition as a direct sum of vector spaces $\cHs_{N}=\bigoplus_{n=0}^{N}\cH_{n}^{0}$ defining the tower of symmetrized states, but we will modify the scalar product to remove its dependence on $N$ and make it compatible with the projective limit:
\be
\label{bosonicN}
\la f|\tf\ra_{N}^{\mathrm{bosonic}}
\,=\,
\sum_{n=0}^{N} \la f_{n}|\tf_{n}\ra_{n}\,.
\ee
This is achieved by simply including the symmetrizing factor in the definition of the injection $I_{N,N+1}:\,\cHs_{N}\hookrightarrow\cHs_{N+1}$ of wave-functions of $N$ loops seen as wave-functions of $(N+1)$ loops:
\be
f\,\in\cHs_{N}
\,\mapsto\,
I_{N,N+1}f\,\in\cHs_{N+1}\,,
\qquad
\big{(}I_{N,N+1}f\big{)}(h_{1},..,h_{N+1})
\,=\,
\f1{N+1}\,
\sum_{i=1}^{N+1}f(h_{1},..,\widehat{h_{i}},..,h_{N+1})\,,
\ee
where the element $\widehat{h_{i}}$ means that we omit it from the list of arguments. This generalizes to injections $\cHs_{N}\hookrightarrow\cHs_{N+p}$ using the binomial coefficients:
\be
f\,\in\cHs_{N}
\,\mapsto\,
I_{N,N+p}f\,\in\cHs_{N+p}\,,
\qquad
\big{(}I_{N,N+p}f\big{)}(h_{1},..,h_{N+p})
\,=\,
\binomial{N+p}{N}^{-1}\,
\sum_{1\le i_{1}<..<i_{N}\le N+p}f(h_{i_{1}},..,h_{i_{N}})\,.
\ee
These factors compensate the binomial factors from the scalar product formula \eqref{scalarN}.
As we will see a little bit further, this scalar product $\la f|\tf\ra_{N}^{\mathrm{bosonic}}$  on symmetric states is the one which makes the holonomy operator(s) Hermitian.
Then we can define the Fock space of loopy spin networks with bosonic little loop excitations.
\begin{definition}
The full Fock space of symmetrized loop states is defined as the projective limit of the Hilbert spaces $\cHs_{N}$, endowed with the bosonic scalar product \eqref{bosonicN}, which amounts to the direct sum of the spaces of proper states:
\be
\cHs\,\equiv\,\bigoplus_{N\in\N}\cH_{N}^{0}\,,\qquad
\forall f,\tf\,\in\cHs\,,\quad
\la f|\tf\ra
\,=\,
\sum_{N=0}^{\infty} \la f_{N}|\tf_{N}\ra\,.
\ee
\end{definition}
This describes bosonic excitations of the holonomy at each vertex of the base graph for the loopy spin network states.
This Fock space of little loops at a vertex have states for an arbitrary number of indistinguishable loops, that can be created and annihilated, each of them carrying a spin $j_{\ell}\in\N/2$ encoding the corresponding excitation of the geometry (area quanta). The spin carried by a loop is similar to the momentum carried by a particle. One must nevertheless keep in mind two differences with the usual Fock space construction used in standard quantum field theory:
\begin{itemize}

\item {\it 0-modes are pure gauge:}
First, we have implemented explicitly the cylindrical consistency requirement in the definition of the Fock space of loopy spin networks. A little loop carrying a spin-0 is identified to a vanishing excitation, i.e. a non-existing loop, so we have systematically removed them using proper states. This is similar to removing particle states with 0-momentum.

\item {\it Non-trivial intertwiner structure:}
Second, for a given number of loops carrying some given spins, the loopy spin network state still contains more information: the state requires the data of an intertwiner linking all these little loops together (and to the external legs of the vertex). Each time we create a loop, the intertwiner space at the vertex is further enlarged. This extra structure implies that factorized states do not constitute a basis of the Fock space of loopy intertwiners.

\end{itemize}
After describing factorized states below, we will define the holonomy operators acting on the Fock space of symmetrized states and show how they shift the number of loops and become the basic creation and annihilation operators.
%

\subsection{Factorized states and $\delta$-distribution}

It is interesting to check how factorized state, with no correlations between the loops, get decomposed onto proper states. Let us consider a integrable function $\vphi$ on $\SU(2)$. We assume it to be invariant under conjugation, so that it can be decomposed over the $\SU(2)$-characters for all spins:
\be
\forall h,g\,\in\SU(2)\,,\quad
\vphi(h)=\vphi(ghg^{-1})\,,
\qquad
\vphi(h)=\sum_{j\in\f\N2}\vphi_{j}\chi_{j}(h)\,.
\ee
We consider the $N$-loop symmetric  state $\vphi^{\otimes N}$ and check its proper state decompositon by the combinatorial formula given above in the lemma \ref{propersym}:
\beq
\vphi^{\otimes N}_{0}&=&\left(\int\vphi\right)^{N}=\vphi_{0}^{N}\,, \\
\vphi^{\otimes N}_{1}(h)&=&
\vphi_{0}^{N-1}
\big{[}
\vphi(h)-\vphi_{0}
\big{]}
\,,\nn\\
\vphi^{\otimes N}_{2}(h_{1},h_{2})&=&
\vphi_{0}^{N-2}
\big{[}
\vphi(h_{1})\vphi(h_{2})
-\vphi_{0}\vphi(h_{1})
-\vphi_{0}\vphi(h_{2})
+\vphi_{0}^{2}
\big{]}
\,=\,
\vphi_{0}^{N-2}
\big{[}
\vphi(h_{1})-\vphi_{0}
\big{]}
\big{[}
\vphi(h_{2})-\vphi_{0}
\big{]}
\,,\nn\\
\vphi^{\otimes N}_{n}(h_{1},..,h_{n})&=&
\vphi_{0}^{N-n}
\,\sum_{m=0}^{n}
(-1)^{n-m}\sum_{1\le i_{1}<..<i_{m}\le n}
\vphi_{0}^{n-m}\,\vphi(h_{i_{1}})..\vphi(h_{i_{m}})
=
\vphi_{0}^{N-n}\prod_{i=1}^{n}
\big{[}
\vphi(h_{i})-\vphi_{0}
\big{]}
\,.\nn
\eeq
We can check the scalar product formula \eqref{scalarN}:
$$
\Bigg{(}\int |\vphi|^2\Bigg{)}^{N}
\,=\,
\Bigg{(} |\vphi_{0}|^2+\int\big{|}\vphi- \vphi_{0}\big{|}^2\Bigg{)}^{N}
\,=\,
\sum_{n}^{N}\binomial{N}{n}|\vphi|_{0}^{2(N-n)}\,\Bigg{(}\int\big{|}\vphi- \vphi_{0}\big{|}^{2}\Bigg{)}^n
\,=\,
\sum_{n}^{N}\binomial{N}{n} \int \big{|}\vphi_{n}^{\otimes N}\big{|}^2
$$

First, we notice that the proper state projections are still factorized. We are merely consistently removing the spin-0 component from all the loops, without creating any correlation during  the process. Second, if we normalize the one-loop wave-function $\vphi_{0}=\int \vphi=1$, then the projections of the factorized state do not depend anymore on the number of loops $N$ and we can take the projective limit. We can define a factorized state $\vphi^{\otimes \infty}$ with support on an infinite number of loops by taking the limit $N\arr\infty$. We define its components, dropping the useless $\infty$ label:
\be
\vphi=1+\widetilde{\vphi}\,,\quad
\int\widetilde{\vphi}=0\,,\qquad
\vphi_{0}=1\,,
\quad
\vphi_{n}=\widetilde{\vphi}^{\otimes n}\,.
\ee
We can for instance apply this to the $\delta$-distribution and define the flat holonomy state in our Fock space of symmetrized states:
\be
\label{delta-def} 
\widetilde{\delta}=\delta-1
=\sum_{j\ne 0}(2j+1)\chi_{j}
\,,\qquad
\delta_{0}=1\,,\quad
\delta_{n}=\widetilde{\delta}^{\otimes n}\,.
\ee

\subsection{Holonomy operator on symmetrized states}

Now that we have defined the Fock space of loopy intertwiners, we would like to have the basic  operators creating and annihilating loop excitations. This is naturally achieved by the (one-loop) holonomy operator. We start, as in the case of distinguishable loops, with multiplying wave-functions by the spin-$\f12$ character $\chi(h_{\ell})$ applied to the group element $h_{\ell}$ living on a little loop $\ell$. Then we will to distinguish three cases: the loop $\ell$ does not belong the existing loops and the operator creates a new loop, or the loop $\ell$ is already excited, in which case it can act on a spin-$\f12$ excitation and actually annihilate the loop, or the operator will generically act on all other spin excitations by simple multiplication. This leads us to defining three components of the holonomy operator $\hat{\chi}$ acting on symmetrized states:
\begin{definition}
\label{AAB-def}
We define three operators $A,\tA,B$ acting on the Fock space of symmetrized loopy intertwiners $\cHs$. They act on an arbitrary state $(f_{N})_{N\in\N}$ as:
\be
(Af)_{N}(h_{1},..,h_{N})
\,=\,
\int \mathrm{d}k\,\chi_{\f12}(k)\,f_{N+1}(h_{1},..,h_{N},k)\,,
\ee
\be
(Bf)_{0}=2f_{0}\,,
\quad
\forall N>0\,,\,\,
(Bf)_{N}(h_{1},..,h_{N})
\,=\,
\f1N\sum_{i=1}^{N}\Bigg{[}
\chi_{\f12}(h_{i})f_{N}(h_{1},..,h_{N})
-\int \mathrm{d}k_{i}\,\chi_{\f12}(k_{i})\,f_{N}(h_{1},..,k_{i},..,h_{N})
\Bigg{]}
\ee
\be
(\tA f)_{0}=0\,,
\quad
\forall N>0\,,\,\,
(\tA f)_{N}(h_{1},..,h_{N})
\,=\,
\f1{N}\sum_{i=1}^{N}\chi_{\f12}(h_{i})f_{N-1}(h_{1},..,\widehat{h_{i}},..,h_{N})
\ee
The operator $B$ is the usual action of the holonomy operator by multiplication by the character up to the subtraction of the resulting spin-0 component.
The operator $A$ is the annihilation operator, removing one loop, while the operator $\tA$ creates a new loop. We have the following relations on $\cHs$:
\be
\tA=A^{\dagger}\,,\qquad
B=B^{\dagger}
\,.
\ee
Finally the one-loop holonomy operator for spin $\f12$  is defined as the sum of these three components and is self-adjoint:
\be
\hchi_{\f12}\,\equiv\,\f12\big{(}A+\tA+B\big{)}\,.
\ee
\end{definition}
The convention $(Bf)_{0}=2f_{0}$ follows the logic that the 0-component $f_{0}$, with no loop, represents by default a flat holonomy and thus should be multiplied by $\chi_{\f12}(\id)=2$.
Here is the proof for the Hermicity relations:
\begin{proof}
We compare the action of $A$ and $\tA$:
$$
\la \phi|A\psi\ra
\,=\,
\sum_{N\in\N}
\int [\mathrm{d}h_{i}]_{i=1}^{N}\mathrm{d}k\,
\chi_{\f12}(k)\,\overline{\phi_{N}(h_{1},..,h_{N})}\,\psi_{N+1}(h_{1},..,h_{N},k)\,,
$$
$$
\la \tA \phi|\psi\ra
\,=\,
\sum_{N>0}
\int [\mathrm{d}h_{i}]_{i=1}^{N}
\f1N\sum_{i=1}^{N}\chi_{\f12}(h_{i})\,
\overline{\phi_{N-1}(h_{1},..,\widehat{h_{i}},..,h_{N})}\,
\psi_{N}(h_{1},..,h_{N})\,.
$$
We shift the sum over $N$ in $\la \tA \phi\,|\,\psi\ra$ and we use the invariance of $\psi_{N}$ under permutation of its arguments to conlude that these two expressions coincides, $\la \phi|A\psi\ra=\la \tA \phi|\psi\ra$. As for the operator $B$, we compute:
\beq
\la \phi|B\psi\ra
&=&
2\overline{\phi_{0}}\,\psi_{0}
\,+\,
\sum_{N>0}
\int [\mathrm{d}h_{i}]_{i=1}^{N}
\f1N\sum_{i}^{N}
\chi_{\f12}(h_{i})\overline{\phi_{N}(h_{1},..,h_{N})}\,\psi_{N}(h_{1},..,h_{N})\nn\\
&&
-\,\int [\mathrm{d}h_{i}]_{i=1}^{N}
\f1N\sum_{i}^{N}
\int \mathrm{d}k_{i}\,\chi_{\f12}(k_{i})\overline{\phi_{N}(h_{1},..,h_{i},..,h_{N})}\,\psi_{N}(h_{1},..,k_{i},..,h_{N})\,.\nn
\eeq
The last term vanishes due to the absence of 0-mode, $\int \mathrm{d}h_{i}\,\phi_{N}=0$. This ensures that $\la \phi|B\psi\ra=\la B\phi|\psi\ra$ and thus $B$ is a Hermitian operator.
\end{proof}

To ensure that the operators $A$ and $B$ are well-defined and that the holonomy operator $\hchi_{\f12}$ is  self-adjoint, it is enough to check that it is bounded. And we show below that it is indeed bounded by 2 as in the usual framework.
\begin{lemma}
The two parts of the holonomy operators are both bounded by 2, that is for all states $\phi\in\cHs$, we have the two inequalities:
\be
|\la \phi|\,(A+\tA)\,|\phi\ra|\,\le\,2\la \phi|\phi\ra
\,,\qquad
|\la \phi|B|\phi\ra|\,\le\,2\la \phi|\phi\ra\,.
\ee
This ensures that they are both self-adjoint. The holonomy operator $\hchi_{\f12}$ is then also bounded by 2 and self-adjoint.
\end{lemma}

\begin{proof}
Let us start with the operator $B$. The analysis is simpler since it doesn't shift the number of loops:
$$
\la \phi|B|\phi\ra=2|\phi_{0}|^{2}+\sum_{N>0}B_{N}\,,
\quad
B_{N}=
\f1N\sum_{i}^{N}
\int [\mathrm{d}h_{i}]_{i=1}^{N}\chi_{\f12}(h_{i})\overline{\phi_{N}}(h_{1},..,h_{N})\,\phi_{N}(h_{1},..,h_{N})\,.
$$
The extra term in the action of $B$ on the state $\phi$ vanishes as earlier due to the integral condition on proper states, $\int \mathrm{d}h_{i}\,\phi_{N}=0$ for all $i$'s. Since the character $\chi_{\f12}$ is bounded by 2, it is direct to conclude:
$$
|B_{N}|\le \f2N\sum_{i}^{N}\int |\phi_{N}|^{2}= 2\int |\phi_{N}|^{2}\,,
\quad
\la \phi|B|\phi\ra
\le 2|\phi_{0}|^{2}+2\sum_{N>0}\int |\phi_{N}|^{2}=2\la \phi|\phi\ra\,.
$$
We can proceed similarly with the operator $A+\tA$:
$$
\la \phi|\,(A+\tA)\,|\phi\ra=
\la \phi|A|\phi\ra+\la \phi|A^{\dagger}|\phi\ra=\la \phi|A|\phi\ra+\overline{\la \phi|A|\phi\ra}\,,\quad
|\la \phi|\,(A+\tA)\,|\phi\ra|\le 2 |\la \phi|A|\phi\ra|\,,
$$
\be
\la \phi|A|\phi\ra=\sum_{N}A_{N}\,,\quad
A_{N}=\int \prod_{i=1}^{N}\mathrm{d}h_{i}\,\mathrm{d}k\,
\chi_{\f12}(k)\,\overline{\phi_{N}}(h_{1},..,h_{N})\,\phi_{N+1}(h_{1},..,h_{N},k)
\ee
As long as the components $\phi_{N}$'s are square-integrable, we can use the Cauchy-Schwarz inequality to bound these integrals:
$$
\left|
A_{N}
\right|
\le
\sqrt{\int  \prod_{i}^{N}\mathrm{d}h_{i}\,\mathrm{d}k\,\chi_{\f12}(k)^{2}\,\big{|}\phi_{N}(h_{1},..,h_{N})\big{|}^{2}}\,
\sqrt{\int  \prod_{i}^{N}\mathrm{d}h_{i}\,\mathrm{d}k\,\big{|}\phi_{N+1}(h_{1},..,h_{N},k)\big{|}^{2}}\,.
$$
We use that the $\SU(2)$ character is normalized, $\int \chi_{\f12}^{2}=1$,  and then apply the inequality bounding a product $ab\le (a^{2}+b^{2})/2$:
\be
\left|
A_{N}
\right|
\le
\f12{\int  \prod_{i}^{N}\mathrm{d}h_{i}\,\big{|}\phi_{N}(h_{1},..,h_{N})\big{|}^{2}}
+\f12{\int  \prod_{i}^{N}\mathrm{d}h_{i}\,\mathrm{d}k\,\big{|}\phi_{N+1}(h_{1},..,h_{N},k)\big{|}^{2}}
=\f12\,\Big{[}\la \phi_{N}|\phi_{N}\ra+\la \phi_{N+1}|\phi_{N+1}\ra\Big{]}\,.
\ee
Summing over $N\in\N$, this allows us to conclude that $|\la \phi|A|\phi\ra|\le \la\phi|\phi\ra-\f12|\phi_{0}|^{2}\le \la\phi|\phi\ra$ and thus reproduces the expected bound $|\la \phi|\,(A+\tA)\,|\phi\ra|^{2}\le 2\la\phi|\phi\ra$.
\end{proof}

Although we consider the holonomy operator $\hchi_{\f12}$ to be the averaged sum of the two self-adjoint components $(A+A^{\dagger})$ and $B$, each of these is a legitimate operator in itself. We could push this logic further and state that we have defined two different holonomy operators, 
on the one hand, a holonomy operator $(A+A^{\dagger})$ that acts as a ladder operator, creating and annihilating loops, and on the other hand, a holonomy operator $B$ which acts as spin shifts on existing loops (i.e. modifies the area quanta carried by each loop).

The important consistency check, which will be essential for the analysis of the BF theory dynamics, is that the flat state is an eigenvector of the one-loop holonomy operator:
\begin{prop}
\label{flat-prop1}
The flat state $\delta$, defined in \eqref{delta-def} by its proper state projections, $\delta_{0}=1$ and $\delta_{N}=(\delta-1)^{\otimes N}$ for $N\ge 1$, is an eigenvector of the spin-$\f12$ one-loop holonomy operator  $\hchi_{\f12}$ with the highest eigenvalue on $\cHs$:
\be
\hchi_{\f12}\,
|\delta\ra
\,=\, 2\,|\delta\ra
\,.
\ee
This distributional flat state is also an eigenvector of the loop annihilation operator $A$ and of the loop creation operator $(B+A^{\dagger})$:
\be
A|\delta\ra
\,=\,
(B+A^{\dagger})|\delta\ra
\,=\,
2\,|\delta\ra
\,.
\ee
\end{prop}
\begin{proof}
We compute the action of the three parts of the holonomy operators acting on the flat state defined explicitly as
$$
\delta_{0}=1\,,\quad
\delta_{N}(h_{1},..,h_{N})
=\prod_{i}^{N}\big{[}\delta(h_{i})-1\big{]}\,.
$$
For the no-loop component, we get:
$$
(A\delta)_{0}=\int \mathrm{d}k\,\chi_{\f12}(k)\,\big{[}\delta(k)-1\big{]}=2
\,,\quad
(A^{\dagger}\delta)_{0}=0
\,,\quad
(B\delta)_{0}=2\,,
$$
while we compute for all other components:
\beq
(A\delta)_{N}(h_{1},..,h_{N})&=& 2\prod_{i}^{N}\big{[}\delta(h_{i})-1\big{]}=2\delta_{N}(h_{1},..,h_{N})
\\
(A^{\dagger}\delta)_{N}(h_{1},..,h_{N})&=& 
\f1N\sum_{i}^{N}\chi_{\f12}(h_{i})\,\prod_{\ell\ne i}^{N}\big{[}\delta(h_{\ell})-1\big{]}
\\
(B\delta)_{N}(h_{1},..,h_{N})&=&
2\prod_{i}^{N}\big{[}\delta(h_{i})-1\big{]}
-\f1N\sum_{i}^{N}\chi_{\f12}(h_{i})\,\prod_{\ell\ne i}^{N}\big{[}\delta(h_{\ell})-1\big{]}
\eeq
Adding these three contributions, we get as expected for all number of loops $(\hchi_{\f12}\,\delta)_{N}=2\delta_{N}$.

\end{proof}

\smallskip

Since we have three operators built in the holonomy operator, it is natural to investigate their commutation algebra. It actually involves higher spin operators. We generalize the definition of the operators $A$, $A^{\dagger}$ and $B$ to arbitrary spins: one simply replaces in their definition \ref{AAB-def} the character in the fundamental representation $\chi_{\f12}$ by the higher spin character $\chi_{j}$ for any $j\in\N^{*}/2$, thus producing new operators $A_{j}$ annihilating a loop excitation of spin $j$,  $A^{\dagger}_{j}$ creating a new loop carrying a spin $j$ and $B_{j}$ acting with a spin $j$ excitation on an existing loop.

Then acting on a an arbitrary state $f$, we get:
$$
(ABf)_{N}=\f{N}{N+1}(BAf)_{N}+\f1{N+1}(A_{1}f)_{N}
\,,\quad
(BA^{\dagger}f)_{N}=\f{N-1}{N}(A^{\dagger}Bf)_{N}+\f1{N}(A_{1}^{\dagger}f)_{N}\,,
$$
$$
(AA^{\dagger}f)_{N}
=\f{N}{N+1}(A^{\dagger}Af)_{N}+\f1{N+1}f_{N}\,.
$$
Remembering that the number of loops $N$ is not constant on the Fock space of loopy intertwiners and should be treated as an operator $\hat{N}$, these translate into commutation relations, being careful about the operator ordering:
\be
\hN B = B\hN
\,,\quad
(\hN+1)A=A\hN
\,,\quad
\hN A^{\dagger}= A^{\dagger}(\hN+1)\,,
\ee
\be
\label{commAB1}
AB\hN=\hN BA+A_{1}
\,,\quad
\hN B A^{\dagger}=A^{\dagger}B\hN+A_{1}^{\dagger}
\,,\quad
A\hN A^{\dagger}=\id+A^{\dagger}A\hN=\id+\hN A^{\dagger}A
\,.
\ee
These generalize to the whole tower of higher spin operators, for all spins $a,b\in\f{\N^{*}}2$:
\be
\label{commAB2}
A_{a}B_{b}\hN=\hN B_{b}A_{a}+A_{a\otimes b}
\,,\quad
\hN B_{b} A_{a}^{\dagger}=A_{a}^{\dagger}B_{b}\hN+A_{a\otimes b}^{\dagger}
\,,\quad
A_{a}\hN A_{b}^{\dagger}=\delta_{ab}\id+A_{b}^{\dagger}A_{a}\hN=\delta_{ab}\id+\hN A_{b}^{\dagger}A_{a}
\,,
\ee
where we use the (natural) convention of the tensor product of spins for the annihilation operator:
\be
A_{a\otimes b}\,\equiv\,
\sum_{c=|a-b|}^{a+b}A_{c}\,.
\ee
We give the last commutation relation:
\be
\hN\,[B_{a},B_{b}]
\,=\,
A^{\dagger}_{b}A_{a}-A^{\dagger}_{a}A_{b}\,.
\ee


We can combine these higher spin creation and annihilation operators to define a spin-$j$ holonomy operator $\hchi_{j}$ as the average sum of those operators as for the fundamental representation:
\be
\hchi_{j}=\f12\,\big{(}
A_{j}+A^{\dagger}_{j}+B_{j}
\big{)}\,.
\ee
This rather natural definition unfortunately doesn't ensure that the operators $\hchi_{j}$'s for different spins $j$'s commute with each other. Using the algebra computed above, the commutator of two holonomy opertaors $\hchi_{a}$ and $\hchi_{b}$ actually looks like a mess.
Nevertheless we can simplify the expressions by introducing suitable number of loops factors. Inserting the operator $\hN$ in the character, we find:
\be
\big{[}
\hN\hchi_{a},\hN\hchi_{b}
\big{]}
=\f{\hN}2\,(A^{\dagger}_{b}A_{a}-A^{\dagger}_{a}A_{b})
=\f{\hN^{2}}2\,[B_{a},B_{b}]\,.
\ee
This combination $\hN\hchi_{a}$ isn't Hermitian, but this can be easily remedied to by considering $\sqrt{\hN}\hchi_{a}\sqrt{\hN}$ instead. This commutator doesn't vanish, but we can easily find other combinations of the creation and annihilation operators that do:
\be
\big{[}
\hN(B_{a}+A_{a}^{\dagger}-A_{a}),\hN(B_{b}+A_{b}^{\dagger}-A_{b})
\big{]}
\,=\,
0
\,.
\ee
This suggests using the operators $A_{a}$ and $(B_{a}+A^{\dagger}_{a})$ as more fundamental as the holonomy operators. Although they are not Hermitian, the flat state is an eigenvector of both operators and we will exploit this fact in defining flatness constraints for BF theory in the  following section \ref{BFsym}.

\medskip

The other way to proceed to defining higher spin holonomy operators is to reproduce the classical algebra of the $\SU(2)$ characters. For instance, a spin-1 is obtained from the tensor product of two spin-$\f12$ representations:
$$
\chi_{1}(h)=\chi_{\f12}(h)^{2}-1\,.
$$
We propose to promote these relations to the quantum level:
\be
\hchi_{1}^{\,\mathrm{full}}
\,\equiv\,
\hchi_{\f12}^{\,2}-1
=
\f14\Big{[}
A^{2}+AB+BA+AA^{\dagger}+A^{\dagger}A+B^{2}+A^{\dagger}B+BA^{\dagger}+A^{\dagger}{}^{2}
\Big{]}-1
\,.
\ee
This new spin-1 holonomy operator is already a multi-loop operator: it has a component $A^{2}$ annihilating two loops and its adjoint component $A^{\dagger}{}^{2}$ creating two loops, and so on.
We then define all the other spin-$j$ holonomy operators by recursion as polynomials of the fundamental $\hchi_{\f12}$ operator:
\be
\chi_{\f12}^{4}=\chi_{2}+3\chi_{1}+2
\,\,\Rightarrow\,\,
\hchi_{2}^{\,\mathrm{full}}
\,\equiv\,
\hchi_{\f12}^{\,4}-3\hchi_{\f12}^{\,2}+1\,,
\quad\dots
\ee
and so on with $\hchi_{j}^{\,\mathrm{full}}$ constructed from $\hchi_{\f12}^{2j}$ and smaller powers.
This construction clearly ensures that all the holonomy operators commute with each other. This method closely intertwines the definition of higher spin operators with  multi-loop holonomies. These multi-loop operators create spins $\f12$ (and then higher spins too) excitations on several loops at once. 

\medskip

To conclude the exploration of the basic loop quantum gravity operators, we should also deal with the symmetrized flux operators (and scalar products) with the $\su(2)$ generators acting as derivations on the wave-functions, and check their commutation relations with our new holonomy operators.
The flux and grasping operators are especially important since they allow to explore the intertwiner structure at the vertices. Indeed, acting with one-loop holonomy operators will only create decoupled loops at the vertex, while a generic intertwiner will couple them. So, even though we postpone the detailed analysis of the action of flux operators on loopy spin network to future investigation, we discuss below multi-loop holonomy operators that allow for coupled loops and thus explore the space of (loopy) intertwiners at the vertex.

For instance, considering two loops with group elements $h_{1}$ and $h_{2}$, we would like to excite the overall holonomy instead of the two independent holonomies, that is act with $\chi(h_{1}h_{2})$ instead of $\chi(h_{1})\chi(h_{2})$. Proceeding similarly to the one-loop holonomy operator, we define a two-loop holonomy operator $\hchi_{\f12}^{(2)}$, which creates and annihilates pairs of coupled loops:

\begin{definition}
\label{C-def}
We define the following five operators $C_{-2,-1,0,+1,+2}$ on the Fock space of symmetrized loopy intertwiners $\cHs$. They act on an arbitrary state $(f_{N})_{N\in\N}$ as:

\be
(C_{-2}f)_{N}(h_{1},..,h_{N})
\,=\,
\int \mathrm{d}k \mathrm{d}k\,\chi_{\f12}(k\tk)\,f_{N+2}(h_{1},h_{N},k,\tk) 
\ee
\be
(C_{-1}f)_{N}(h_{1},..,h_{N})
\,=\,
\f1N\sum_{i}^{N}
\Bigg{[}
\int \mathrm{d}k\,\chi_{\f12}(kh_{i})\,f_{N+1}(h_{1},..,h_{N},k)
-\int \mathrm{d}k\, \mathrm{d}k_{i}\,\chi_{\f12}(kk_{i})\,f_{N+1}(h_{1},..,k_{i},..,h_{N},k)
\Bigg{]}
\ee
\beq
(C_{0}f)_{N}(h_{1},..,h_{N})
\,=
&&
\f2{N(N-1)}\sum_{i<j}^{N}\Bigg{[}
\chi_{\f12}(h_{i}h_{j})f_{N}(h_{1},..,h_{N})
+\int \mathrm{d}k_{i} \mathrm{d}k_{j}\,\chi_{\f12}(k_{i}k_{j})\,f_{N}(h_{1},..,k_{i},..,k_{j},..,h_{N}) \nn\\
&&-\int \mathrm{d}k_{i}\,\chi_{\f12}(k_{i}h_{j})\,f_{N}(h_{1},..,k_{i},..,h_{N})
-\int \mathrm{d}k_{i}\,\chi_{\f12}(h_{i}k_{j})\,f_{N}(h_{1},..,k_{j},..,h_{N})
\Bigg{]}
\eeq
\be
(C_{+1}f)_{N}(h_{1},..,h_{N})
\,=\,
\f1{N(N-1)}\sum_{i\ne j}^{N}
\Bigg{[}
\chi_{\f12}(h_{i}h_{j})
f_{N-1}(h_{1},..,\widehat{h_{i}},..,h_{N})
-\int \mathrm{d}k_{j}\,\chi_{\f12}(h_{i}k_{j})
f_{N-1}(h_{1},..,\widehat{h_{i}},..,k_{j},..,h_{N})
\Bigg{]}
\ee
\be
(C_{+2}f)_{N}(h_{1},..,h_{N})
\,=\,
\f2{N(N-1)}\sum_{i<j}^{N}
\chi_{\f12}(h_{i}h_{j})f_{N-2}(h_{1},..,\widehat{h_{i}},..,\widehat{h_{j}},..,h_{N})
\ee
We complete this definition with the ``initial conditions'' for $N=0$ and $N=1$:
\be
(C_{-1}f)_{0}=\int \chi_{\f12}f_{1}
\,,\quad
(C_{0}f)_{0}=2f_{0}
\,,\quad
(C_{+1}f)_{0}=(C_{+2}f)_{0}=0\,,
\ee
\be
(C_{0}f)_{1}(h)=\chi_{\f12}(h)f_{1}(h)-\int \chi_{\f12}f_{1}
\,,\quad
(C_{+1}f)_{1}(h)=\chi_{\f12}(h)f_{0}
\,,\quad
(C_{+2}f)_{0}=0\,.
\ee
They satisfy the Hermiticity relations:
\be
C_{-2}=C_{2}^{\dagger}\,,\,\,
C_{-1}=C_{1}^{\dagger}\,,\,\,
C_{0}=C_{0}^{\dagger}
\,,
\ee
and the bounds for the $L^{2}$-norm:
\be
||C_{-2}+C_{+2}||_{2}\le 2\,,\quad
||C_{-1}+C_{+1}||_{2}\le 2\,,\quad
||C_{0}||_{2}\le 2\,.
\ee
Finally the two-loop holonomy operator $\hchi_{\f12}^{(2)}$ for spin $\f12$  is defined as the sum of these five components:
\be
\hchi_{\f12}^{(2)}\,\equiv\,
\f14
\big{(}C_{-2}+2C_{-1}+C_{0}+2C_{+1}+C_{+2}\big{)}\,.
\ee
This operator is essentially self-adjoint and bounded by 2.
\end{definition}

So the spectrum of the two-loop holonomy operator is once again bounded by 2. An important consistency check is that this bound is saturated by the flat state. The proof is a straightforward computation, with special care to the initial conditions for $N=0$ and $N=1$.
\begin{prop}
The flat state $\delta$, defined by $\delta_{0}=1$ and $\delta_{N\ge 1}=(\delta-1)^{\otimes N}$, is an eigenvector of the spin-$\f12$ two-loop holonomy operator  $\hchi_{\f12}^{(2)}$ with the highest eigenvalue on $\cHs$:
\be
\hchi_{\f12}^{(2)}\,
|\delta\ra
\,=\, 2\,|\delta\ra
\,.
\ee
\end{prop}

We could then similarly define an operator $\chi(h_{i}h_{j}^{-1})$ with a loop reversal or multi-loop operators $\chi(h_{i_{1}}..h_{i_{n}})$ taking care of properly symmetrizing the group elements. We can also generalize our construction replacing the spin-$\f12$ character by an arbitrary spin $j$ and define the spin-$j$ two-loop holonomy operator $\hchi^{(2)}_{j}$, and so on for more loops.
We will not go into these details, although we do not foresee any obstacle (beside the inflation of indices and sums).

We now turn to the first application of the holonomy operators discuss below the Hamiltonian constraints for BF theory on the Fock space of loopy spin networks.

\subsection{Revisiting the BF constraints as creation and annihilation of loops}
\label{BFsym}

Let us see how to implement the flatness constraint on our Fock space of loopy spin networks with bosonic statistics for the little loops. As earlier, we do not discuss the flatness constraints around loops of the base graph $\Gamma$, which are implemented as usual by using the standard holonomy operators around those loops. Here, we will focus on the fate of the little loop excitations at every vertex of the background graph $\Gamma$. For this purpose, we can focus on a single vertex and we can restrict ourselves to the flower graph, i.e. to the Fock space of loop intertwiners around a unique vertex. As we have constructed the holonomy operator in the previous section, we propose to use it as the Hamiltonian constraints for BF theory and simply impose:
\be
H^{\mathrm{BF}}=\hchi_{\f12}-2\,.
\ee
This is a self-adjoint operator and imposing this constraint amounts to projecting onto the highest eigenvalue of the holonomy operator. Since  $\hchi_{\f12}$ creates and annihilates loops by construction, $H^{\mathrm{BF}}$ shifts the number of loops and its flow should imply that the number of loops becomes pure gauge. Let us look at the space of physical states solving this flatness constraint.
By  proposition \ref{flat-prop1}, we already know that the flat state, defined as the factorized $\delta$-distribution state, saturates the holonomy bound, $H^{\mathrm{BF}}\,|\delta\ra=0$ .
The natural question is whether the flat state is the only solution to this constraint.

We will run into the same problem as in the case of distinguishable loops of higher derivative solutions to the holonomy constraint. In order to deal with this potential infinite-dimensional space of solutions, we will introduce as before a Laplacian constraint and multi-loop holonomy operators. However we will ultimately show that we require only a finite number of constraint operators (three to be exact) to impose full flatness and the uniqueness of the physical state despite the infinite number of loop excitation modes that need to be constrained. 

\medskip

More precisely, the holonomy constraint amounts to solving  functional recursion relations, relating $f_{N+1}$, $f_{N}$ and $f_{N-1}$ at each step. The problem is that this relation doesn't entirely fix $f_{N+1}$ in terms of $f_{N}$ and $f_{N-1}$, even assuming that these functions are invariant under permutations of their arguments and invariant under conjugation. Indeed it only fixes the integral $\int\mathrm{d}k \chi_{\f12}(k)\,f_{N+1}(h_{1},..,h_{N},k)$. This condition seems to fix only the spin-$\f12$ component of the function, so we face two obstacles: the non-trivial internal intertwiner structure and arbitrary higher spin excitations on each loop. We explain below how to get rid of all those modes by introducing constraints on the creation and annihilation of loops together with a Laplacian constraint.


Before treating the general case, we explore two simplified cases. First, factorized states avoid the problem of possible non-trivial intertwiner structure. It turns out that the spin-$\f12$ one-loop holonomy constraint is enough to constrain all the higher spin excitations and lead to the flat state as the unique physical state.  Second we consider the larger class of states with decoupled loops, defined mathematically as the wave-functions which are invariant under conjugation of its individual arguments (and not simply under the simultaneous conjugation of all its arguments as required by gauge invariance). In this case, the spin-$\f12$ constraint is not enough anymore and we need to explicitly introduce explicit constraints for all the higher spin excitations. We summarize these two cases in the following two propositions.

\begin{prop}
Let us consider a factorized state $\vphi\in\cHs$, that $\vphi_{0}=1$ and $\vphi_{N}=F^{\otimes N}$ for an integrable $F$ invariant under conjugation, $F(h)=F(ghg^{-1})$. Then the constraint $\hchi_{\f12}\,\vphi=\,2\vphi$ has a unique solution, which is the flat state, $\vphi=\delta$ and $F(h)=\delta(h)-1$.
\end{prop}
\begin{proof}
Let us look at the eigenvector equation on factorized states $\vphi$ defined as $\vphi_{0}=1$ and $\vphi_{N}=F^{\otimes N}$:
$$
\big{(}A+A^{\dagger}+B\big{)}\,F^{\otimes N}
\,=\,
4\,F^{\otimes N}\,.
$$
For $N=0$, this gives an integral condition on the one-loop wave-function $F$:
$$
\int \mathrm{d}k\,\chi_{\f12}(k)F(k)=2.
$$
Then, for $N=1$, we get a functional equality:
$$
\chi_{\f12}F+\chi_{\f12}-2=2F\,.
$$
Let us decompose $F$ on the spin basis. Since it is invariant under conjugation, it decomposes onto the characters $F=\sum_{j\ne 0} F_{j}\chi_{j}$. The $N=1$ equation translates into a recursion relation on the coefficients $F_{j}$ while the $N=0$ equation sets its initial condition:
\be
F_{\f12}=2\,,\quad
F_{1}+1=2F_{\f12}\,,\quad
\forall {j\ge 1}\,,\,\,
F_{j+\f12}+F_{j-\f12}=2F_{j}\,.
\ee
This has a unique solution $F_{j}=(2j+1)$, which translates to $F=\delta-1$. The constraint equation for $N\ge 2$ automatically follows.
\end{proof}

The case of factorized state works because the holonomy operator couples the creation of loops and the exploration of the higher spin components of the one-loop wave-function.
Next, we move to the larger class of functions which are invariant under conjugation of its individual arguments. Then the functions $\phi_{N}$ decompose on the character basis. Imposing the one-loop holonomy constraints for all spins leads to functional recursion equations such that the flat state is solution to the holonomy constraint. 
\begin{lemma}
Considering a state invariant under conjugation of each of its arguments,
$$
\phi_{N}(h_{1},..,h_{N})=\phi_{N}(g_{1}h_{1}g_{1}^{-1},..,g_{N}h_{N}g_{N}^{-1})\,,\quad\forall g_{i}\in\SU(2)^{N}\,,
$$
it decomposes on the character basis:
$$
\phi_{N}(h_{1},..,h_{N})=\sum_{j_{1},..,j_{N}}\phi_{N}^{j_{1},..,j_{N}}\prod_{i}^{N}\chi_{j_{i}}(h_{i})\,.
$$
Assuming that the $\phi_{N}$'s are all symmetric under permutations of their arguments and that they have no 0-modes, $\int dh_{1}\phi_{N}=0$ for all $N\ge1$, then
the only such solution to the set of holonomy constraints $\hchi_{j}\,\phi\,=\,(2j+1)\phi$ for all spins $j\in\f{\N^{*}}2$ is  the flat state $\phi_{N}=(\delta-1)^{\otimes N}$ (up to a global factor).
\end{lemma}
\begin{proof}
The proof is straightforward by recursion. For $N=0$, the constraint gives $\phi_{1}$ in terms of the no-loop mode $\phi_{0}$:
\be
\forall j\ge \f12\,,\,\,
\int dk\,\chi_{j}(k)\phi_{1}(k)=(2j+1)\phi_{0}\,,
\ee
which gives $\phi_{1}^{j}=(2j+1)\phi_{0}$ for all non-vanishing spins $j$ while $\phi_{1}^{0}=0$ by hypothesis. This way, if we fix the initial normalization to $\phi_{0}=1$,  we recover $\phi_{1}=(\delta-1)$.
Then the constraint equations for $N\ge1$ reads:
\beq
\forall j\ge\f12\,,
\,\,
2(2j+1)\phi_{N}(h_{1},..,h_{N})
&=&
\int \dd k\,\chi_{j}(k)\phi_{N+1}(h_{1},..,h_{N},k)+\f1N\sum_{i=1}^{N}\phi_{N-1}(h_{1},..,\widehat{h_{i}},..,h_{N}) \nn\\
&&
+\,
\f1N\sum_{i=1}^{N}\Bigg{[}
\chi_{j}(h_{i})\phi_{N}(h_{1},..,h_{N})-\int \dd k_{i}\,\chi_{j}(k_{i})\phi_{N}(h_{1},..,{k_{i}},..,h_{N})
\Bigg{]}
\,.
\eeq
We can solve this equation by recursion, determining the Fourier coefficients of $\phi_{N+1}$ in terms of $\phi_{N}$ and $\phi_{N-1}$. The coefficients $\phi_{N}^{j_{1}..j_{N}}$ vanish by assumption if one of the spins $j_{i}$ is zero. When none of the spins vanishes, we show that 
\be
\phi_{N}^{j_{1}..j_{N}}=\prod_{i=1}^{N}(2j_{i}+1)\,.
\ee
\end{proof}
Comparing to the case of distinguishable loops, in this case where the loops are individually gauge-invariant and thus decoupled, we have traded the infinity of holonomy constraints, one for each distinguishable loop, for the infinite tower of one holonomy constraint per spin mode for indistinguishable loops. Exploiting further the Fock space structure for the bosonic little loops, we can nevertheless reduce this infinity of holonomy constraints to a pair of constraints.
Indeed, checking the details of the proof of proposition \ref{flat-prop1} on the action of holonomy operators on the $\delta$-state, we propose to use  non-Hermitian constraints and characterize the flat state as an eigenvector of the loop annihilation operator $A$ and the loop creation operator $(B+A^{\dagger})$:
\begin{lemma}
\label{lemmaAB}
Considering a state $\phi$ invariant under conjugation of each of its arguments, and with no 0-modes, we introduce the pair of non-Hermitian constraint operators defined by the spin-$\f12$ annihilation and creation operators acting on $\cHs$:
\be
A\,|\phi\ra=2\,|\phi\ra
\,,\quad
(B+A^{\dagger})\,|\phi\ra=2\,|\phi\ra
\,.
\ee
Then the only solution to all these constraints is the flat state  $|\phi\ra=|\delta\ra$.
\end{lemma}
\begin{proof}
Let us write explicitly the eigenvalue equations for the state $\phi$:
\be
\forall N\ge 0,\,\,
\int \mathrm{d}k\,\chi_{\f12}(k)\phi_{N+1}(h_{1},..,h_{N},k)
\,=\,
2\phi_{N}(h_{1},..,h_{N})
\ee
\be
\forall N\ge 1,\,\,
\sum_{\ell=1}^{N}
\Bigg{[}
\chi_{\f12}(h_{\ell})\,\Big{[}\phi_{N-1}(h_{1},..,\widehat{h_{\ell}},..,h_{N})
+\phi_{N}(h_{1},..,h_{N})\Big{]}
\,-\,\int \mathrm{d}k_{\ell}\,\chi_{\f12}(k_{\ell})\phi_{N}(h_{1},..,{k_{\ell}},..,h_{N})
\Bigg{]}
\,=\,
2N\,\phi_{N}(h_{1},..,h_{N})
\nn
\ee
with the initial conditions equation at $N=0$ for the creation operator  $(B+A^{\dagger})$ trivially satisfied.
We could translate these equations into recursion relations on the Fourier coefficients, but there is actually a simpler and more direct route.
The first equation (for $A$) can be injected in the second equation turning it into a functional recursion:
\be
\sum_{\ell=1}^{N}
(2-\chi_{\f12}(h_{\ell}))\,
\Big{[}
\phi_{N-1}(h_{1},..,\widehat{h_{\ell}},..,h_{N})
+\phi_{N}(h_{1},..,h_{N})
\Big{]}
\,=\,
0\,.
\ee
For $N=1$, this relates the one-loop wave-function $\phi_{1}$ to the no-loop normalization $\phi_{0}$:
$$
\forall h\in\SU(2),\quad
(2-\chi_{\f12}(h))\,(\phi_{0}+\phi_{1}(h))\,=\,0\,.
$$
Since $\phi_{1}$ is invariant under conjugation, this holonomy constraint has a unique distributional solution up to an arbitrary factor, $(\phi_{0}+\phi_{1})\propto\delta$. The integral condition, $\int \chi_{\f12}\phi_{1}=2\phi_{0}$, fixes this factor and we recover $\phi_{1}=(\delta -1)$ as expected as we fix the normalization $\phi_{0}=1$.

We then proceed by recursion, fixing the number of loops $N\ge 2$ and assuming that $\phi_{n}=(\delta-1)^{\otimes n}$ for all $n\le (N-1)$. Let us now prove this statement holds for $n=N$. Using the identity $(2-\chi_{\f12}(h))\delta(h)=0$ , we start by checking that:
$$
\sum_{\ell}^{N}
(2-\chi_{\f12}(h_{\ell}))\,
\Big{[}
\prod_{i}(\delta(h_{i})-1)-\prod_{i\ne\ell}(\delta(h_{i})-1)
\Big{]}
\,=\,0\,.
$$
This implies that:
$$
\Bigg{(}
\sum_{\ell}^{N} (2-\chi(h_{\ell})
\Bigg{)}
\,
\Bigg{(}
\phi_{N}(h_{1},..,h_{N})-\prod_{i}(\delta(h_{i})-1)
\Bigg{)}
\,=\,0\,.
$$
Since every holonomy operator $(2-\hchi_{\ell})$ is Hermitian positive, this means that the holonomy constraint holds for each loop individually:
\be
\forall \ell \le N\,,\,\,
\sum_{\ell}^{N} (2-\chi(h_{\ell})
\,
\Bigg{(}
\phi_{N}(h_{1},..,h_{N})-\prod_{i}(\delta(h_{i})-1)
\Bigg{)}
\,=\,0\,.
\ee
Since we have assumed that the wave-function is invariant under conjugation individually for each of its arguments, the only distribution solution to this equation is the product of $\delta$-function up to a global factor:
$$
\phi_{N}(h_{1},..,h_{N})=\prod_{i}(\delta(h_{i})-1)+\alpha \prod_{i}\delta(h_{i})\,,
$$
for some factor $\alpha$ to be determined. Checking this identity against the integral condition $\int \chi_{\f12}\phi_{N}=\phi_{N-1}$ yields $\alpha=0$ thus proving the proposition.
\end{proof}
We see that the requirement of the invariance under conjugation for each loop individually (stronger than gauge-invariance requiring the invariance under global conjugation) is crucial in the last step of the proof. Else we would have to deal with derivative solutions, in $\pp\delta$ and so on, as in the case of distinguishable loops.

\medskip

Our proposal amounts to adding another constraint along side the Hermitian holonomy constraint $\hchi_{\f12}=2$. Instead of taking the average of the two operators $A$ and $(B+A^{\dagger})$ and defining the holonomy operator, we subtract them and get the other constraint $B+(A^{\dagger}-A)=0$. This new constraint operator has a Hermitian part $B$ and a anti-Hermitian part $(A^{\dagger}-A)$, such that the overall structure can be interpreted as a holomorphic constraint, similar to the annihilation operator $a=\hat{x}-i\hat{p}$ for the harmonic oscillator. 
From this perspective, eigenvectors of this ``holomorphic'' operator $B+(A^{\dagger}-A)$ can be considered as coherent states, which is pretty natural since we are looking into coherent superpositions of any number of loops summing over $N$, and the $\delta$-state, as a null eigenvector of that operator, can be considered as a ground state.

The trick why these two  constraint operators $A$ and $(B+A^{\dagger})$ are enough to kill all the degrees of freedom and lead to a single physical state is that they do not commute and their commutators actually generate higher spin constraints:
\begin{lemma}
Imposing the two constraints with $A$ and $(B+A^{\dagger})$ on $\cHs$ implies a tower of constraints with all the higher spin annihilation operators:
\be
A\,|\phi\ra=(B+A^{\dagger})\,|\phi\ra=2\,|\phi\ra
\qquad\Longrightarrow\qquad
\forall j\ge \f12\,,\,\,
A_{j}\,|\phi\ra=(2j+1)\,|\phi\ra\,.
\ee
\end{lemma}
\begin{proof}
Let us look at the commutator $(\hN+1)\,[A_{j},(B+A^{\dagger})]$. This commutator will vanish on solution states $|\phi\ra$. Using the commutation relations computed earlier \eqref{commAB1} and \eqref{commAB2}, we get for $j=\f12$:
$$
(\hN+1)\,[A,(B+A^{\dagger})]=A_{1}+\id-(B+A^{\dagger})A\,,
\quad
A_{1}\,|\phi\ra=\Big{[}(B+A^{\dagger})A-\id\Big{]}\,|\phi\ra=(4-1)\,|\phi\ra=3\,|\phi\ra\,,
$$
then  for higher spins $j\ge 1$ the commutation relation $(\hN+1)\,[A_{j},(B+A^{\dagger})]=A_{j+\f12}+A_{j-\f12}\id-(B+A^{\dagger})A_{j}$ implies:
\be
A_{j+\f12}\,|\phi\ra=\Big{[}(B+A^{\dagger})A_{j}-A_{j-\f12}\Big{]}\,|\phi\ra
=\Big{[}2(2j+1)-2j\Big{]}\,|\phi\ra=(2j+2)\,|\phi\ra\,.
\ee
\end{proof}

Our two (non-Hermitian) constraints do not commute and generate an infinite number of constraints killing all the higher spin excitations, leaving us at the end of the day with the single totally flat state.
In some sense, this pair of annihilation and creation constraint operators can be considered as the generators of the algebra of holonomy operators on the Fock space of loopy spin networks.

\medskip

Let us move on to the general case. Working with states invariant under conjugation for each loop individually amounts to considering states created loop by loop, only by the action of one-loop holonomy operators. This leads to decoupled loops and unfortunately  does not explore the whole space of intertwiners: we still need to reach all the states globally invariant under conjugation but not invariant under conjugation of the individual arguments, such as $\chi(h_{1}h_{2}..)$. In the spin decomposition of the wave-functions, this corresponds to the fact that the modes are not simply $\phi_{N}^{j_{1},..,j_{N}}$ but should be labelled as $\phi_{N}^{j_{1},..,j_{N},\cI}$: they do not depend only on the spins $j_{i=1..N}$ but further depend on the data of a (loopy) intertwiner $\cI$ between (two copies of) all the spins. This leads to the existence of the derivative solutions to the holonomy constraints, defined by applying differential operators (graspings) to the $\delta$-distribution.

The intertwiner structure is hard to constraint completely. One way to go is to not only use higher spin operators but introduce multi-loop holonomy constraints, as in the case of disitinguishable loops. Indeed, since we would like to freeze all the spin excitations on the possible infinity of loops, it is natural to introduce one constraint operator per mode\footnotemark. 
\footnotetext{
This leads us to conjecture a set of complete holonomy constraints for BF theory. Considering all the multi-loop holonomy operators for arbitrary spins acting on the Fock space of loopy intertwiners $\cHs$:
$$
\forall j\in\f{\N^{*}}2\,,\quad
\forall n\in\N^{*}\,,\quad
\hchi^{(n)}_{j}\,|\phi\ra=\,(2j+1)\,|\phi\ra\,,
$$
then the only solution to all these constraints is the flat state $\phi=\delta$.
We have checked this conjecture up to the three-loop component of the state, $N=3$, but we haven't gone further. This would require explicitly and carefully defining the multi-loop holonomy operators. We should also take special care of working with legitimate states, controlling the convergence/divergence of the series in $j$ and $N$ to ensure that the states are distributions.
}

We propose to take a different route in order to keep a finite number of (primary) constraints. We introduce a Laplacian constraint to project onto the space of wave-functions invariant under conjugation and use the creation and annihilation operators for loops to impose flatness:
\begin{prop}
We consider the pair of non-Hermitian constraint operators defined by the spin-$\f12$ annihilation and creation operators acting on $\cHs$:
\be
A\,|\phi\ra=2\,|\phi\ra
\,,\quad
(B+A^{\dagger})\,|\phi\ra=2\,|\phi\ra
\,.
\ee
We supplement these constraints with the Laplacian constraint:
\be
(\tDelta-\Delta)\,|\phi\ra=0
\qquad\textrm{with}\quad
(\Delta\vphi)_{N}=\f1N\sum_{\ell}^{N}\Delta_{\ell}\,\vphi_{N}\,,
\quad
(\tDelta\vphi)_{N}=\f1N\sum_{\ell}^{N}\tDelta_{\ell}\,\vphi_{N}\,.
\ee
Imposing these three eigenvalue equations leads to a unique solution (up to a global factor), the flat state $|\phi\ra=|\delta\ra$ defined by $\phi_{N}=(\delta-1)^{\otimes N}$.
\end{prop}
\begin{proof}
We start with the Laplacian constraints:
$$
\sum_{\ell}(\tDelta_{\ell}-\Delta_{\ell})\phi_{N}=0\,.
$$
Since every Laplacian constraint operator $(\tDelta_{\ell}-\Delta_{\ell})$ on each loop $\ell$ is Hermitian and positive, this imposes that each of them vanish on the wave-function, i.e. for all $\ell$ we have $(\tDelta_{\ell}-\Delta_{\ell})\phi_{N}=0$. This implies that $\phi_{N}$ is invariant under conjugation of each of its argument. Then we apply lemma \ref{lemmaAB} to prove the uniqueness of the solution state.

\end{proof}

On the one hand, the Laplacian constraint fixes how every loop is attached to the vertex, through a trivial spin-0. Each loop is invariant under conjugation on its own, states are collections of bosonic loops, each carrying a spin and with a trivial intertwiner between them. On the other hand, the constraints $A$ and $(B+A^{\dagger})$ realize explicitly the idea that BF dynamics impose that the creation and annihilation of loops are pure gauge. 

\medskip

To conclude this section, we would like to underline the similarities and differences between the case of distinguishable loops and the Fock space of indistinguishable little loop excitations. When working with little loops endowed with bosonic statistics, one must take a special care to consistently remove the spin-0 modes on every loop to implement the cylindrical consistency of the wave-functions. This leads to a (spin-$\f12$) holonomy operator also creating and annihilating loops . We explicitly separate its components respectively creating and annihilating loops and use them as legitimate constraint operators for BF theory. This is different from distinguishable loops where holonomy operators are defined as attached to a loop: a holonomy operator acts on a given loop, exciting and shifting the spin carried by the loop.

Nevertheless, the issue of the intertwiner space living at vertex and coupling the loops is the same in both frameworks. We have identified an infinity of solutions to the holonomy constraints, constructed as differential operators acting on the $\delta$-distribution (as graspings on the spin network wave-function). These are still peaked on the identity group element, but they potentially define an infinity of gauge-invariant local degrees of freedom living at the vertex. To get rid of these ``spurious'' solutions, we have introducing a Laplacian constraint that forces each loop excitation to be invariant under conjugation, thus linking it trivially to the vertex. This allows to kill all local intertwiner excitation. Then we take as  Hamiltonian constraints for BF theory this combination on holonomy and Laplacian constraints, which lead as wanted to a unique physical state, the flat $\delta$-state.

\section{Tagged spin networks}

After folded spin networks, which retains the internal combinatorial structure inside coarse-grained regions, and loopy spin networks, which keep local curvature excitations as little loops attached to the vertices of the base graph and which we have explored in great details in the previous sections, we move on to the third and last step of coarse-grained structures: {\it tagged spin networks}.

When integrating out the connection group elements inside a bounded region, as discussed in \cite{Livine:2013gna,Dittrich:2014wpa,Bahr:2015bra} and reviewed in the first section \ref{sec:overview}, and coarse-graining the region to a single vertex, we naturally break the local gauge invariance at the resulting coarse-grained vertex.
This  also happens as soon as we introduce fermionic matter fields, which act as sources for loop quantum gravity's Gauss law and thus create non-trivial closure defects.
At the classical level, this is reflected by a non-vanishing closure defect: the sum of the flux vectors around the effective vertex does not vanish anymore and is actually balanced by the internal fluxes living on the loops living inside the bounded region and carrying non-trivial holonomies. Overall, the gauge invariance is restored if we take into account the internal degrees of freedom of the region however, once we have coarse-grained it, the breaking of the gauge invariance reflects the geometry excitations which have developed in the region's bulk and which we have traced out.
At the quantum level, the closure defect becomes a tag, attached to each vertex, as drawn on fig.\ref{tagged}. This internal degree of freedom is defined as an extra spin coupling to the actual spin living on the links and edges attached to the effective vertex and connecting the coarse-grained region to its exterior. This tag allows to relax the gauge invariance in a controlled way.

Mathematically, we are thus led to consider the whole space of non-gauge-invariant cylindrical functionals of the connection on a given fixed background graph $\Gamma$. The tagged spin networks will provide a basis of that space, with the tags record how much the local gauge-invariance is broken: when the tags vanish, we recover the usual gauge-invariant spin network basis states. This allows to account for graph changing dynamics in an effective manner. Even though the graph changes and might get more complexed as the geometry evolves, we keep on coarse-graining the state projecting it onto the fixed base graph (chosen by the observer), then the internal degrees of freedom and non-trivial curvature developed inside the coarse-grained regions gets translated into excitations of the effective tag degree of freedom attached to the base graph vertices.

The space of tagged spin networks is naturally quite simple and we believe it offers a useful framework for the study of coarse-graining of the loop quantum gravity dynamics.

\begin{figure}[h!]

\centering

\begin{tikzpicture}[scale=0.7]
\coordinate(O1) at (0,0);

\draw (O1) -- ++(-2,1) node[above] {$j_1$};
\draw (O1) -- ++(2,1) node[above] {$j_2$};
\draw (O1) -- ++(2,-1) node[below] {$j_3$};
\draw (O1) -- ++(-2,-1) node[below] {$j_4$};
\draw[thick,red] (O1) to ++(0,0.5) node[above]{$j,m$};

\draw (O1) node {$\bullet$} ++(0,-0.3) node{$i$};

\end{tikzpicture}

\caption{We consider tagged vertices: a vertex with the additional tag corresponding to a closure defect. The representations $j_{1},..,j_{4}$ living on the graph edges linked to the vertex do not form an intertwiner on their own, but they recouple to the spin $j$ defining the tag.
\label{tagged}}

\end{figure}
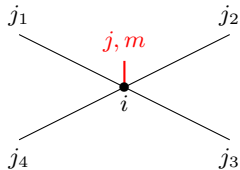

\subsection{The tagged spin network basis}

We consider the space $\cH_{\Gamma}^{tag}$ of non-gauge-invariant wave-functions on the (oriented and connected) graph $\Gamma$. This is simply the space of $L^{2}$ functions on $\SU(2)^{\times E}$, where $E$ is the number of edges or links of $\Gamma$, with no further assumption. Considering such a function, we can project onto the usual space of gauge-invariant states by group averaging:
\be
\psi(\{g_{e}\}_{e\in\Gamma})
\,\in\cH_{\Gamma}^{tag}
\quad\mapsto\quad
\psi_{0}(\{g_{e}\}_{e\in\Gamma})
\,=\,
\int_{\SU(2)^{V}}\dd h_{v}
\psi(\{h_{s(e)}^{-1}g_{e}h_{t(e)}\}_{e\in\Gamma})
\,\,\in\cH_{\Gamma}\,.
\ee
We can generalize this projection to non-trivial recouplings at every vertex and get an exact decomposition of the full non-invariant state:
\be
\psi
\,=\,
\sum_{\{J_{v}\}\in\f\N2}
\prod_{v}(2J_{v}+1)\,P_{\{J_{v}\}}\psi\,,
\quad
P_{\{J_{v}\}}\psi\,(\{g_{e}\}_{e\in\Gamma})
\,=\,
\int \dd h_{v}\,
\prod_{v}\chi_{J_{v}}(h_{v})\,
\psi(\{h_{s(e)}^{-1}g_{e}h_{t(e)}\}_{e\in\Gamma})\,.
\ee
The spin $J_{v}$ is the tag living at the vertex $v$ and provides a measure of how much gauge-invariance is relaxed at that vertex. It is the variable conjugate to the group averaging variable $h_{v}$. 
Following this logic, we can make all states in   $\cH_{\Gamma}^{tag}$ gauge-invariant by adding  $h_{v}$ as an actual argument of the wave-function. This provides a isomorphism between  $\cH_{\Gamma}^{tag}=L^{2}(\SU(2)^{\times E})$ and $\cH_{\Gamma}^{ext}=L^{2}(\SU(2)^{\times E}\times\SU(2)^{\times V}/\SU(2)^{\times V})$ where $ext$ stands for ``extended'':
\be
\Psi(\{g_{e},h_{v}\}_{e,v\in\Gamma})
\,=\,
\Psi(\{k_{s(e)}^{-1}g_{e}k_{t(e)},k_{v}^{-1}h_{v}\}_{e,v\in\Gamma})
\quad\mapsto\quad
\psi(\{g_{e}\}_{e\in\Gamma})
\,=\,
\Psi(\{g_{e},h_{v}=\id\}_{e,v\in\Gamma})\,.
\ee
We define tagged spin networks as basis states for $\cH_{\Gamma}^{ext}$ thus providing through this gauge-fixing map a basis for generic non-gauge-invariant states. These generalizations of spin networks are labeled by spins $j_{e}$ on every edge $e$, the tag spin $J_{v}$ an magnetic momentum $M_{v}$ at every vertex, as well as an intertwiner $\cI_{v}$ recoupling at each vertex between the tag  and the spins on the edges attached to that vertex:
the incoming and outgoing edges attached to the vertex $v$:
$$
\cI_{v}:\cV^{J_{v}}\otimes\bigotimes_{e|s(e)=v}\cV^{j_{e}}\longrightarrow\bigotimes_{e|t(e)=v}\cV^{j_{e}}\,,
$$
\be
\label{tagbasis}
\Psi_{\{j_{e},J_{v},M_{v},\cI_{v}\}}
(\{g_{e},h_{v}\})
\,\equiv\,
\prod_{v}D^{J_{v}}_{m_{v}M_{v}}(h_{v})\,
\prod_{e}\la j_{e}m_{e}^{s}|\,g_{e}\,|j_{e}m_{e}^{t}\ra\,
\prod_{v}\la \otimes_{e|t(e)=v}j_{e}m_{e}^{t}|\,\cI_{v}\,|J_{v}m_{v}\otimes_{e|s(e)=v}j_{e}m_{e}^{s}\ra\,,
\ee
with an implicit over the magnetic momenta $m_{e}^{v}$ and $m_{v}$. In simple words, we work with spin network on graphs with an extra open edge at every vertex. The spins carried by those open edges are the tags.

\medskip

The whole question is the physical interpretation of these tags, which we added to the usual spin network states. It is mathematically clear how the closure defects arise from coarse-graining and that the tags reflect non-trivial holonomies around the loops of the subgraph within the coarse-graining regions. The next challenge would be to show that they can be related to some physical notions of (quasi-)local energy density or mass (see e.g. \cite{Yang:2008th} for a definition of the quasi-local energy operator in loop quantum gravity).

\subsection{Tags from coarse-graining and tracing out little loops}

Let us show how starting from a loopy spin network and tracing out the little loops attached to the vertices leads naturally to a reduced density matrix defined in terms of tagged spin networks. So we consider a gauge-invariant loopy state defined on the base graph $\Gamma$ with a certain number of loops $n_{v}$ attached to each vertex $v$:
$$
\phi(\{g_{e},h^{v}_{\ell}\})
\,=\,
\phi(\{k_{s(e)}^{-1}g_{e}k_{t(e)},k_{v}^{-1}h^{v}_{\ell}k_{v}\})\,\,
\forall k_{v}\in\SU(2)^{\times V}\,,
$$
where the group elements $h^{v}_{\ell}$ live on the little loops $\ell$ attached to the vertex $v$, and we integrate out the loops:
\be
\rho(\{g_{e},\tg_{e}\}_{e\in\Gamma})
\,=\,
\int \prod_{v}\prod_{\ell=1}^{n_{v}}\dd h_{\ell}^{v}\,\,
\overline{\phi(\{g_{e},h^{v}_{\ell}\})}\,
\phi(\{\tg_{e},h^{v}_{\ell}\})\,.
\ee
Let us compute the reduced density matrix using the natural loopy spin network basis. We focus on the little loops attached to single vertex, say $v_{0}$, and drop the index $v$ from the little loop group elements for the sake of simplicity.
We consider the loopy states defined by basis intertwiners defined by two intertwiners, one recoupling the spins living on the edges linked to the vertex $v_{0}$ and one recoupling the little loops attached to that vertex, glued through an intermediate spin $J_{v_{0}}$, as drawn on fig.\ref{fig:intermediatespin2}:
\beq
\Phi_{\{j_{e},J_{v_{0}},i_{v},j_{\ell},\tilde{i}_{v_{0}}\}}
(\{g_{e},h_{\ell}\})
&=&
\prod_{e}\la j_{e}m_{e}^{s}|\,g_{e}\,|j_{e}m_{e}^{t}\ra\,
\prod_{\ell}\la j_{\ell}m_{\ell}^{s}|\,h_{\ell}\,|j_{\ell}m_{\ell}^{t}\ra\,
\prod_{v\ne v_{0}}\la \otimes_{e|t(e)=v}j_{e}m_{e}^{t}|\,i_{v\,}|\otimes_{e|s(e)=v}j_{e}m_{e}^{s}\ra\,
\nn\\
&&
\la \otimes_{e|t(e)=v_{0}}j_{e}m_{e}^{t}|\,i_{v_{0}}\,|J_{v_{0}}M_{v_{0}}\otimes_{e|s(e)=v}j_{e}m_{e}^{s}\ra\,
\la J_{v_{0}}M_{v_{0}} \otimes_{\ell}j_{\ell}m_{\ell}^{t}|\,\tilde{i}_{v_{0}}\,|\otimes_{\ell}j_{\ell}m_{\ell}^{s}\ra
\,,
\eeq
with an implicit sum over all magnetic moment labels.
We have assumed, as announced, that only the vertex $v_{0}$ has little loops attached to it, so all other vertices are thought as having a vanishing intermediate spin $J_{v\ne v_{0}}=0$.
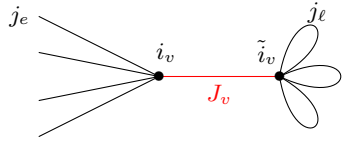
\begin{figure}
\begin{tikzpicture}[scale=0.8]

\coordinate(O2) at (8,0);
\coordinate(O3) at (10,0);

\draw (O2) -- ++(-2,1) node[left]{$j_{e}$};
\draw (O2) -- ++(-2,0.4);
\draw (O2) -- ++(-2,-0.4);
\draw (O2) -- ++(-2,-1);

\draw[in=-25,out=25,scale=3] (O3) to[loop] (O3);
\draw[in=30,out=80,scale=3] (O3)  to[loop] node[pos=0.6,right,above]{$j_{\ell}$} (O3);
\draw[in=-80,out=-30,scale=3] (O3) to[loop] (O3);

\draw[red] (O2) -- (O3) node[midway,below]{$J_{v}$};
\draw (O2) node {$\bullet$} ++(0.12,0.4) node{$i_v$};
\draw (O3) node {$\bullet$} ++(-0.2,0.4) node{$\tilde{i}_v$};

\end{tikzpicture}

\caption{We introduce the intermediate spin basis for the vertices of loopy spin networks: intertwiners will decompose into two intertwiners, the first one $i_{v}$ recoupling the spins living on the base graph edges and the other $\tilde{i}_{v}$ recoupling the spins living on the little loops attached to the vertex, which are linked together by the intermediate spin $J_{v}$. When we coarse-grain by tracing over the little loops, the only remaining information is this intermediate spin $J_{v}$, which becomes the tag measuring the closure defect at the vertex. It is the remnant of the curvature fluctuations and internal geometry within the vertex.} 
\label{fig:intermediatespin2}

\end{figure}

A loopy state will decompose onto that basis, $|\phi\ra=\phi_{\{j_{e},J_{v_{0}},i_{v},j_{\ell},\tilde{i}_{v_{0}}\}}
\,|\Phi_{\{j_{e},J_{v_{0}},i_{v},j_{\ell},\tilde{i}_{v_{0}}\}}\ra$, and we easily compute the resulting reduced density matrix using the orthonormality of the Wigner matrices with respect to the Haar measure on $\SU(2)$ and find that it naturally decompose onto the tagged spin network basis introduced $\Psi_{\{j_{e},J_{v},M_{v},\cI_{v}\}}$ above in \eqref{tagbasis}:
\be
\rho
\,=\,
\tr_{\{h_{\ell}\}}|\phi\ra\la\phi|
\,=\,
\Big{(}
\sum_{\{m^{s,t}_{\ell}\}}
\big{|}
\la J_{v_{0}}M_{v_{0}} \otimes_{\ell}j_{\ell}m_{\ell}^{t}|\,\tilde{i}_{v_{0}}\,|\otimes_{\ell}j_{\ell}m_{\ell}^{s}\ra
\big{|}^{2}
\Big{)}
\,
|\Psi_{\{j_{e},J_{v_{0}},M_{v_{0}},i_{v}\}}\ra\la\Psi_{\{j_{e},J_{v_{0}},M_{v_{0}},i_{v}\}}|\,.
\ee
Thus the intermediate spins of the loopy spin networks, which recouple between the base graph edges and the little loop excitations, become the tags of the tagged spin network basis after tracing out the holonomies living on the little loops. This concludes the coarse-graining of the geometry of a bounded region to a single vertex plus one extra degree of freedom -the tag- registering the  excitations of geometry and curvature within that region's bulk.

\subsection{Revisiting BF theory: the non-trivial space of flat spin network states}

The (totally) flat state on a closed and connected graph $\Gamma$ is defined by a $\delta$-distribution on every loop of the graph. This is generically redundant, so we choose a set of independent loops\footnotemark, that is $L\equiv (E-V+1)$ loops generating all the loops of $\Gamma$, and define the flat state as the product of the $\delta$-distribution on those independent loops.
\footnotetext{
The simplest way to proceed is, as in the gauge fixing procedure, to choose a maximal tree on $\Gamma$, then to associate one  loop to each edge which does not belong to the tree. This ensures that each of those loops contains one edge that no other loop contain and are thus independent. 
}
Now we would like to impose holonomy operator constraints along all the (independent) loops of the graph in order to this flat state as single physical states. We will encounter the same obstacle of the derivative solutions to the holonomy constraints as  with loopy spin networks.

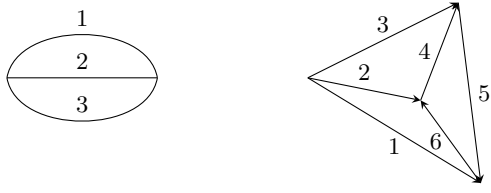
\begin{figure}[h!]
\centering
\begin{tikzpicture}[scale=1]

\coordinate(A) at (-4,0);
\coordinate(B) at (-2,0);

\draw(A) -- node[midway,above]{2} (B) ;
\draw (A) to[bend left=80] node[midway,above]{1}(B);
\draw (A) to[bend right=80] node[midway,above]{3}(B);

\coordinate(a) at (0,0);
\coordinate(b) at (1.5,-0.3);
\coordinate(c) at (2,1);
\coordinate(d) at (2.3,-1.4);

\draw[->,>=stealth](a) --  node[midway,above]{2} (b) ;
\draw[->,>=stealth](a) -- node[midway,above]{3} (c) ;
\draw[->,>=stealth](a) -- node[midway,below]{1} (d) ;
\draw[->,>=stealth](b) -- node[midway,left]{4} (c) ;
\draw[->,>=stealth](c) -- node[midway,right]{5} (d) ;
\draw[->,>=stealth](d) -- node[midway,left]{6} (b) ;

\end{tikzpicture}

\caption{The $\Theta$-graph on the left and the tetrahedron graph on the right.}
\label{fig:thetatetra}

\end{figure}

Let us, for the sake of simplicity, look at the $\Theta$-graph, made of two vertices connected by three edges, as in fig.\ref{fig:thetatetra}. The flat state is:
\be
\delta_{\Theta}(g_{1},g_{2},g_{3})
\,=\,
\delta(g_{1}g_{2}^{-1})\delta(g_{1}g_{3}^{-1})
\,=\,
\delta(g_{1}g_{2}^{-1})\delta(g_{2}g_{3}^{-1})
\,=\,
\delta(g_{1}g_{3}^{-1})\delta(g_{2}g_{3}^{-1})
\,,
\ee
where we give the three possible  sets of independent loops, corresponding to the three choices of tree on the $\Theta$-graph.
This is clearly a solution to the three holonomy constraints:
$$
\Big{(}\chi_{\f12}(g_{1}g_{2}^{-1})-2\Big{)}\,\delta_{\Theta}
\,=\,
\Big{(}\chi_{\f12}(g_{2}g_{3}^{-1})-2\Big{)}\,\delta_{\Theta}
\,=\,
\Big{(}\chi_{\f12}(g_{3}g_{1}^{-1})-2\Big{)}\,\delta_{\Theta}
\,=\,
0\,.
$$
However, any first derivative give also a solution to those holonomy constraints. As an example, if we differentiate along the first group element $g_{1}$, we ge the distribution 
$\pp_{a}^{(1)}\delta_{12}\delta_{23}
=
\pp_{a}^{(1)}\delta(g_{1}g_{2}^{-1})\delta(g_{2}g_{3}^{-1})
=
i\delta(g_{1}J_{a}g_{2}^{-1})\delta(g_{2}g_{3}^{-1})$
 also satisfying:
\be
\Big{(}\chi_{\f12}(g_{1}g_{2}^{-1})-2\Big{)}\,\pp_{a}^{(1)}\delta_{12}\delta_{23}
\,=\,
\Big{(}\chi_{\f12}(g_{2}g_{3}^{-1})-2\Big{)}\,\pp_{a}^{(1)}\delta_{12}\delta_{23}
\,=\,
\Big{(}\chi_{\f12}(g_{3}g_{1}^{-1})-2\Big{)}\,\pp_{a}^{(1)}\delta_{12}\delta_{23}
\,=\,
0
\,.
\ee
And we can go on constructing an infinity of higher order derivative solutions, as explained in the framework of loopy spin networks in sections \ref{derivativesolution1} and \ref{derivativesolution2}.
In some sense, the totally flat state, defined by the $\delta$-distribution, is the primary state. Then we act on it with differentiation operators and generate a whole space of solutions. All these states are technically still flat, since they are peaked exclusively on the identity group element $\id$. They seem to be pure excitations of the triad conjugate to the connection (similar to quantum excitations of the electric field at vanishing magnetic field).  This non-trivial space of flat states over the flat state sounds similar in spirit to the construction by Dittrich and collaborators \cite{Dittrich:2014wpa,Bahr:2015bra} where they build a Hilbert space representation for loop quantum gravity with the flat state as ``ground state''. It would be interesting to check if an explicit relation with our present construction could be identified.

\medskip

As the first derivative are not gauge-invariant distributions, it is reasonable to wonder if restricting to gauge-invariant spin network states by imposing vanishing tags allows to kill all those higher order flat solutions and leave the original $\delta$-distribution as sole physical state. We show below that this is indeed the case, but with the very important subtlety that we need to impose the holonomy constraints along all the loops of the graph and not only a set of independent loops.


First, we impose the tags to vanish by a Laplacian constraint at every vertex:
\be
\Delta_{v}
\,\equiv\,
\big{(}
\sum_{e|v=t(e)}\pp_{a}^{(e)R}
-\sum_{e|v=s(e)}\pp_{a}^{(e)L}
\big{)}^{2}\,.
\ee
Requiring $\Delta_{v}\,\vphi=0$ implies that the recoupling of the spins $j_{e}$ on the edges attached to the vertex $v$ is trivial, i.e. that the tag $J_{v}$ is the trivial representation, i.e that the wave-function $\vphi$ satisfies the  gauge-invariance at the vertex $v$. Imposing this constraint at all the vertices allows to project from the tagged spin networks down to the usual gauge-invariant spin network states.

\medskip

Second, let us revisit the flatness constraints to impose on standard spin networks in order  to implement BF theory and its projector on the flat state. Imposing the holonomy constraints on a set of independent loops turns out not to be enough to get a unique solution state and we have gauge-invariant remnants of the derivative solutions.  To remedy this, it is necessary and sufficient to impose the holonomy constraints on all possible loops of the graph. This is the counterpart of the multi-loop holonomy constraints for imposing the flatness of loopy intertwiners as explained in sections \ref{derivativesolution2} and \ref{derivativesolution3}.

For instance, on the $\Theta$-graph, a set of independent loops is provided by the two loops, $g_{1}g_{2}^{-1}$ and $g_{2}g_{3}^{-1}$. Classically, imposing that these two group elements vanish implies that the flatness of the third ``composite'' loop, $g_{1}g_{3}^{-1}$. At the quantum level however, we can construct an infinity of derivative solutions to the holonomy constraints $\hchi_{12}$ and $\hchi_{23}$ by acting with grasping operators, as an example:
\beq
\vphi=\sum_{a}\pp_{a}^{(1)}\delta(g_{1}g_{2}^{-1})\pp_{a}^{(3)}\delta(g_{3}g_{2}^{-1})
\quad\Rightarrow
&&
\Big{(}\chi_{\f12}(g_{1}g_{2}^{-1})-2\Big{)}\,\vphi
\,=\,
\Big{(}\chi_{\f12}(g_{2}g_{3}^{-1})-2\Big{)}\,\vphi
=0
\nn\\
\quad\textrm{but}&&
\int f (\chi_{\f12}(g_{1}g_{3}^{-1})-2)\,\vphi
\,=\,
\left.-\f14f\Delta\chi_{\f12}\right|_{\id}
=-\f32\,f(\id)\ne 0\,.
\eeq
These distributions will actually not be solution to the third holonomy constraint $\hchi_{13}$ and, so imposing directly the three holonomy constraints together determines the $\delta$-state $\delta_{\Theta}$ as unique solution.

Another interesting case is on the tetrahedron graph, see fig.\ref{fig:thetatetra}, where imposing the holonomy constraints around three triangles $(g_{1}g_{6}g_{2}^{-1})$, $(g_{2}g_{4}g_{3}^{-1})$ and $(g_{3}g_{5}g_{1}^{-1})$ does not imply the holonomy constraint around the fourth triangle $(g_{4}g_{5}g_{6})$. This is realized by a triple grasping around the vertex $(123)$ acting on the $\delta$-distribution, which gives the following gauge-invariant state:
\be
\vphi
=
\eps^{abc}\pp_{a}^{(1)L}\delta(g_{1}g_{6}g_{2}^{-1})
\pp_{b}^{(2)L}\delta(g_{2}g_{4}g_{3}^{-1})
\pp_{c}^{(3)L}\delta(g_{3}g_{5}g_{1}^{-1})\,.
\ee
Applying this distribution against a gauge-invariant test function $f(g_{1},..,g_{6})$, we compute the action of the holonomy operators around the four 3-cycles of the tetrahedron graph:
\be
\int f (\chi_{\f12}(g_{1}g_{6}g_{2}^{-1})-2)\,\vphi
\,=\,
\int f (\chi_{\f12}(g_{2}g_{4}g_{3}^{-1})-2)\,\vphi
\,=\,
\int f (\chi_{\f12}(g_{3}g_{5}g_{1}^{-1})-2)\,\vphi
\,=\,
 0\,,\nn
\ee
\beq
\int f (\chi_{\f12}(g_{4}g_{5}g_{6})-2)\,\vphi
&=&
\int f(\id,\id,\id,g_{4},g_{5},g_{6})\,
(\chi_{\f12}(g_{4}g_{5}g_{6})-2)\,
\pp_{a}\delta(g_{6})
\pp_{b}\delta(g_{4})
\pp_{c}\delta(g_{5}) \nn\\
&=&
\f i8\eps^{abc}\chi_{\f12}(\sigma_{a}\sigma_{b}\sigma_{c})\,f(\id)\,\ne 0\,.
\eeq
So this triple-grasped flat state is a solution to three holonomy constraints, but the flatness around these three loops does not imply the flatness around the composite loop at the quantum level.

At the end of the day, to fully impose the flatness of the physical state, we require ``redundant'' holonomy constraints: imposing the holonomy constraints on an independent set of loops, as expected at the classical level, is not sufficient anymore at the quantum level to kill all the potential geometry excitations. This is especially relevant for the coarse-graining of the dynamics of loop quantum gravity. A common scenario is that we impose the flatness of the smallest loops, at the ``fundamental'' Planck scale, and that these will induce the flatness of the larger loops at larger scale. However, we see that it is not enough: there exist flat states (distributions peaked on the identity), solutions to the fundamental holonomy constraints, but which couple and entangle the fundamental loops through some grasping operators in such a way that they are not solutions anymore to the holonomy constraints around larger loops. In order to impose the complete flatness  of the state at all scales, we need to impose the flatness of all the loops at all scales, allowing for holonomy constraints around loops of arbitrary size similarly to the construction of Ising-like states for loop quantum gravity defined in \cite{Feller:2015yta}.

\section*{Conclusion \& Outlook}

Following the logic of coarse-graining the quantum geometry of loop quantum gravity, we have extended spin network states by enriching 
the structure of its vertices: we have attached to the vertices new local degrees of freedom so that they effectively represent coarse-grained regions of space with non-trivial gravitational field and geometry fluctuations. To this purpose, we have introduced a hierarchy of three generalization of spin network states, based on the coarse-graining through gauge-fixing approach developed in \cite{Livine:2013gna}: folded, loopy and tagged spin networks.

Folded spin networks on a graph $\Gamma$ are spin network states with an arbitrary number of additional little loops attached to each vertex  of the graph and the extra information of a circuit, or folding tree, for each vertex describing how these little loops are connected to each other and to the edges of the graph. This is a mathematical reformulation of gauge-fixed spin networks where the original states live on finer graphs which have been coarse-grained down to $\Gamma$.
Loopy spin networks forget about the folding trees and describe spin networks on the base graph $\Gamma$ plus the little loops attached to its vertices. These little loops describe local excitations of curvature and geometry, which can then propagate on top of the background geometry defined by the base graph. We have payed special attention to describing the cases of distinguishable and indistinguishable little loops, leading to a definition of a Fock space for loopy spin network with bosonic little loop excitations.
Tagged spin networks are the last step of coarse-graining and define a basis for non-gauge-invariant spin network states. Each vertex carries an extra spin, or tag, represented as living on an open leg attached the vertex, which defines the closure defect, that is how much the local gauge-invariance is broken. From the perspective of coarse-graining, this closure defect provides an overall measure of the non-trivial holonomies which have developed within the coarse-grained region or along the little loops attached to the vertex. Ultimately we would like to interpret this tag as some quasi-local mass or energy density for the (quantum) gravitational field fluctuations inside the coarse-grained region.

\medskip

These structures define a new framework for  loop quantum gravity, where we can implement and study its graph-changing dynamics while working on a fixed background graph. Indeed, starting from a given base graph $\Gamma$, we represent any spin network states on a finer graph as a loopy spin network on the base graph plus little loops taking into account the more complex structure of the original graph. In a  way, we constantly coarse-grain spin networks to our chosen background graph and the little loops represent all the finer geometry excitations. This proposes a truncation of loop quantum gravity where the little loops are effective local degrees of freedom, which can propagate and interact on and with the base graph $\Gamma$. For instance, we could choose as background graph, a regular 3d cubic lattice (e.g. as for defining Bianchi I cosmology as a truncation of loop quantum gravity\cite{Alesci:2013xd}) or any other base graph suited for the case at study, and consider all the finer geometry fluctuations from an effective point of view as little loop inhomogeneities propagating on that base graph, as illustrated on fig.\ref{fig:lattice}.
The little loops are the extra information carried by the loopy spin networks compared to a lattice formulation of loop quantum gravity. They encode an infinity of degrees of freedom attached to each vertex of the graph and describing excitations and fluctuations of the gravitational field.
\begin{figure}[h!]
\centering
\begin{tikzpicture}[scale=0.6]

\coordinate(a1) at (0,1);
\draw (a1) node {$\bullet$};
\draw[in=-30,out=30,scale=1.8,rotate=45] (a1)  to[loop] (a1);

\coordinate(a2) at (1,2);
\draw (a2) node {$\bullet$};
\draw[in=-30,out=30,scale=1.8,rotate=45] (a2)  to[loop] (a2);
\draw[in=-30,out=30,scale=1.8,rotate=135] (a2)  to[loop] (a2);

\coordinate(a3) at (3,3);
\draw (a3) node {$\bullet$};
\draw[in=-30,out=30,scale=1.8,rotate=-135] (a3)  to[loop] (a3);

\coordinate(a4) at (3,0);
\draw (a4) node {$\bullet$};
\draw[in=-30,out=30,scale=1.8,rotate=-135] (a4)  to[loop] (a4);
\draw[in=-30,out=30,scale=1.8,rotate=135] (a4)  to[loop] (a4);

\coordinate(a5) at (2,1);
\draw (a5) node {$\bullet$};
\draw[in=-30,out=30,scale=1.8,rotate=135] (a5)  to[loop] (a5);
\draw[in=-30,out=30,scale=1.8,rotate=-135] (a5)  to[loop] (a5);
\draw[in=-30,out=30,scale=1.8,rotate=45] (a5)  to[loop] (a5);

\coordinate(a6) at (0,0);
\draw (a6) node {$\bullet$};
\draw[in=-30,out=30,scale=1.8,rotate=135] (a6)  to[loop] (a6);

\draw (0,-.5) --(0,3.5);
\draw (1,-.5) --(1,3.5);
\draw (2,-.5) --(2,3.5);
\draw (3,-.5) --(3,3.5);

\draw (-.5,0) --(3.5,0);
\draw (-.5,1) --(3.5,1);
\draw (-.5,2) --(3.5,2);
\draw (-.5,3) --(3.5,3);

\end{tikzpicture}

\caption{Loopy spin networks on a 2d square lattice: the little loop excitations attached to its vertices represent finer graph structures which have been coarse-grained and allow to take into account graph changing dynamics while working on a fixed base graph. }
\label{fig:lattice}

\end{figure}
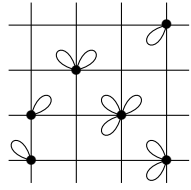

Re-introducing in such a way a background structure  offers a perfect setting for studying the coarse-graining of the dynamics of loop quantum gravity.
In a sense, we have split the gravitational field degrees of freedom into a dynamical background geometry on a fixed graph and localized fluctuations of geometry, which can be though of as higher energy or finer scale excitations. As we coarse-grain, structures of the base graph will become little loops.

\medskip

We have also studied in great detail how to implement the dynamics of BF theory and defined suitable Hamiltonian constraints that select ultimately the flat state as unique physical state. In particular, it should kill any local degree of freedom and project out all the potential little loop excitations (or non-trivial tags). We faced two subtleties. First, the holonomy operators of loop quantum gravity now act also as creation and annihilation operators for the little loops. Second we identified an infinity of solutions to the holonomy constraints, defined as higher derivative of the $\delta$-distribution or equivalently by acting with grasping operators on the totally flat state. These are still peaked exclusively on the flat connection, but introduce some non-trivial correlation and entanglement between the little loop excitations. This comes from the fact that the dimension of the intertwiner space at a vertex grows with the number of little loops and these ``spurious'' solutions can be interpreted as non-trivial intertwiners between flat little loops. We have introduced Laplacian constraints to supplement the holonomy constraints and decouple the loops, allowing us to get rid of this tower of higher order flat states and get finally as wanted a unique physical state. Nevertheless, this Hilbert space of ``grasped flat states'' could prove an interesting sector of loop quantum gravity, with an infinity of  degrees of freedom, especially as a toy model to investigate the continuum limit of the theory.

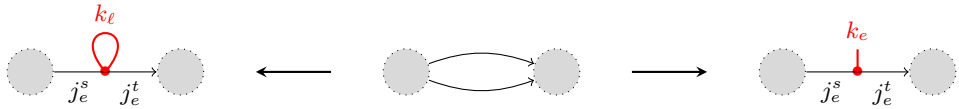
\begin{figure}[h!]
\centering

\begin{tikzpicture}[scale=1]

\node[draw,circle,dotted,fill=gray!30,scale=2] (A) at (0,0) {};
\node[draw,circle,dotted,fill=gray!30,scale=2] (B) at (2,0) {};
\draw[->] (A) to[bend left=20] (B);
\draw[->] (A) to[bend right=20] (B);

\draw[->,>=stealth,thick] (3,0) -- (4,0);

\node[draw,circle,dotted,fill=gray!30,scale=2] (C) at (5,0) {};
\node[draw,circle,dotted,fill=gray!30,scale=2] (D) at (7,0) {};

\coordinate(O) at (6,0);
\draw[red] (O) node {$\bullet$};

\draw(C) -- (O) node[midway,below]{$j_{e}^{s}$};
\draw[->] (O) -- (D) node[midway,below]{$j_{e}^{t}$};
\draw[red,thick] (O) -- ++ (0,0.3) node[above] {$k_{e}$};

\draw[->,>=stealth,thick] (-1,0) -- (-2,0);

\node[draw,circle,dotted,fill=gray!30,scale=2] (E) at (-5,0) {};
\node[draw,circle,dotted,fill=gray!30,scale=2] (F) at (-3,0) {};

\coordinate(P) at (-4,0);
\draw[red] (P) node {$\bullet$};

\draw(E) -- (P) node[midway,below]{$j_{e}^{s}$};
\draw[->] (P) -- (F) node[midway,below]{$j_{e}^{t}$};
\draw[in=-40,out=40,scale=1.8,rotate=90,red,thick] (P)  to[loop] node[midway,above] {$k_{\ell}$}(P);

\end{tikzpicture}

\caption{Two regions, to be coarse-grained, can be related by several links. These links should in turn be coarse-grained into a single edge. However since the holonomies along the original links might have been different, the effective edge should be able to carry some data reflecting this non-trivial curvature. Identifying a spin network edge as a bivalent vertex, it is natural to introduce tagged edges, where the tag $k_{e}$ creates a trivalent intertwiner with the spins at the source and target of the edge, $j_{e}^{s}$ and $j_{e}^{t}$, which can now be different. Another natural possibility is to attach a loop or tadpole to the edge to account for the non-trivial curvature between the two regions, which would also create a defect along the edge leading to different spins at the source and target of the edge. The difference between these two generalizations of spin network edges is that the loop implies that the midway vertex now has a non-trivial volume  while the tag does not.
}
\label{fig:tagged-edge}

\end{figure}

After having set in the present work the kinematics of loopy and tagged spin networks and shown how to implement the topological dynamics of BF theory, the next step will be to define some non-topological loop quantum gravity dynamics, coupling the degrees of freedom on the base graph to the little loops, and study its coarse-graining flow.

There is however another generalization of spin networks worth investigating before moving on to the dynamics of the theory. We have focussed up to now on coarse-graining bounded regions of spaces into effective vertices and therefore introduced the notions of loopy and tagged spin network states with extra structure and data attached to the vertices. We have dressed the graph's nodes, so shouldn't we also consider dressing the its links? Comparing to Feynman diagrams in quantum field theory, we have renormalized the interaction vertices but we should also describe how to renormalize the propagator. Indeed, after partitioning the graph spanning the 3d space into bounded regions and coarse-graining each region to a single vertex, there is generically several edges linking these effective vertices. We need to bundle them together into a single new effective edge. Since those edges to coarse-grain link the same two regions and thus form loops, there is naturally non-trivial holonomies which can formed between the two regions and therefore the new effective edge should be able to carry some notion of  curvature. One way to proceed is to consider all the edges of a spin network state as bivalent intertwiners, recoupling between the two spins living at its source and at its target. Such intertwiners are trivial and the source and target spins are identified. But as curvature builds up, it is natural to allow this intertwiner to acquire a tag or little loops, accounting for the curvature excitations carried by the edge. For instance, as shown on fig.\ref{fig:tagged-edge}, a tag would turn the bivalent intertwiner into a trivalent intertwiner allowing the source and target spins to differ. It would be interesting to develop the notion of spin networks with tagged links. And it could be relevant to compare such spin networks with both tagged vertices and edges to projected $\SL(2,\C)$ spin networks\cite{Livine:2002ak,Dupuis:2010jn}, which allow for both features of non-vanishing closure defect at vertices and non-matching spins along edges, and which are the basic states and building blocks for the EPRL-FK spinfoam models for loop gravity path integrals\cite{Engle:2007wy,Geloun:2010vj,rovelli2014quantum}.

Finally, tags and little loops closely resemble particle insertions on spin network states and it would be enlightening to understand if they can truly be interpreted as matter field degrees of freedom, especially from the perspective of working out the continuum limit of loop quantum gravity as a quantum field theory.
 
\section*{Acknowledgments}

C.C. would like to thank Michel Fruchart and Dimitri Cobb for their keen insights and useful discussions with them.

\appendix

\section{Projective limits of loopy spin networks}
\label{app:ProjectiveLimit}

The general framework of a projective family and the projective limit can be found in \cite{Ashtekar:1994mh}, where it is  applied to define the kinematical Hilbert space of spin network states for loop quantum gravity.
Here we apply  these definitions to loopy spin networks,  in order to define superposition states of potentially an infinite number of little loops. To this purpose, we focus on the flower graph, with a single vertex and an arbitrary number of loops attached to that central node.

In order to define precisely this idea of varying number of loops, we start with wave-functions over a finite number of loops and define a nesting, that is describe how to include a set of loops inside a larger one.
We will identity the set of all potential loops with the set of integers. Finite sets of loops are defined as finite subsets of integers.
Loops are labeled by the integers and are a priori distinguishable. For instance, a wave-function with the support on the loop number $2$ and a wave-function on the loop number $277$ are not the same though they both are one-loop states and depend on only one variable, as illustrated in fig.\ref{fig:potentialloops}.

\begin{figure}[h!]

\centering

\begin{tikzpicture}
\coordinate(O1) at (0,0);
\coordinate(O2) at (3,0);

\draw (O1) to[loop,scale=3] (O1) ++(0,1.2) node {$\circled{2}$};
\draw[dashed,gray] (O1) to[loop,scale=3,rotate=180] (O1) ++(0,-1.2) node {$\circled{277}$};
\draw (O1) node[scale=1] {$\bullet$};

\draw (1.5,0) node[scale=2] {$\neq$};

\draw[dashed,gray,scale=3] (O2) to[loop] (O2);
\draw[dashed,gray] (O2) ++(0,1.2) node {$\circled{2}$};
\draw[scale=3,rotate=180] (O2) to[loop] (O2);
\draw (O2) ++(0,-1.2) node {$\circled{277}$};
\draw (O2) node[scale=1] {$\bullet$};

\end{tikzpicture}

\caption{We distinguish the different potential loops and therefore consider the resulting wave-functions as different even when they excite the same number of loops.}
\label{fig:potentialloops}

\end{figure}
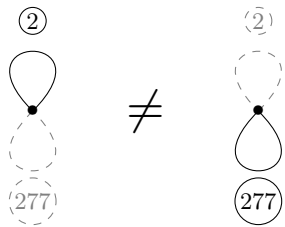

Mathematically, we consider the set of all finite subsets of integers $\mathcal{P}_{<\infty}(\mathbb{N})$. To each subset $E \in \mathcal{P}_{<\infty}(\mathbb{N})$, we associate the set $\SU(2)^{E}$ of colorings of the corresponding loops by $\SU(2)$ group elements. Then wave-functions on $E$ are gauge-invariant functions over $\SU(2)^{E}$:
\begin{equation}
\{\Psi : \SU(2)^E \rightarrow \mathbb{C}~:~ \forall g \in \SU(2),~
\Psi(\{gh_{\ell}g^{-1}\}_{\ell\in E}) = \Psi(\{h_{\ell}\}_{\ell\in E})\}\,.
\end{equation}
Defining the scalar product using the Haar measure over $\SU(2)^{E}$, the Hilbert space $\cH_{E}$ of quantum states on the loopy spin network defined by the subset $E$ of loops is the $L^{2}(\SU(2)^{E}/\textrm{Ad}\SU(2))$.

The space of loops is equipped with a partial directed order given by the inclusion of subsets of integers. The partial directed order encodes how different subsets are nested within one another: a wave-function over the loop number $2$ and a wave-function over the loop number $277$ are different but they are both embedded in the larger class of wave-functions which depend on both loop number $2$ and loop number $277$ as illustrated in fig.\ref{fig:Embedding}.

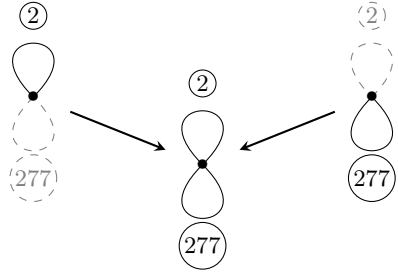
\begin{figure}[h!]

\centering

\begin{tikzpicture}[scale=0.9]
\coordinate(O1) at (0,0);
\coordinate(O2) at (5,0);
\coordinate(O3) at (2.5,-1);

\draw (O1) to[loop,scale=3] (O1) ++(0,1.2) node {$\circled{2}$};
\draw[dashed,gray] (O1) to[loop,scale=3,rotate=180] (O1) ++(0,-1.2) node {$\circled{277}$};
\draw (O1) node[scale=1] {$\bullet$};

\draw[dashed,gray,scale=3] (O2) to[loop] (O2);
\draw[dashed,gray] (O2) ++(0,1.2) node {$\circled{2}$};
\draw[scale=3,rotate=180] (O2) to[loop] (O2);
\draw (O2) ++(0,-1.2) node {$\circled{277}$};
\draw (O2) node[scale=1] {$\bullet$};

\draw[scale=3] (O3) to[loop] (O3);
\draw (O3) ++(0,1.2) node {$\circled{2}$};
\draw[scale=3,rotate=180] (O3) to[loop] (O3);
\draw (O3) ++(0,-1.2) node {$\circled{277}$};
\draw (O3) node[scale=1] {$\bullet$};

\draw[<-,>=stealth,thick, shorten >=15pt, shorten <=15pt] (O3) -- (O2);
\draw[<-,>=stealth,thick, shorten >=15pt, shorten <=15pt] (O3) -- (O1);

\end{tikzpicture}

\caption{Though different, two functions over two different loops can be embedded in a larger space of functions depending on several loops by identifying them with functions with trivial dependancy on some loops.}
\label{fig:Embedding}

\end{figure}

This partial ordering by inclusion of subsets induces a projective structure on the loop colorings by group elements.
We define a projector $p_{EE'}$ defined for every pair of subsets $(E,E')$ such that $E \subseteq E'$ by:
\begin{eqnarray*}
p_{EE'} : \SU(2)^{E'} &\rightarrow& \SU(2)^{E} \\
s &\mapsto& \restriction{s}{E}
\end{eqnarray*}
This projector is simply the canonical restriction from the larger subset $E'$ to the smaller subset $E$. These projectors satisfy a key transitivity property:
\begin{equation}
\forall E,E',E'', \quad E\subset E'\subset E''\, \Longrightarrow p_{E'E''} \circ p_{EE'} = p_{EE''}\,,
\end{equation}
so that the couple of sets $(\SU(2)^{E},p_{EE'})_{E,E' \in \mathcal{P}_{<\infty}(\mathbb{N})}$ form what is called a {\it projective family}.
The projective limit $\overline{\SU(2)}$ is then defined by:
\begin{equation}
\overline{\SU(2)} = \Big{\{}
(g_{E})_{E \in \mathcal{P}} \in \times_{E\in \mathcal{P}} \SU(2)^{E} ~:~ E \subseteq E' \Rightarrow p_{EE'} g_{E'} = g_{E}
\Big{\}}
\end{equation}
Intuitively, this corresponds to collections of colorings on all possible subsets of loops which are compatible with each other with respect to the inclusion. Therefore, these compatibility conditions between all finite samplings of the collection, as illustrated in fig.\ref{fig:Compatibility}, is the precise implementation of the notion of a coloring of an infinite number of loops.

\begin{figure}[h!]

\centering

\begin{tikzpicture} [scale=0.9]
\coordinate(O1) at (0,0);
\coordinate(O2) at (-6,-4);
\coordinate(O3) at (-2,-4);
\coordinate(O4) at (2,-4);
\coordinate(O5) at (6,-4);

\draw[red,scale=3] (O1) to[loop] (O1);
\draw[red] (O1) ++(0,1.2) node {$\circled{2}$};
\draw[blue,scale=3,rotate=-90] (O1) to[loop] (O1);
\draw[blue] (O1) ++(1.2,0) node {$\circled{277}$};
\draw[green,scale=3,rotate=90] (O1) to[loop] (O1);
\draw[green] (O1) ++(-1.2,0) node {$\circled{42}$};
\draw (O1) node[scale=1] {$\bullet$};

\draw[dashed,red!50,scale=3] (O2) to[loop] (O2);
\draw[dashed,red!50] (O2) ++(0,1.2) node {$\circled{2}$};
\draw[blue,scale=3,rotate=-90] (O2) to[loop] (O2);
\draw[blue] (O2) ++(1.2,0) node {$\circled{277}$};
\draw[green,scale=3,rotate=90] (O2) to[loop] (O2);
\draw[green] (O2) ++(-1.2,0) node {$\circled{42}$};
\draw (O2) node[scale=1] {$\bullet$};

\draw[red,scale=3] (O3) to[loop] (O3);
\draw[red] (O3) ++(0,1.2) node {$\circled{2}$};
\draw[dashed,blue!50,scale=3,rotate=-90] (O3) to[loop] (O3);
\draw[dashed,blue!50] (O3) ++(1.2,0) node {$\circled{277}$};
\draw[green,scale=3,rotate=90] (O3) to[loop] (O3);
\draw[green] (O3) ++(-1.2,0) node {$\circled{42}$};
\draw (O3) node[scale=1] {$\bullet$};

\draw[dashed,red!50,scale=3] (O4) to[loop] (O4);
\draw[dashed,red!50] (O4) ++(0,1.2) node {$\circled{2}$};
\draw[dashed,blue!50,scale=3,rotate=-90] (O4) to[loop] (O4);
\draw[dashed,blue!50] (O4) ++(1.2,0) node {$\circled{277}$};
\draw[green,scale=3,rotate=90] (O4) to[loop] (O4);
\draw[green] (O4) ++(-1.2,0) node {$\circled{42}$};
\draw (O4) node[scale=1] {$\bullet$};

\draw (O5) node[scale=2] {$\cdots$};

\draw[->,>=stealth,thick] ($(O1)+(-1.5,-1)$) -- ($(O2)+(1,1)$);
\draw[->,>=stealth,thick] ($(O1)+(-0.5,-1.5)$) -- ($(O3)+(0.5,1)$);
\draw[->,>=stealth,thick] ($(O1)+(0.5,-1.5)$) -- ($(O4)+(-0.5,1)$);
\draw[->,>=stealth,thick] ($(O1)+(1.5,-1)$) -- ($(O5)+(-1,1)$);

\end{tikzpicture}

\caption{The projective limit is made of  collections of colorings of finite subsets equipped with compatibility conditions: the coloring of a finite subset is the projection of the coloring of any larger subset.}
\label{fig:Compatibility}

\end{figure}
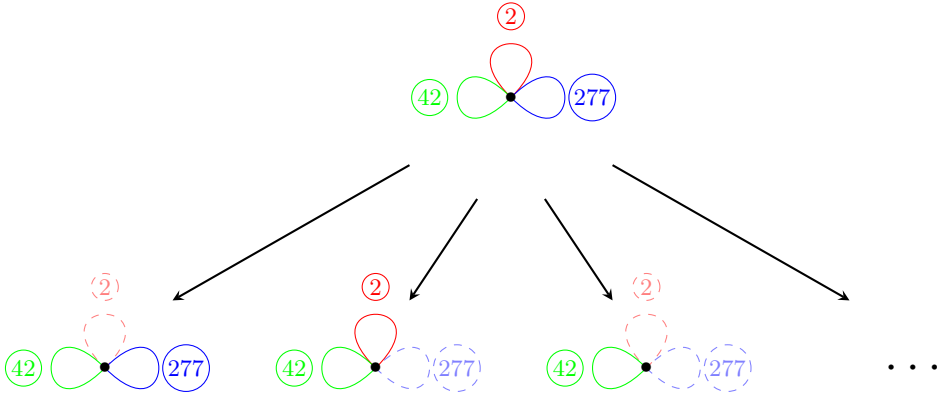

We translate the projective structure to the space of wave-functions. The projectors $p_{EE'}$ for $E\subset E'$ turn into injections $I_{EE'}\equiv p^{*}_{EE'}$ defined by their pull-backs:
\begin{eqnarray}
I_{EE'} : \cH_{E} &\rightarrow& \cH_{E'}\\
f &\mapsto& p^{*}_{EE'}f \,\,:\, p^{*}_{EE'}f \big{(}\{g_{\ell}\}_{\ell\in E'}\big{)}=f \big{(}\{\{g_{\ell}\}_{\ell\in E}\}\big{)}\,,\nn
\end{eqnarray}
where $p^{*}_{EE'}f$ trivially depends on group elements living on loops of $E'$ which do not belong to the smaller set $E$.
The compatibility conditions translates into an equivalence relation:
\begin{equation}
f_{E_1} \sim f_{E_2} \Leftrightarrow \forall E_3\in\mathcal{P},~E_1 \subseteq E_3,~E_2\subseteq E_3,~p^*_{E_1 E_3} f_{E_1} = p^*_{E_2E_3} f_{E_2}
\end{equation}
This allows to define wave-functions on the projective limit of the loop colorings $\overline{\SU(2)}$ and  give a precise sense to functions over an infinite number of loops:
\begin{equation}
\mathcal{H}=\left(\bigcup_{E\in\mathcal{P}} \mathcal{H}_{E}\right)/\sim
\end{equation}

In order to make this projective limit less abstract and easier to handle, we use another representation. For every equivalence class of wave-functions in the projective limit $\cH$, let us remove all the trivial dependency and pick its representant based on the smallest subset of loops. So, in practice, we define spaces of ``proper states'', i.e. wave-functions that have no trivial dependency:
\begin{eqnarray}
\mathcal{H}_{E}^0 =
\{
f \in \mathcal{H}_{E}\,\, :\, \forall \ell\in E\,,\,\int_{\SU(2)}\dd h_{\ell}\,f=0
\}\,.
\end{eqnarray}
This is the space of functions really defined on the subset $E$, with an actual dependance on each loop  and no constant term. The integral condition removes all the spin-0 components of the wave-functions.
First, we show that an arbitrary wave-function over the subset $E$ of loops can be fully decomposed into proper states with support on all the subsets of $E$:
\begin{lemma}
The following  isomorphism holds as a pre-Hilbertian space isomorphism:
\begin{equation}
\forall E \in \mathcal{P}_{<\infty}(\mathbb{N}),
\quad
\mathcal{H}_{E} \simeq \bigoplus_{F \in \mathcal{P}(E)} \mathcal{H}_{F}^0\,,
\end{equation}
where the direct sum is over all subsets $F\subset E$. This isomorphism is realized  through the projections $f_{F}=P_{E,F}f$ acting on wave-functions $f\in\cH_{E}$:
\be
f_{F}\big{(}
\{h_{\ell}\}_{\ell\in F}
\big{)}
\,=\,
\sum_{\tF\subset  F}
(-1)^{\#\tF}
\int \prod_{\ell\in E\setminus F}\mathrm{d}g_{\ell}
\prod_{\ell\in\tF}\mathrm{d}k_{\ell}\,
f\big{(}
\{h_{\ell}\}_{\ell\in F\setminus \tF},
\{k_{\ell}\}_{\ell\in\tF},
\{g_{\ell}\}_{\ell\in E\setminus F}
\big{)}\,.
\ee
Its inverse is the re-summation of the projections:
\be
f=\sum_{F\subset E} f_{F}\,.
\ee

\end{lemma}
\begin{proof}
We proceed in two steps. First we check that each projection $f_{F}$ is a proper state,
$$
\forall\ell\in F\,,\quad\int \mathrm{d}h_{\ell}\,f_{F}=0\,,
$$
and that re-summing these projections $\sum_{F\subset E} f_{F}$ yields $f$.
Second, we check that the integral condition, ensuring that there is no spin-0 mode, also implies that the subspaces $\mathcal{H}_{F}^0$ are pairwise orthogonal, which concludes the proof.
\end{proof}

This decomposition generalizes to the projective limit:
\begin{prop}
The following  isomorphism holds as a pre-Hilbertian space isomorphism:
\begin{equation}
\mathcal{H} \simeq \bigoplus_{E \in \mathcal{P}_{<\infty}(\mathbb{N})} \mathcal{H}_{E}^0\,.
\end{equation}
\end{prop}
\begin{proof}
If $(f_E)_{E \in \mathcal{P}_{<\infty}(\mathbb{N})}$ is in $\bigoplus_{E \in \mathcal{P}_{<\infty}(\mathbb{N})} \mathcal{H}_{E}^0$, we  define the set of subsets on which the state $f$ does not vanish:
\begin{equation}
C_f = \{ E \in \mathcal{P}_{<\infty}(\mathbb{N}) : f_E \neq 0 \}\,.
\end{equation}
By definition of the direct sum, $C_f$ is finite, so we can define the finite subset $F = \cup_{E \in C_f} E$ and the re-summation map:
\begin{eqnarray}
\phi : \bigoplus_{E \in \mathcal{P}_{<\infty}(\mathbb{N})} \mathcal{H}_{E}^0 &\rightarrow& \mathcal{H} \\
(f_E)_{E \in \mathcal{P}_{<\infty}(\mathbb{N})} &\mapsto& \left[\sum_{E \in C_f} f_E\right]
\nn
\end{eqnarray}
where the brackets refer to the equivalence class of the function.
This map is obviously  linear and we now look for a definition of its inverse.
So let us consider a state in  $\mathcal{H}$, that is an equivalence class $s$. We define the set of subsets of loops on which it has support:
\begin{equation}
D_s = \{ E \in \mathcal{P}_{<\infty}(\mathbb{N}) : \exists f \in s,~f \in \mathcal{H}_{E} \}
\end{equation}
Then we consider the smallest set in $D_s$ , which can be defined\footnotemark as the intersection $F_s = \bigcap_{E \in D_s} E$.
\footnotetext{This is the point where we choose not to use the completion and just have an isomorphism of pre-Hilbertian spaces in order to have the existence of $F_f$.}.
In a sense, this is the minimal support of the state $s$.
We choose a representative $f^{s}$ of the equivalence class $s$ in $F_{s}$. It is actually unique by definition of the equivalence relation.
Then we consider the decomposition in proper states of $f^{s}$ over all subsets $F$ of $F_{s}$ and define:
\begin{eqnarray}
\psi : \mathcal{H} &\rightarrow& \bigoplus_{E \in \mathcal{P}_{<\infty}(\mathbb{N})} \mathcal{H}_{E}^0 \\
~s &\mapsto& \sum_{F\subset F_{s}}P_{F_{s},F} f^{s} \nn
\end{eqnarray}
It is direct to check that it is indeed the inverse of $\phi$.

\end{proof}

This decomposition into proper states is very useful to visualize the space: each wave-function can be decomposed into a sum of wave-functions over a finite number of loops but with no trivial dependancy. This gives a precise meaning to superpositions of numberss of loops.

\section{Holonomy Constraint on $\SU(2)$}
\label{app:deltaSU2}

\subsection{Distributions on $\SU(2)$ and conjugation-invariant solutions}
\label{app:distribution}

We would like to impose the holonomy constraints for BF theory which read for a single group element:
\begin{equation}
\forall g \in \mathrm{SU}(2),\,\,
\hchi\vphi\,(g)=
\chi_\frac{1}{2}(g) \vphi(g) = 2\,\ ,\vphi(g)\,.
\end{equation}
If we stay in the strict framework of the Hilbert space $L^{2}(\SU(2))$, no square integrable function actually provides such an eigenvector for $\hchi$ and we should solve this equation in the dual space. As is standard in  quantum mechanics, the natural framework for solving the equation is a rigged-Hilbert space (or Gelfand triple), that is a triplet: $\mathcal{S} \subset \mathcal{H} \subset \mathcal{S}^*$. The space $\mathcal{H}$ is the Hilbert space. The smaller space $\mathcal{S}$ is provided with a stronger topology than the induced one and can thought of as the test function space, while its dual $\mathcal{S}^{*}$ is the  space of continuous linear forms on $\cS$ and defines the distribution space.
The major property of $\mathcal{S}$ is to be small enough for the algebra of observables to be defined over it. Then the operator algebra can be naturally extended on $\mathcal{S}^{*}$ and thus on $\mathcal{H}$. For instance, an operator $A$ defined on $\mathcal{S}$ acts on a (dual) state $\varphi$ be in $\mathcal{S}^*$ as:
\begin{equation}
\forall f \in \mathcal{S},~A\varphi(f) = \varphi(A^\dagger f)\,.
\end{equation}

So let us be explicit for functions over $\SU(2)$. The Hilbert space $\mathcal{H}$ is the space of square-integrable functions.  The space $\mathcal{S}$  is usually chosen to be the Schwarz space so that canonical position and momentum operator can be defined. Here, the rapid fall-off condition is not needed since we are dealing with a compact group, but we keep the smoothness requirement:
\begin{equation}
\mathcal{S} = \{\psi \in \mathcal{H} / \psi \in \mathcal{C}^\infty\}
\end{equation}
%
Regarding the topology, the space $\mathcal{S}$ is naturally endowed with the convergence on every $\mathcal{C}^k$ space. 
More precisely  the space of $\mathcal{C}^k$ functions is equipped with the following norm:
\begin{equation}
\|f\|_{\mathcal{C}^k} = \sup_{0 \le i \le k} \|\partial_{a_1,...,a_i} f\|_\infty
\end{equation}
This norm has two nice properties. First, differentiation is continuous from $\mathcal{C}^k$ to $\mathcal{C}^{k-1}$. Second, the topology induced by the norms are finer as $k$ goes to infinity. So the limit topology on $\mathcal{S} = \bigcap_{k \in \mathbb{N}} \mathcal{C}^k$goes as follows: a sequence of functions $f_{n\in\N}$ in $\mathcal{S}$ admits  $0$ as its limit if the sup-norm of all its derivatives $\|\partial_\alpha f_{n}\|_{\infty}$ go to $0$ for arbitrary multi-index $\alpha$. This is topology is naturally finer than all the $\mathcal{C}^k$ topologies and the differentiation is still continuous.  Provided with this topology, $\mathcal{S}$ is a Frechet space: it is complete and metrizable (though no norm is defined). Note that, although all the $\mathcal{C}^k$ are Banach spaces, their descending intersection $\bigcap_{k \in \mathbb{N}} \mathcal{C}^k$ is not.

\medskip

Things are usually clearer and more explicit in the Fourier decomposition. Let us consider the Fourier decomposition of a function over $\SU(2)$ on the Wigner matrices:
$$
f(g)=\sum_{j,a,b}f^{j}_{ab}D^{j}_{ab}\,.
$$
By the Fourier convergence theorem, smoothness actually translates into a rapid fall-off of the  Fourier coefficients:
\begin{equation}
f\in\cS
\quad\Longleftrightarrow\quad
\forall K\in\N\,,\,\,\sum_{j} |f^{j}| d_j^K < +\infty\,,
\end{equation}
where $d_{j}=(2j+1)$ is the dimension of the spin-$j$ representation and $|f^{j}|$ can equally be the sup-norm or the square-norm of the matrix $f^{j}_{ab}$.
This also means that the Fourier coefficients of a distribution  cannot diverge faster than polynomially:
\be
\vphi\in\cS^{*}
\quad\Longleftrightarrow\quad
\exists K\,,\,\,\sum_{j} |\vphi^{j}| d_j^{-K} < +\infty\,.
\ee
The strong topology on $\cS$ means that a sequence of smooth functions $f_{n}$ converges to 0 in $\cS$ if and only if all the $K$-power sums go to 0:
\be
\lim_{n\arr\infty}f_{n}=0
\quad\Longleftrightarrow\quad
\forall K\in\N\,,\,\,\sum_{j} |f^{j}_{n}| d_j^K \,\,\underset{n\arr\infty}\longrightarrow 0\,.
\ee
This ensures that the evaluations of a distribution $\vphi$ will also converge $\vphi(f_{n})\arr0$.

\medskip

Let us apply this to the holonomy constraints for functions invariant under conjugation on $\SU(2)$. In this case, all functions decompose on the characters,
$$
\vphi(g)=\sum_{j\in\f\N2}\vphi_{j}\chi_{j}(g)\,,
$$
and the eigenvalue problem $\chi(g)\vphi(g)=2\vphi(g)$ translates into a recursion relation on the Fourier coefficients:
\be
\rho \vphi_{0}=\vphi_{\f12}\,,\quad
\rho \vphi_{j\ge\f12}=\vphi_{j-\f12}+\vphi_{j+\f12}\,.
\ee
Once we fix the initial condition $\vphi_{0}$, this recursive equation has a solution for every complex value $\rho\in\C$, but this does not systematically defines a solution state, in $L^{2}$ or a distribution. The solution to the recursion is given in terms of the two solutions $\mu_{\pm}$ of the quadratic equation $\mu^{2}-\rho \mu+1=0$:
\be
\forall j\in\f\N2\,,\,\,
\vphi_{j}=\f{\mu_{+}^{2j+1}-\mu_{-}^{2j+1}}{\mu_{+}-\mu_{-}}\,.
\ee
For $\rho=2$, the discriminant $(\rho^{2}-4)$ vanishes and this ansatz fails leads: instead of the power law, we get a linear growth $f_{j}=(2j+1)$, which leads back to the $\delta$-distribution peaked on the identity. For real values $|\rho|< 2$, the discriminant is negative and we get an oscillatory solution. Mapping $\rho$ to an angle $\theta$ defined as $\rho=2\cos\theta$, the two solutions are $\mu_{\pm}=\exp(\pm i \theta)$ and the Fourier coefficients are $\vphi_{j}=\sin(2j+1)\theta\,/\sin\theta\,=\chi_{j}(\theta)$, which gives a $\delta$-distribution fixing the class angle of the group element $g$ to $\theta$. For $|\rho|>2$, the positive discriminant will leads to exponentially divergent coefficients $\vphi_{j}$ and do not define a proper distribution.

\subsection{The holonomy constraint as a recursion relation}
\label{app:recursion}

Let us write the holonomy constraint equation $\hchi\,\vphi=2\,\vphi$ for distributions on $\SU(2)$ in the Fourier decomposition. We decompose the distribution $\vphi$ on the Wigner matrices:
$$
\vphi(h)=\sum_{j,m,n} (2j+1)\vphi^{j}_{nm}\,D^{j}_{mn}(h)\,,
$$
with the magnetic moment indices $-j\le m,n\le +j$.
We know from spin recoupling the action of a spin-$\f12$ Wigner matrix on an arbitrary spin (for example calculating the corresponding Clebsh-Gordan coefficients using the Schwinger representation of the $\su(2)$ Lie algebra in terms of a pair of harmonic oscillators):
\beq
\forall A,B=\pm\,,
\quad
D^{\f12}_{\f A2,\f B2}(h)\,D^{j}_{m,n}(h)
\,=
&
\f1{2j+1}\Bigg{[}&
AB\sqrt{(j+Am+1)(j+Bn+1)}\,D^{j+\f12}_{m+\f A2,n+\f B2}(h) \nn\\
&&+\,
\sqrt{(j-Am)(j-Bn)}\,D^{j-\f12}_{m+\f A2,n+\f B2}(h)
\Bigg{]}\,,
\eeq
with the obvious action on the trivial spin $j=0$ mode.
The action of the spin-$\f12$ character is obtained by adding the two operators for $A=B=\pm$:
$$
\chi_{\f12}(h)
\,=\,
D^{\f12}_{+\f 12,+\f 12}(h)+D^{\f12}_{-\f 12,-\f 12}(h)\,.
$$
We can then translate the functional equation $\chi_{\f12}\,\vphi=2\,\vphi$ into a recursion relation on the Fourier coefficients $\vphi^{j}_{nm}$:
$$
2\vphi^{0}=
\vphi^{\f12}_{-\f12,-\f12}
+
\vphi^{\f12}_{+\f12,+\f12}\,,
$$
\beq
2(2j+1)\vphi^{j}_{n,m}
&=&
\vphi^{j-\f12}_{n-\f12,m-\f12}\sqrt{(j+m)(j+n)}
+
\vphi^{j+\f12}_{n-\f12,m-\f12}\sqrt{(j-m+1)(j-n+1)}
\nn\\
&&
+\,
\vphi^{j-\f12}_{n+\f12,m+\f12}\sqrt{(j-m)(j-n)}
+
\vphi^{j+\f12}_{n+\f12,m+\f12}\sqrt{(j+m+1)(j+n+1)}\,.
\eeq
We easily check that $\vphi^{j}_{nm}=\delta_{nm}=D^{j}_{nm}(\id)$ is a solution to these recursion relations, which corresponds to the $\delta$-distribution on the group $\vphi(h)=\delta(h)$. But straightforward computations also tell us that $\vphi^{j}_{nm}=D^{j}_{nm}(J_{a})$ for $a=z,\pm$ are also solutions, or explicitly:
\be
\vphi^{j}_{nm}=\,\,
\left|\begin{array}{l}
D^{j}_{nm}(J_{z})=m\delta_{n,m}\\
D^{j}_{nm}(J_{+})=\sqrt{(j+m+1)(j-m)}\delta_{n,m+1}\\
D^{j}_{nm}(J_{-})=\sqrt{(j-m+1)(j+m)}\delta_{n,m-1}
\end{array}\right.
\ee
In general, we can show that for every vector $\vx\in\R^{3}$, the coefficients $\vphi^{j}_{nm}=D^{j}_{nm}(\vx\cdot\vJ)$ solve the recursion equation. Indeed, one differentiates $\chi_{\f12}(h)D^{j}_{nm}(h)$ and evaluates it at the identity,
\be
\left.-i\pp_{x}\chi_{\f12}(h)D^{j}_{nm}(h)\right|_{\id}
\,=\,
\chi_{\f12}(\id)D^{j}_{nm}(\vx\cdot\vJ)+\chi_{\f12}(\vx\cdot\vJ)D^{j}_{nm}(\id)
=2D^{j}_{nm}(\vx\cdot\vJ)
\,.
\ee
Then one applies the recoupling relations given above to decompose $\chi_{\f12}D^{j}_{nm}$ onto the Wigner matrices for spins $j\pm\f12$ in order to prove that $\vphi^{j}_{nm}=D^{j}_{nm}(\vx\cdot\vJ)$ do satisfy the required recursion relation.
These coefficients actually correspond to the first derivative of the $\delta$-distribution, $\vphi(h)=\pp_{x}\delta(h)$.

\medskip

In order to classify all the solutions to these recursion relations,  we need to determine  the initial data required to implement the recursion.
The problem is that the size of the matrices $\vphi^{j}$ increases with the mode $j$. If we start with fixing $\vphi^{0}$, then this determines only one component  out of four of the next matrix $\vphi^{\f12}$, thus leaving three new initial conditions to be freely chosen. The $4=(1+3)$ solutions corresponding to those required initial conditions are $\delta$ and the $\pp_{a}\delta$'s.

Moving up to the spin $j=1$ components, we have four recursion relations determining the matrix elements $\vphi^{1}$ in terms of the $\vphi^{\f12}$ matrix (indeed one relation per matrix elements of $\vphi^{\f12}$). This leaves us with $5=(9-4)$ undetermined matrix elements, which we need to specify as initial conditions. Choosing vanishing $\vphi^{0}$ and $\vphi^{\f12}$ initial conditions, these five new initial conditions correspond to the five linearly-independent second-order-derivative solutions of the holonomy constraint: $\pp_{a}\pp_{b}$ and $\pp_{a}\pp_{a}-\pp_{b}\pp_{b}$ for $a\ne b$. So in total, we need to specify 9 initial conditions up to the spin 1 modes. And this will inexorably grow as we explore higher and higher spins, with $4j+1=[(2j+1)^{2}-(2j)^{2}]$ new initial conditions for the matrix $\vphi^{j}$, thus reproducing the infinite dimensional space of solutions of the holonomy constraint, with the tower of higher and higher derivatives.

\subsection{Laplacian constraint, double recursion and flatness equations}
\label{app:Laplacian}

We introduce another constraint supplementing the holonomy constraint in order to truly impose flatness and get the $\delta$-distribution as unique solution: we impose the Laplacian constraint $(\tDelta-\Delta)=0$, where $\Delta=\pp^{L}_{a}\pp^{L}_{a}=\pp^{R}_{a}\pp^{R}_{a}$ is the usual Laplacian operator and $\tDelta=\pp^{L}_{a}\pp^{R}_{a}$ mixes the right and left derivations. These two operators do not change the spin $j$ and act rather simply on the Wigner matrices:
\be
\Delta D^{j}_{mn}(h)=-D^{j}_{mn}(hJ_{a}J_{a})=-j(j+1)D^{j}_{mn}(h)
\,,
\qquad
\Delta D^{j}_{mn}(h)=-D^{j}_{mn}(J_{a}hJ_{a})\,.
\ee
Using the explicit action of the three $\su(2)$ generators, and applying the Cauchy-Schwarz inequality to bound the sums, we can check that the operator $(\tDelta-\Delta)$ is positive.
We can also translate the Laplacian constraint $\Delta \vphi =\tDelta\vphi$ into equations on the Fourier coefficient matrices $\vphi^{j}$:
\beq
j(j+1)\vphi^{j}_{nm}&=nm\,\vphi^{j}_{nm}
&+\,\,\f12\vphi^{j}_{n-1,m-1}\sqrt{(j+n)(j-n+1)(j+m)(j-m+1)}\nn\\
&&+\,\,\f12\vphi^{j}_{n+1,m+1}\sqrt{(j-n)(j+n+1)(j-m)(j+m+1)}\,.
\label{Lrecursion}
\eeq
This is a recursion at fixed spin $j$ on the matrix elements of each $\vphi^{j}$ independently. It works at fixed $(n-m)$, that is along the diagonals of the matrix, determining the matrix elements from, say, the highest weight components:
$$
\begin{array}{ll}
\vphi^{j,j}&\arr \,\vphi^{j-1,j-1}\arr\vphi^{j-2,j-2}\arr\dots \\
\vphi^{j,j-1}&\arr\, \vphi^{j-1,j-2}\arr\vphi^{j-2,j-3}\arr\dots \\
\vphi^{j,j-2}&\arr \,\vphi^{j-1,j-3}\arr\vphi^{j-2,j-4}\arr\dots 
\end{array}
$$
Putting this constraint with the holonomy constraint, we get a double recursion structure. The holonomy constraint realizes a recursion on the spin $j$, determining the matrix $\vphi^{j}$  from the lower spin matrices, while the Laplacian constraint implements a recursion on the magnetic moment $m$ within each matrix $\vphi^{j}$:
$$
\vphi^{0} \arr
\mat{cc}{\searrow & \searrow \\
\searrow & \searrow }
\arr
\mat{ccc}{\searrow & \searrow & \searrow \\
\searrow & \searrow & \searrow \\ \searrow & \searrow & \searrow  }
\arr
\mat{cccc}{\searrow & \searrow & \searrow & \searrow \\
\searrow & \searrow & \searrow & \searrow \\ \searrow & \searrow & \searrow & \searrow \\
 \searrow & \searrow & \searrow & \searrow}
 \arr\dots
$$
This allows to solve the problem of the infinite initial conditions needed for the holonomy constraint.
We easily check that $\vphi_{nm}=\delta_{nm}$ is a solution:
$$
j(j+1)=m^{2}+\f12(j+m)(j-m+1)+\f12(j-m)(j+m+1)\,.
$$
But then, we completely solve the constraint and show that it implies that the function is invariant under conjugation:
\begin{prop}
Let us consider the Laplacian constraint $(\Delta-\tDelta)\,\vphi=0$ translated to the Fourier decomposition $\vphi(h)=\sum_{j,m,n}(2j+1)\tr\,\vphi^{j}D^{j}(h)$. Then the each of the Fourier coefficient matrix $\vphi^{j}$ at fixed spin $j$ is proportional to the identity. This means that $\vphi$ is invariant under conjugation.
\end{prop}
\begin{proof}
Let us fix $j\ge \f12$. The spin-0 component $\vphi^{0}$ is unconstrained and left free.
The recursion relation \eqref{Lrecursion} allows to start with an element $\vphi_{j,j-M}$ with $0\le M \le 2j$ and to determine all of the following components along the corresponding diagonal, $\vphi_{j-N,j-M-N}$ for $0\le N\le (2j-M)$.
One actually gets a relation in terms of combinatorial factors:
\be
\vphi_{j-N,j-M-N}=
\vphi_{j,j-M}\, \sqrt{\f{(M+N)!}{M!N!}\,\f{(2j)!(2j-M-N)!}{(2j-M)!(2j-N)!}}
\,.
\ee
In particular, one obtains for the other end of the diagonal:
\be
\vphi_{-j+M,-j}
\,=\,
\vphi_{j,j-M}\,\f{(2j)!}{M!(2j-M)!}\,.
\ee
The trick is that the recursion relation \eqref{Lrecursion} is symmetric under the exchange $(n,m)\leftrightarrow (-m,-n)$: we start now from the other end of the same diagonal and work our way back to the initial top element. Therefore the previous equality holds but in the opposite  way:
\be
\vphi_{j,j-M}
=
\vphi_{-j+M,-j}\,\f{(2j)!}{M!(2j-M)!}
=
\vphi_{j,j-M}\left(\f{(2j)!}{M!(2j-M)!}\right)^{2}\,,
\ee
\be
\textrm{thus either}
\quad
\f{(2j)!}{M!(2j-M)!}=1
\quad\textrm{or}\quad
\vphi_{j,j-M}=0\,.\nn
\ee
In the special case $M=0$, along the principal diagonal, the recursion relation simplifies and reads $\vphi^{j,j}=\vphi^{j-1,j-1}=\vphi^{j-2,j-2}=\dots$. For all the other cases $M\ne 0$, the matrix elements must vanish. This concludes the proof that the matrix $\vphi^{j}$ must be proportional to the identity.

\end{proof}

\bibliographystyle{bib-style}
\bibliography{biblio}

\providecommand{\href}[2]{#2}\begingroup\raggedright\begin{thebibliography}{10}

\bibitem{rovelli2007quantum}
C.~Rovelli, {\em Quantum Gravity}.
\newblock Cambridge Monographs on Mathematical Physics. Cambridge University
  Press, 2007.

\bibitem{thiemann2007modern}
T.~Thiemann, {\em Modern Canonical Quantum General Relativity}.
\newblock Cambridge Monographs on Mathematical Physics. Cambridge University
  Press, 2007.

\bibitem{Gambini:2011zz}
R.~Gambini and J.~Pullin, {\em {A first course in loop quantum gravity}}.
\newblock Oxford University Press,
2011.
\newblock

\bibitem{PhysRevLett.57.2244}
A.~Ashtekar, ``New Variables for Classical and Quantum Gravity,'' Phys. Rev.
  Lett. {\bf 57} (Nov, 1986) 2244--2247.

\bibitem{Barbero:1994ap}
J.~F. Barbero~G., ``Real Ashtekar variables for Lorentzian signature space
  times,'' Phys.Rev. {\bf D51} (1995) 5507--5510,
\href{http://arXiv.org/abs/gr-qc/9410014}{{\texttt{arXiv:gr-qc/9410014}}}.

\bibitem{Immirzi:1996di}
G.~Immirzi, ``Real and complex connections for canonical gravity,''
  Class.Quant.Grav. {\bf 14} (1997) L177--L181,
\href{http://arXiv.org/abs/gr-qc/9612030}{{\texttt{arXiv:gr-qc/9612030}}}.

\bibitem{Thiemann:1996ay}
T.~Thiemann, ``Anomaly - free formulation of nonperturbative, four-dimensional
  Lorentzian quantum gravity,'' Phys.Lett. {\bf B380} (1996) 257--264,
\href{http://arXiv.org/abs/gr-qc/9606088}{{\texttt{arXiv:gr-qc/9606088}}}.

\bibitem{Thiemann:1996aw}
T.~Thiemann, ``{Quantum spin dynamics (QSD)},'' Class. Quant. Grav. {\bf 15}
  (1998) 839--873,
\href{http://arXiv.org/abs/gr-qc/9606089}{{\texttt{arXiv:gr-qc/9606089}}}.

\bibitem{Thiemann:1996av}
T.~Thiemann, ``{Quantum spin dynamics (qsd). 2.},'' Class. Quant. Grav. {\bf
  15} (1998) 875--905,
\href{http://arXiv.org/abs/gr-qc/9606090}{{\texttt{arXiv:gr-qc/9606090}}}.

\bibitem{Thiemann:2003zv}
T.~Thiemann, ``The Phoenix project: Master constraint program for loop quantum
  gravity,'' Class.Quant.Grav. {\bf 23} (2006) 2211--2248,
\href{http://arXiv.org/abs/gr-qc/0305080}{{\texttt{arXiv:gr-qc/0305080}}}.

\bibitem{Alesci:2011aj}
E.~Alesci, ``{Regularized Hamiltonians and Spinfoams},'' J. Phys. Conf. Ser.
  {\bf 360} (2012) 012041,
\href{http://arXiv.org/abs/1110.6150}{{\texttt{arXiv:1110.6150}}}.

\bibitem{Alesci:2015wla}
E.~Alesci, M.~Assanioussi, J.~Lewandowski, and I.~M�kinen, ``{Hamiltonian
  operator for loop quantum gravity coupled to a scalar field},'' Phys. Rev.
  {\bf D91} (2015), no.~12, 124067,
\href{http://arXiv.org/abs/1504.02068}{{\texttt{arXiv:1504.02068}}}.

\bibitem{Assanioussi:2015gka}
M.~Assanioussi, J.~Lewandowski, and I.~M�kinen, ``{New scalar constraint
  operator for loop quantum gravity},'' Phys. Rev. {\bf D92} (2015), no.~4,
  044042,
\href{http://arXiv.org/abs/1506.00299}{{\texttt{arXiv:1506.00299}}}.

\bibitem{Bonzom:2011jv}
V.~Bonzom and A.~Laddha, ``{Lessons from toy-models for the dynamics of loop
  quantum gravity},'' SIGMA {\bf 8} (2012) 009,
\href{http://arXiv.org/abs/1110.2157}{{\texttt{arXiv:1110.2157}}}.

\bibitem{Engle:2007wy}
J.~Engle, E.~Livine, R.~Pereira, and C.~Rovelli, ``LQG vertex with finite
  Immirzi parameter,'' Nucl.Phys. {\bf B799} (2008) 136--149,
\href{http://arXiv.org/abs/0711.0146}{{\texttt{arXiv:0711.0146}}}.

\bibitem{Geloun:2010vj}
J.~Ben~Geloun, R.~Gurau, and V.~Rivasseau, ``{EPRL/FK Group Field Theory},''
  Europhys. Lett. {\bf 92} (2010) 60008,
\href{http://arXiv.org/abs/1008.0354}{{\texttt{arXiv:1008.0354}}}.

\bibitem{rovelli2014quantum}
F.~Vidotto and C.~Rovelli, {\em Covariant Loop Quantum Gravity}.
\newblock Cambridge Monographs on Mathematical Physics. Cambridge University
  Press, 2014.

\bibitem{Dupuis:2011fz}
M.~Dupuis and E.~R. Livine, ``{Holomorphic Simplicity Constraints for 4d
  Spinfoam Models},'' Class. Quant. Grav. {\bf 28} (2011) 215022,
\href{http://arXiv.org/abs/1104.3683}{{\texttt{arXiv:1104.3683}}}.

\bibitem{Speziale:2012nu}
S.~Speziale and W.~M. Wieland, ``{The twistorial structure of loop-gravity
  transition amplitudes},'' Phys. Rev. {\bf D86} (2012) 124023,
\href{http://arXiv.org/abs/1207.6348}{{\texttt{arXiv:1207.6348}}}.

\bibitem{Wieland:2013cr}
W.~M. Wieland, ``{Hamiltonian spinfoam gravity},'' Class. Quant. Grav. {\bf 31}
  (2014) 025002,
\href{http://arXiv.org/abs/1301.5859}{{\texttt{arXiv:1301.5859}}}.

\bibitem{Livine:2010zx}
E.~R. Livine, {\em {The Spinfoam Framework for Quantum Gravity}}.
\newblock PhD thesis, Lyon, IPN, 2010.
\newblock
\href{http://arXiv.org/abs/1101.5061}{{\texttt{arXiv:1101.5061}}}.
\newblock

\bibitem{Perez:2012wv}
A.~Perez, ``{The Spin Foam Approach to Quantum Gravity},'' Living Rev. Rel.
  {\bf 16} (2013) 3,
\href{http://arXiv.org/abs/1205.2019}{{\texttt{arXiv:1205.2019}}}.

\bibitem{Bianchi:2012nk}
E.~Bianchi and F.~Hellmann, ``{The Construction of Spin Foam Vertex
  Amplitudes},'' SIGMA {\bf 9} (2013) 008,
\href{http://arXiv.org/abs/1207.4596}{{\texttt{arXiv:1207.4596}}}.

\bibitem{Koslowski:2011vn}
T.~Koslowski and H.~Sahlmann, ``Loop quantum gravity vacuum with nondegenerate
  geometry,'' SIGMA {\bf 8} (2012) 026,
\href{http://arXiv.org/abs/1109.4688}{{\texttt{arXiv:1109.4688}}}.

\bibitem{Dittrich:2014wpa}
B.~Dittrich and M.~Geiller, ``{A new vacuum for Loop Quantum Gravity},'' Class.
  Quant. Grav. {\bf 32} (2015), no.~11, 112001,
\href{http://arXiv.org/abs/1401.6441}{{\texttt{arXiv:1401.6441}}}.

\bibitem{Bahr:2015bra}
B.~Bahr, B.~Dittrich, and M.~Geiller, ``{A new realization of quantum
  geometry},''
\href{http://arXiv.org/abs/1506.08571}{{\texttt{arXiv:1506.08571}}}.

\bibitem{Oriti:2013aqa}
D.~Oriti, ``{Group field theory as the 2nd quantization of Loop Quantum
  Gravity},''
\href{http://arXiv.org/abs/1310.7786}{{\texttt{arXiv:1310.7786}}}.

\bibitem{Rivasseau:2011hm}
V.~Rivasseau, ``{Quantum Gravity and Renormalization: The Tensor Track},'' AIP
  Conf. Proc. {\bf 1444} (2011) 18--29,
\href{http://arXiv.org/abs/1112.5104}{{\texttt{arXiv:1112.5104}}}.

\bibitem{Rivasseau:2013uca}
V.~Rivasseau, ``{The Tensor Track, III},'' Fortsch. Phys. {\bf 62} (2014)
  81--107,
\href{http://arXiv.org/abs/1311.1461}{{\texttt{arXiv:1311.1461}}}.

\bibitem{Carrozza:2013mna}
S.~Carrozza, {\em {Tensorial methods and renormalization in Group Field
  Theories}}.
\newblock PhD thesis, Orsay, LPT, 2013.
\newblock
\href{http://arXiv.org/abs/1310.3736}{{\texttt{arXiv:1310.3736}}}.
\newblock

\bibitem{Carrozza:2014rya}
S.~Carrozza, ``{Group field theory in dimension $4-\epsilon$},'' Phys. Rev.
  {\bf D91} (2015), no.~6, 065023,
\href{http://arXiv.org/abs/1411.5385}{{\texttt{arXiv:1411.5385}}}.

\bibitem{Livine:2006xk}
E.~R. Livine and D.~R. Terno, ``{Reconstructing quantum geometry from quantum
  information: Area renormalisation, coarse-graining and entanglement on spin
  networks},''
\href{http://arXiv.org/abs/gr-qc/0603008}{{\texttt{arXiv:gr-qc/0603008}}}.

\bibitem{Livine:2013gna}
E.~R. Livine, ``Deformation Operators of Spin Networks and Coarse-Graining,''
\href{http://arXiv.org/abs/1310.3362}{{\texttt{arXiv:1310.3362}}}.

\bibitem{Ashtekar:1994mh}
A.~Ashtekar and J.~Lewandowski, ``Projective techniques and functional
  integration for gauge theories,'' J.Math.Phys. {\bf 36} (1995) 2170--2191,
\href{http://arXiv.org/abs/gr-qc/9411046}{{\texttt{arXiv:gr-qc/9411046}}}.

\bibitem{Ashtekar:1994wa}
A.~Ashtekar and J.~Lewandowski, ``{Differential geometry on the space of
  connections via graphs and projective limits},'' J. Geom. Phys. {\bf 17}
  (1995) 191--230,
\href{http://arXiv.org/abs/hep-th/9412073}{{\texttt{arXiv:hep-th/9412073}}}.

\bibitem{Ashtekar:1993wf}
A.~Ashtekar and J.~Lewandowski, ``{Representation theory of analytic holonomy
  C* algebras},''
\href{http://arXiv.org/abs/gr-qc/9311010}{{\texttt{arXiv:gr-qc/9311010}}}.

\bibitem{Rovelli:1994ge}
C.~Rovelli and L.~Smolin, ``Discreteness of area and volume in quantum
  gravity,'' Nucl.Phys. {\bf B442} (1995) 593--622,
\href{http://arXiv.org/abs/gr-qc/9411005}{{\texttt{arXiv:gr-qc/9411005}}}.

\bibitem{Rovelli:1995ac}
C.~Rovelli and L.~Smolin, ``Spin networks and quantum gravity,'' Phys.Rev. {\bf
  D52} (1995) 5743--5759,
\href{http://arXiv.org/abs/gr-qc/9505006}{{\texttt{arXiv:gr-qc/9505006}}}.

\bibitem{Freidel:2010aq}
L.~Freidel and S.~Speziale, ``Twisted geometries: A geometric parametrisation
  of SU(2) phase space,'' Phys.Rev. {\bf D82} (2010) 084040,
\href{http://arXiv.org/abs/1001.2748}{{\texttt{arXiv:1001.2748}}}.

\bibitem{Dupuis:2012yw}
M.~Dupuis, J.~P. Ryan, and S.~Speziale, ``{Discrete gravity models and Loop
  Quantum Gravity: a short review},'' SIGMA {\bf 8} (2012) 052,
\href{http://arXiv.org/abs/1204.5394}{{\texttt{arXiv:1204.5394}}}.

\bibitem{Knizhnik:1988ak}
V.~Knizhnik, A.~M. Polyakov, and A.~Zamolodchikov, ``Fractal Structure of 2D
  Quantum Gravity,'' Mod.Phys.Lett. {\bf A3} (1988)
819.

\bibitem{Freidel:2002xb}
L.~Freidel and E.~R. Livine, ``Spin networks for noncompact groups,''
  J.Math.Phys. {\bf 44} (2003) 1322--1356,
\href{http://arXiv.org/abs/hep-th/0205268}{{\texttt{arXiv:hep-th/0205268}}}.

\bibitem{Holst:1995pc}
S.~Holst, ``{Barbero's Hamiltonian derived from a generalized Hilbert-Palatini
  action},'' Phys. Rev. {\bf D53} (1996) 5966--5969,
\href{http://arXiv.org/abs/gr-qc/9511026}{{\texttt{arXiv:gr-qc/9511026}}}.

\bibitem{Samuel:2000ue}
J.~Samuel, ``Is Barbero's Hamiltonian formulation a gauge theory of Lorentzian
  gravity?,'' Class. Quant. Grav. {\bf 17} (2000) L141--L148,
\href{http://arXiv.org/abs/gr-qc/0005095}{{\texttt{arXiv:gr-qc/0005095}}}.

\bibitem{Alexandrov:2001wt}
S.~Alexandrov, ``{On choice of connection in loop quantum gravity},'' Phys.
  Rev. {\bf D65} (2002) 024011,
\href{http://arXiv.org/abs/gr-qc/0107071}{{\texttt{arXiv:gr-qc/0107071}}}.

\bibitem{Geiller:2011cv}
M.~Geiller, M.~Lachieze-Rey, K.~Noui, and F.~Sardelli, ``{A Lorentz-Covariant
  Connection for Canonical Gravity},'' SIGMA {\bf 7} (2011) 083,
\href{http://arXiv.org/abs/1103.4057}{{\texttt{arXiv:1103.4057}}}.

\bibitem{Geiller:2011bh}
M.~Geiller, M.~Lachieze-Rey, and K.~Noui, ``{A new look at Lorentz-Covariant
  Loop Quantum Gravity},'' Phys. Rev. {\bf D84} (2011) 044002,
\href{http://arXiv.org/abs/1105.4194}{{\texttt{arXiv:1105.4194}}}.

\bibitem{Charles:2015rda}
C.~Charles and E.~R. Livine, ``Ashtekar-Barbero holonomy on the hyperboloid:
  Immirzi parameter as a Cut-off for Quantum Gravity,''
\href{http://arXiv.org/abs/1507.00851}{{\texttt{arXiv:1507.00851}}}.

\bibitem{Freidel:2013bfa}
L.~Freidel and J.~Ziprick, ``Spinning geometry = Twisted geometry,''
  Class.Quant.Grav. {\bf 31} (2014), no.~4, 045007,
\href{http://arXiv.org/abs/1308.0040}{{\texttt{arXiv:1308.0040}}}.

\bibitem{Dupuis:2013haa}
M.~Dupuis and F.~Girelli, ``{Quantum hyperbolic geometry in loop quantum
  gravity with cosmological constant},'' Phys. Rev. {\bf D87} (2013), no.~12,
  121502,
\href{http://arXiv.org/abs/1307.5461}{{\texttt{arXiv:1307.5461}}}.

\bibitem{Bonzom:2014wva}
V.~Bonzom, M.~Dupuis, F.~Girelli, and E.~R. Livine, ``{Deformed phase space for
  3d loop gravity and hyperbolic discrete geometries},''
\href{http://arXiv.org/abs/1402.2323}{{\texttt{arXiv:1402.2323}}}.

\bibitem{Dupuis:2014fya}
M.~Dupuis, F.~Girelli, and E.~R. Livine, ``{Deformed Spinor Networks for Loop
  Gravity: Towards Hyperbolic Twisted Geometries},'' Gen. Rel. Grav. {\bf 46}
  (2014), no.~11, 1802,
\href{http://arXiv.org/abs/1403.7482}{{\texttt{arXiv:1403.7482}}}.

\bibitem{Charles:2015lva}
C.~Charles and E.~R. Livine, ``{Closure constraints for hyperbolic
  tetrahedra},'' Class. Quant. Grav. {\bf 32} (2015), no.~13, 135003,
\href{http://arXiv.org/abs/1501.00855}{{\texttt{arXiv:1501.00855}}}.

\bibitem{Haggard:2014xoa}
H.~M. Haggard, M.~Han, W.~Kami?ski, and A.~Riello, ``{SL(2,C) Chern-Simons
  Theory, a non-Planar Graph Operator, and 4D Loop Quantum Gravity with a
  Cosmological Constant: Semiclassical Geometry},'' Nucl. Phys. {\bf B900}
  (2015) 1--79,
\href{http://arXiv.org/abs/1412.7546}{{\texttt{arXiv:1412.7546}}}.

\bibitem{Haggard:2015ima}
H.~M. Haggard, M.~Han, and A.~Riello, ``{Encoding Curved Tetrahedra in Face
  Holonomies: a Phase Space of Shapes from Group-Valued Moment Maps},''
\href{http://arXiv.org/abs/1506.03053}{{\texttt{arXiv:1506.03053}}}.

\bibitem{Haggard:2015yda}
H.~M. Haggard, M.~Han, W.~Kami?ski, and A.~Riello, ``{Four-dimensional Quantum
  Gravity with a Cosmological Constant from Three-dimensional Holomorphic
  Blocks},'' Phys. Lett. {\bf B752} (2016) 258--262,
\href{http://arXiv.org/abs/1509.00458}{{\texttt{arXiv:1509.00458}}}.

\bibitem{Pithis:2014uva}
A.~G. Pithis and H.-C. Ruiz~Euler, ``Anyonic statistics and large horizon
  diffeomorphisms for Loop Quantum Gravity Black Holes,'' Phys.Rev. {\bf D91}
  (2015) 064053,
\href{http://arXiv.org/abs/1402.2274}{{\texttt{arXiv:1402.2274}}}.

\bibitem{Yang:2008th}
J.~Yang and Y.~Ma, ``{Quasi-Local Energy in Loop Quantum Gravity},'' Phys. Rev.
  {\bf D80} (2009) 084027,
\href{http://arXiv.org/abs/0812.3554}{{\texttt{arXiv:0812.3554}}}.

\bibitem{Korotkin:1997ps}
D.~Korotkin and H.~Samtleben, ``Canonical quantization of cylindrical
  gravitational waves with two polarizations,'' Phys.Rev.Lett. {\bf 80} (1998)
  14--17,
\href{http://arXiv.org/abs/gr-qc/9705013}{{\texttt{arXiv:gr-qc/9705013}}}.

\bibitem{Ashtekar:1996cm}
A.~Ashtekar, J.~Bicak, and B.~G. Schmidt, ``Behavior of Einstein-Rosen waves at
  null infinity,'' Phys.Rev. {\bf D55} (1997) 687--694,
\href{http://arXiv.org/abs/gr-qc/9608041}{{\texttt{arXiv:gr-qc/9608041}}}.

\bibitem{Borja:2010rc}
E.~F. Borja, L.~Freidel, I.~Garay, and E.~R. Livine, ``U(N) tools for Loop
  Quantum Gravity: The Return of the Spinor,'' Class.Quant.Grav. {\bf 28}
  (2011) 055005,
\href{http://arXiv.org/abs/1010.5451}{{\texttt{arXiv:1010.5451}}}.

\bibitem{Bonzom:2009zd}
V.~Bonzom, E.~R. Livine, and S.~Speziale, ``{Recurrence relations for spin foam
  vertices},'' Class. Quant. Grav. {\bf 27} (2010) 125002,
\href{http://arXiv.org/abs/0911.2204}{{\texttt{arXiv:0911.2204}}}.

\bibitem{Bonzom:2011hm}
V.~Bonzom and L.~Freidel, ``{The Hamiltonian constraint in 3d Riemannian loop
  quantum gravity},'' Class. Quant. Grav. {\bf 28} (2011) 195006,
\href{http://arXiv.org/abs/1101.3524}{{\texttt{arXiv:1101.3524}}}.

\bibitem{Bonzom:2014bua}
V.~Bonzom, M.~Dupuis, and F.~Girelli, ``{Towards the Turaev-Viro amplitudes
  from a Hamiltonian constraint},'' Phys. Rev. {\bf D90} (2014), no.~10,
  104038,
\href{http://arXiv.org/abs/1403.7121}{{\texttt{arXiv:1403.7121}}}.

\bibitem{Freidel:1998ua}
L.~Freidel and K.~Krasnov, ``{Discrete space-time volume for three-dimensional
  BF theory and quantum gravity},'' Class. Quant. Grav. {\bf 16} (1999)
  351--362,
\href{http://arXiv.org/abs/hep-th/9804185}{{\texttt{arXiv:hep-th/9804185}}}.

\bibitem{Bonzom:2011nv}
V.~Bonzom and E.~R. Livine, ``{A New Hamiltonian for the Topological BF phase
  with spinor networks},'' J. Math. Phys. {\bf 53} (2012) 072201,
\href{http://arXiv.org/abs/1110.3272}{{\texttt{arXiv:1110.3272}}}.

\bibitem{Bonzom:2013tna}
V.~Bonzom and B.~Dittrich, ``{Dirac's discrete hypersurface deformation
  algebras},'' Class. Quant. Grav. {\bf 30} (2013) 205013,
\href{http://arXiv.org/abs/1304.5983}{{\texttt{arXiv:1304.5983}}}.

\bibitem{Feller:2015yta}
A.~Feller and E.~R. Livine, ``{Ising Spin Network States for Loop Quantum
  Gravity: a Toy Model for Phase Transitions},'' Class. Quant. Grav. {\bf 33}
  (2016), no.~6, 065005,
\href{http://arXiv.org/abs/1509.05297}{{\texttt{arXiv:1509.05297}}}.

\bibitem{Alesci:2013xd}
E.~Alesci and F.~Cianfrani, ``{Quantum-Reduced Loop Gravity: Cosmology},''
  Phys. Rev. {\bf D87} (2013), no.~8, 083521,
\href{http://arXiv.org/abs/1301.2245}{{\texttt{arXiv:1301.2245}}}.

\bibitem{Livine:2002ak}
E.~R. Livine, ``{Projected spin networks for Lorentz connection: Linking spin
  foams and loop gravity},'' Class. Quant. Grav. {\bf 19} (2002) 5525--5542,
\href{http://arXiv.org/abs/gr-qc/0207084}{{\texttt{arXiv:gr-qc/0207084}}}.

\bibitem{Dupuis:2010jn}
M.~Dupuis and E.~R. Livine, ``{Lifting SU(2) Spin Networks to Projected Spin
  Networks},'' Phys. Rev. {\bf D82} (2010) 064044,
\href{http://arXiv.org/abs/1008.4093}{{\texttt{arXiv:1008.4093}}}.

\end{thebibliography}\endgroup

\end{document}